\documentclass[10pt,reqno,final]{amsart}
\usepackage{epsfig,amssymb,amsmath,version}
\usepackage{amssymb,version,graphicx,fancybox,mathrsfs,multirow,stmaryrd}

\usepackage{url,hyperref}
\usepackage[notcite,notref]{showkeys}

\usepackage{subfigure}
\usepackage{color,xcolor}
\usepackage{cases}
\usepackage{mathtools}

\catcode`\@=11 \theoremstyle{plain}
\@addtoreset{equation}{section}
\renewcommand{\theequation}{\arabic{section}.\arabic{equation}}
\@addtoreset{figure}{section}
\renewcommand\thefigure{\thesection.\@arabic\c@figure}

\theoremstyle{proposition}
\newtheorem{prop}{\bf Proposition}[section]
\theoremstyle{remark}
\newtheorem{rem}{\bf Remark}[section]

\newcommand{\bs}[1]{\boldsymbol{#1}}
\def \ri {{\rm i}}

\def \Sum {\sum_{|m|=0}^{\infty}}
\def \Hankel {H_m ^{(1)}}

\begin{document}


\bibliographystyle{elsarticle-num}

{\title[Integration of DtN and SEM] {Seamless Integration of global Dirichlet-to-Neumann  boundary condition  and  spectral elements for transformation electromagnetics}
\author[Z. Yang,\;  L. Wang, \;  Z. Rong, \;   B. Wang \;  $\&$\; B. Zhang] {Zhiguo Yang$^{1}$,\;  Li-Lian Wang$^{1},$\;  Zhijian Rong$^{2}$,\;       Bo Wang$^{3}$\;  and\;  Baile Zhang$^{4}$}
\thanks{\noindent $^{1}$Division of Mathematical Sciences, School of Physical
and Mathematical Sciences,  Nanyang Technological University,
637371, Singapore. The research of the first two  authors is  supported by Singapore  MOE AcRF Tier 2 Grant (MOE 2013-T2-1-095, ARC 44/13),  Singapore A$^\ast$STAR-SERC-PSF Grant (122-PSF-007) and  Singapore MOE AcRF Tier 1 Grant (RG 15/12).\\
\indent$^{2}$School of Mathematical Sciences, Xiamen University, Xiamen 361005, China. The research of this author is supported by 
NSF of Fujian Province of China under Grant No. 2013J05019 and NSFC under Grant No. 11201393. \\
\indent$^{3}$College of Mathematics and Computer Science, Hunan Normal University, 410081, China. The research of this author is supported by NSFC under Grants No. 11341002 and No. 11401206. \\
\indent$^{4}$Division of Physics and Applied Physics, School of Physical and Mathematical Sciences, Nanyang Technological University, 637371, Singapore. The research of this author is partially supported by Nanyang Technological University under start-up grants, and by the Singapore Ministry of Education under Grants No. Tier 1 RG27/12 and No. MOE2011-T3-1-005. \\
\indent The third and forth authors would like to thank the hospitality of the Division of Mathematical Sciences, School of Physical
and Mathematical Sciences,  Nanyang Technological University in Singapore, for hosting their visit. 
}
\keywords{Dirichlet-to-Neumann (DtN) boundary condition, Helmholtz equation in anisotropic media,  invisibility cloaks, singular coordinate transformations, cloaking  boundary conditions, spectral-element method}
 \subjclass[2000]{65Z05, 74J20, 78A40, 33E10,  35J05, 65M70,  65N35}

\begin{abstract}
In this paper, we present an efficient   spectral-element method (SEM) for solving  general two-dimensional Helmholtz equations in anisotropic media,  with particular applications in accurate  simulation of polygonal invisibility cloaks, concentrators and circular rotators arisen from  the field of transformation electromagnetics (TE).  In practice, we adopt  
a transparent boundary condition (TBC)  characterized by  the 
Dirichlet-to-Neumann (DtN) map to reduce  wave propagation in an unbounded domain to a bounded domain. We then  introduce  a  semi-analytic technique  to  integrate  the global  TBC  with  local curvilinear elements seamlessly, which is accomplished by  
using a novel elemental mapping and  analytic formulas for  evaluating  global Fourier coefficients on spectral-element grids exactly.  

From the perspective of TE,  an invisibility cloak is devised by a singular coordinate transformation of Maxwell's equations  that  leads to anisotropic materials coating the cloaked region to render any object inside invisible to observers outside. An important issue 
 resides in the imposition of appropriate conditions at the outer boundary of the cloaked region, i.e., cloaking boundary conditions (CBCs), in order to achieve  perfect invisibility.  Following the spirit of \cite{yang2014accurate}, we propose new CBCs for polygonal invisibility cloaks from 
 the essential ``pole" conditions related to  singular transformations.  This  allows for the decoupling of the governing equations of  inside and outside the cloaked regions.  With this efficient spectral-element solver at our disposal, we can study the interesting phenomena when   some defects and lossy or dispersive media are placed in the cloaking layer of an ideal polygonal cloak. 
\end{abstract}
\maketitle

\vspace*{-10pt}
\section{Introduction and problem statement}
Accurate simulation of   wave propagations in inhomogeneous and anisotropic media  plays an exceedingly important part in a wide range of applications related to the
exploration and  design of   novel materials that  enjoy unusual and remarkable  properties in steering waves.
%
In many situations involving  time-harmonic wave propagations, the development of high-order methods (i.e., spectral and spectral-element solvers)  for the  Helmholtz equation and time-harmonic Maxwell equations,    is of fundamental  importance.
%

We are concerned with   the two-dimensional Helmholtz equation governing  time-harmonic wave propagation in  anisotropic media:    
\begin{equation} \label{Helm1}
\nabla \cdot \big({\bs C}(\bs r )\,\nabla u(\bs r)\big)+k^2 n(\bs r) u(\bs r)=f(\bs r),\quad\bs r=\bs x=(x,y) \in {\mathbb R^2},
\end{equation}
where  $k>0$ is the wave number in free space. 
 In general,  we make the following assumptions.  
 \begin{itemize}
 \item[(i)] $\bs C$ is a symmetric positive definite  matrix in ${\mathbb R^{2\times 2}},$  and   for some positive constants $c_0,c_1,$
\begin{equation}\label{eqnA}
0<c_0\le \bs \xi^t\, \bs C\, \bs \xi\le c_1,\;\;\; \forall\, \bs \xi\in {\mathbb R}^2, \;\;   \text{a.e. in} \;\;  {\mathbb R^2}.
\end{equation} 
\item[(ii)]  The coefficient 
\begin{equation}\label{eqnnA}
0<n \le n_1,\;\;\;\text{a.e. in} \;\; {\mathbb R^2}.
\end{equation}
\item[(iii)]  The inhomogeneity of the medium is confined in a bounded domain $\Omega_-$ with Lipschitz boundary,  and  $f$ is compactly supported in disc $B_{\!R}$ of radius $R>0$ (see Figure \ref{sketch1}): 
\begin{equation}\label{Parameter2}
{\bs C}={\bs I}_2,\;\;\; n=1\;\;\;  {\rm in}\;\;\; \mathbb{R}^2 \setminus \bar \Omega_{-}; \quad {\rm supp}(f)\subseteq B_{\!R},
\end{equation}
where $\bs I_2$ is the $2\times 2$ identity matrix. In what follows,  we are interested in the case where   $\Omega_-$  is  a penetrable scatterer.
 \end{itemize}  
  
\vspace*{-10pt}
\begin{figure}[h!]
 \begin{minipage}[b]{0.45\textwidth}
    \centering
 \includegraphics[width=0.8\textwidth]{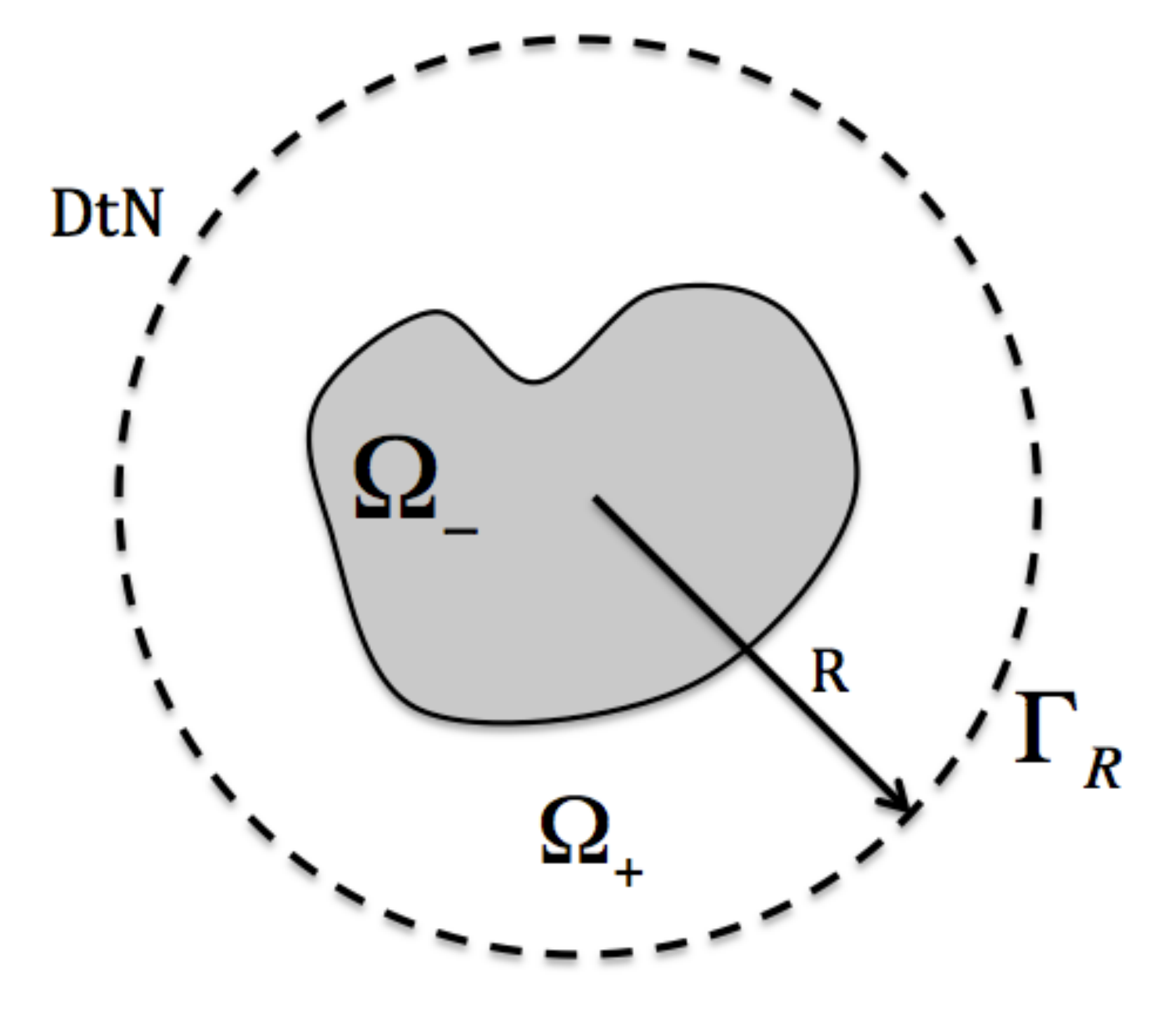}
    \caption{\small Illustration of  geometry}
 \label{sketch1}
  \end{minipage}
 \begin{minipage}[b]{0.54\textwidth}
\quad  We  impose the well-known Sommerfeld radiation boundary condition upon  the scattering wave:
$u_{\rm sc}:=u-u_{\rm in}$ (where  $u_{\rm in}$ is a  given  incident wave):
\begin{equation}\label{Radiation1}
\partial_{r} u_{\rm sc}-\ri k\, u_{\rm sc}=o(r^{-{1}/{2}}) \;\;\;  {\rm as}\;\;r \rightarrow \infty,
\end{equation}
where $\ri=\sqrt{-1}$ is the complex unit. 
\vskip 5pt 
\quad   The challenges of the above problem are  at least threefold: (i) unboundedness of the  computational domain; 
(ii) indefiniteness of the variational formulation;  and (iii) highly oscillatory solution  decaying slowly  when $k\gg 1.$
In addition,  the coefficients ${\bs C}(\bs r)$ and $n(\bs r)$ might be singular at some interior interface in $\Omega_-$ (see Section \ref{sect:pcloak}).  
\end{minipage}
\end{figure}

%

\vspace*{-5pt}
The methods of choice to deal with  the first issue  typically include the perfectly matched layer (PML) technique \cite{Bere94}, boundary integral method \cite{jin1991finiteBI, liu2009hybrid},   and the artificial  boundary condition \cite{ Hagstrom99, Enguist77, Gro.K95, Nede01}.
 The latter is  known as the absorbing boundary condition  (ABC), if it   leads  to a well-posed initial-boundary value problem (IBVP)  and  some ``energy'' can be absorbed  at the boundary.  In particular,  if the solution of the reduced problem coincides with  that of the original problem,  then the related   ABC is dubbed as a transparent (or nonreflecting) boundary condition (TBC) (or NRBC). 
 In this paper, we adopt the exact TBC (see  \cite{Nede01} and Figure \ref{sketch1}):   
\begin{equation}
\partial_{r} u_{\rm sc}-\mathscr {T}_R [u_{\rm sc}]=0 \quad {\rm at}\;\; \Gamma_{\!R},\label{DtN}
\end{equation}
where the DtN map $\mathscr{T}_R$ is defined as
\begin{align}
& {\mathscr T}_{R} [\psi] = \Sum  \mathcal{T}_m \hat \psi_m  e^{\ri m\theta}, \quad {\rm with} \label{DtNoperator}\\
&\hat \psi_m=\frac 1 {2\pi}\int_0^{2\pi} \psi(R,\theta) e^{-\ri  m\theta} {\rm d}\theta,\quad   {\mathcal T}_m:=\frac{k{\Hankel} '(kR)}{\Hankel (k R)}. \label{kernelK}
\end{align}
Here,     $\Hankel (z)$ is  the Hankel function of the first kind (cf. \cite{Abr.S84}).  This yields the exact boundary condition for  the total field:
\begin{equation}\label{uDtN}
\partial_{r} u-{\mathscr T} _{R} [u]= \partial_{r} u_{\rm in}-{\mathscr T} _{R}[u_{\rm in}]:=h   \quad {\rm at}\;\;  \Gamma_{\! R}.
\end{equation}

We find it is advantageous to impose  DtN  TBC  for the following reasons.
\begin{itemize}
 \item[(i)] The original problem in ${\mathbb R}^2$ reduces to an equivalent boundary value problem (BVP) in $B_{\!R}.$ 
One can place  $\Gamma_{\!R}$ as close as possible to  $\Omega_-,$ as long as the inhomogeneity  of the media and support of the source term are confined in $B_{\!R}.$ 
\item[(ii)]  It is essential for accurate and stable simulations  especially when the  wavenumber is  large.   
 \end{itemize}
 
However, the  TBC \eqref{uDtN} is global in space, that is, evaluating ${\mathscr T}_R[u](x)$ at any point $x_0\in \Gamma_{\!R}$ requires to compute  the  global Fourier integral along the circle $\Gamma_{\!R}.$  This poses challenges in  solving the reduced Helmholtz problem by a local element-based  method. 
\begin{itemize}
\item[(i)] The spectral-element  solution (using curvilinear elements along $\Gamma_{\!R}$)  is piecewise continuous (i.e., 
only in  $C^0 (\Gamma_{\!R})$), and  defined on spectral-element grids, so the interplay between Fourier points and spectral-element grids via interpolation and fast Fourier transform (FFT) only leads to  a first-order convergence.  
\item[(ii)]  One can evaluate the Fourier integral by a composite rule with a decomposition  coherent to the spectral-element partition,  but due to  the  elemental mapping between the curvilinear element and reference square, a numerical quadrature is usually necessary,  which is prohibitive as the integrands  are highly oscillatory for high Fourier modes  (see \eqref{mainint3}). 
\end{itemize}

One main purpose of this paper is to seamlessly integrate the global DtN BC with  local spectral elements. 
The key idea is to construct  a new elemental mapping between the curvilinear elements along $\Gamma_{\!R}$ and the reference square (see Figure \ref{domaingraph}), which leads to exact evaluation of  the Fourier integrals. 
 It is noteworthy that Fournier \cite{fournier2005exact} proposed a method
  for calculating  global Fourier coefficients
 for given nodal values on  non-conforming spectral elements, where a similar semi-analytic approach was essential for the success of the method therein.  
 We  also remark that the recent work \cite{He2014DtN}  addressed the integration of one-dimensional DtN TBC  imposed on a line segment with standard rectangular elements along the boundary.
  Different from these works, our ``local-to-global" method is built upon 
 the use of curvilinear elements seamlessly fitting the circular boundary,  and the design of a new elemental mapping leading to  exact calculation of the involved Fourier integrals (see Subsection \ref{newsectA}).   
 
 


Underpinned by the advent of metamaterials, transformation electromagnetics (TE)  (cf.   \cite{pendry.2006,leonhardt2006optical})  provides a powerful tool for creating novel devices and new materials  with  unconventional properties 
 (see, e.g., \cite{yan2008superlense,rahm2008design,chen2007rotator,chen2008anti,yang2008superscatterer,rahm2008beam} and   \cite{werner2013transformation} for many original references therein). 
 Some exciting applications of TE include  the invisibility cloaks  (see, e.g., \cite{pendry.2006, Greenleaf2009siamreview}), rotators (see, e.g., \cite{chen2007rotator}) and concentrators (see, e.g., \cite{rahm2008design}), which   naturally give rise to the model  problem \eqref{Helm1}-\eqref{Radiation1}.
In particular,  the   invisibility cloak is  devised by a singular coordinate transformation \cite{pendry.2006}  that leads to singular materials  coating  the cloaked regions and preventing waves from penetrating into the inside region.  
The imposition of appropriate interface conditions at the inner boundary, i.e., CBCs, where the material parameters are singular, 
becomes critical.  Significant efforts have been devoted to  CBCs for circular cylindrical and spherical cloaks. Ruan et al. \cite{ruan2007ideal}  first analytically studied the sensitivity of the ideal circular cloak  \cite{pendry.2006} to a small $\delta$-perturbation of the inner  boundary. 
Zhang et al. \cite{zhang2007response} provided    deep  insights into the physical effects of the  singular transformation  (also see \cite{zhang2012electrodynamics}). To shield the incoming  waves,  the perfect magnetic conductor  {\rm(PMC)}  condition
 was  imposed at the inner boundary in finite-element simulations (see,  e.g.,  \cite{cummer2006full,Jichun.2012,ma2008material}).  Weder \cite{Weder08}  proposed CBCs for the ideal  spherical cloak  of Pendry  et al.  \cite{pendry.2006} from the perspective of energy conservation.
 Lassas  and Zhou \cite{LaZhou11,LaZhou14}  proposed some non-local pesudo-differential CBCs. 
Based upon the principle that  a well-behaved  electromagnetic field in the original space must be  well-behaved   in the transformed space as well,  Yang and Wang \cite{yang2014accurate} obtained  CBCs for circular and elliptical cloaks that intrinsically relate to the essential ``pole'' conditions  of a singular transformation. 

The polygonal cloaks enjoy  more  flexibility to hide objects with  complex shapes, which are  however  much less studied. Indeed, many of the previous principles and approaches for CBCs are not extendable to the polygonal case. 
    Following the spirit  of \cite{yang2014accurate}, we propose  new CBCs under a ``local'' coordinate system (see Proposition 
    \ref{essentialpole}), under which the governing equation in the cloaked region is decoupled from the exterior region.
Accordingly, no wave can propagate into the cloaked region, and vice versa. We  emphasise that the new CBCs are indispensable for spectrally accurate simulations.  We also show that the proposed spectral-element solver provides a reliable tool  to study how the defects affect the perfectness of an ideal cloak (see Subsection \ref{newsectBA}).

The rest of the paper is organised as follows. In Section \ref{sect3:sem}, we review the form invariant of Maxwell equations and 
illustrate the derivation of the above model problem. We then  introduce the technique to seamlessly integrate the global DtN BC with local spectral elements. 
Section \ref{sect:pcloak} is for accurate simulation of polygonal invisibility cloak, where new CBCs are derived and efficient techniques are introduced to deal with singular material parameters. Various numerical results are provided to show the perfectness of invisibility, and the effects of defects and lossy or dispersive media.   Section  \ref{sect:App}  concerns   the extension of the spectral-element solver to the simulation of electromagnetic concentrators and rotators.


\section{TE and spectral-element discretization of DtN BC}\label{sect3:sem}  

In this section, we first illustrate the scenarios of  the aforementioned  model Helmholtz problem  arisen from transformation electromagnetics.  
We then discretise the model  problem  by  spectral-element method and 
 focus on how to seamlessly integrate the global DtN BC with  local  elements.

\subsection{Form invariant of Maxwell  equations}\label{subsect:TE}   
%
%

Consider the time-harmonic Maxwell system:
\begin{equation}\label{Maxwell}\breve
\nabla_{\!\breve{\bs r}}\times\breve {\bs E}-{\rm i}\omega \mu_0\, \breve {\bs H}=\bs 0, \quad \nabla_{\!\breve{\bs r}} \times \breve {\bs H}+{\rm i}\omega \epsilon_0\,\breve {\bs E}=\bs 0, 
\end{equation}
in  Cartesian coordinates: $ \breve {\bs r}=\breve {\bs x}=(\breve x,\breve y,\breve z)\in {\mathbb R}^3,$ where the electric permittivity $\epsilon_0,$ the magnetic permeability $\mu_0,$   and the angular frequency $\omega$  are positive constants.  Note that    $e^{-\ri \omega t}$ time-dependence is assumed for the electric and magnetic fields.

A remarkable property of the Maxwell system is   its form invariant under any  coordinate transformation (cf.  \cite{post1997formal}). More precisely, given a coordinate transformation $\bs r=\bs r(\breve{\bs r})$ with the Jacobian matrix  $\bs J=\partial {\bs r}/\partial{\breve{\bs r}},$
 the transformed  Maxwell system  takes the same form:  
\begin{equation} \label{newMaxwell}
\nabla  \times {\bs E}-{\rm i}\omega \mu_0 {\bs \mu}\, {\bs H}=\bs 0, \quad \nabla  \times {\bs H}+{\rm i}\omega \epsilon_0 {\bs \epsilon}\, {\bs E}=\bs 0,
\end{equation}
where  $\nabla \times$ is the curl operator in the new coordinates, and
 \begin{equation}\label{newfields}
\bs E(\bs r)=(\bs J^t)^{-1}\breve{\bs E}(\breve{\bs r}),\quad \bs H(\bs r)=(\bs J^t)^{-1}\breve{\bs H}(\breve{\bs r}),\quad {\bs \mu}={\bs \epsilon}=  {\bs J \bs J^t}\big/{{\rm det}({\bs J})}.
\end{equation}
%

We are  concerned with the two-dimensional electromagnetic wave propagations in media with 
in-plane anisotropy.      
Accordingly, under the transverse-electric (TE) polarization,  we consider ${\bs E}=(0,0,u(x,y))^t$ and  $\bs H=(H_1, H_2, 0)^t$.  Letting $z=\breve z$ in the coordinate transformation, 
the material parameters in  \eqref{newfields}  reduce  to 
\begin{equation} \label{parameter2}
{\bs \mu}={\bs \epsilon}=
\begin{bmatrix}
{\bs C} & \bs 0^t\\[2pt]
               \bs 0      & n
\end{bmatrix}
=\begin{bmatrix}
C_{11} & C_{12}& 0\\[1pt]
C_{12} & C_{22}&0\\[1pt]
0& 0 & n
\end{bmatrix}, 
\end{equation}
where
\begin{equation}\label{parameter3}
{\bs C}=\frac{\bs J_{\! {\rm cn}}\,\bs J^t_{\!{\rm cn}}} {{\rm det}(\bs J_{\! {\rm cn}})},\quad n=\frac 1{{\rm det}(\bs J_{\! {\rm cn}})}\;\;\; {\rm with}\;\;\; \bs J_{\! {\rm cn}}:=
\begin{bmatrix}
\partial_{\breve x} x & \partial_{\breve y} x  \\[1pt]
\partial_{\breve x} y & \partial_{\breve y} y  \\[1pt]
\end{bmatrix}.
\end{equation}
Note that   ${\rm det}(\bs C)=1,$ and 
  \begin{equation} \label{parameter2s}
{\bs \mu}^{-1}={\bs \epsilon}^{-1}
=\begin{bmatrix}
C_{22} & -C_{12}& 0\\[1pt]
-C_{12} & C_{11}&0\\[1pt]
0& 0 & n^{-1}
\end{bmatrix}.
\end{equation}
%
Then we  derive   from  the first equation of \eqref{newMaxwell} and \eqref{parameter2s} that 
\begin{equation}\label{magnetic}
\begin{split}
{\bs H}&=\frac{{\bs \mu}^{-1}}{{\rm i} \omega \mu_0}  \nabla  \times {\bs E}=\frac{{\bs \mu}^{-1}}{{\rm i} \omega \mu_0}\big(u_y,-u_x,0\big)^t=\frac{1}{{\rm i} \omega \mu_0}\big(C_{12} u_x+C_{22} u_y, -C_{11}u_x-C_{12}u_y,0\big)^t. 
\end{split}
\end{equation}
Inserting it into the second equation of \eqref{newMaxwell},  we obtain the two-dimensional Helmholtz equation:
\begin{equation} \label{HelmMeta}
\nabla \cdot ({\bs C}(\bs r)\;\nabla u(\bs r))+k ^2 n(\bs r)\, u(\bs r)=0,
\end{equation}
where $k=\omega \sqrt{\epsilon_0 \mu_0}$ is  the wavenumber in free space.

%
%
%

  In Sections \ref{sect:pcloak}-\ref{sect:App}, we shall introduce the coordinate transformations  for   polygonal 
  invisibility cloaks,  concentrators and rotators, and compute the corresponding material parameters $\bs C$ and $n$ via \eqref{parameter3}.  It is  noteworthy  that in all cases, the coordinate transformations are  identity in ${\mathbb R}^2\setminus(\Omega_-\cup\Omega_+)$ (cf. Figure \ref{sketch1}),  so \eqref{Parameter2} can be met.
  Moreover,  we can derive \eqref{Usuperscript} below  from  the standard transmission conditions (see,  e.g., \cite[Sec. 1.5]{orfanidis2002electromagnetic} and \cite{Monk03}), that is,
the continuity of the  tangential components of  $\bs E$ and $\bs H$  at the interface $\Gamma:=\partial \Omega_-.$ 
  


In summary, the problem of interest reads  
\begin{align}
& 
\nabla \cdot({\bs C}(\bs r) \nabla u(\bs r))+ k^2n(\bs r) u(\bs r)=f(\bs r) \quad  {\rm in}\;\; B_R, \label{eq1}   \\
&\llbracket u \rrbracket=\llbracket \bs C\, \nabla u \rrbracket=0\; \quad {\rm at}\;\;  \Gamma,\label{eq11}\\
&\partial_{r} u-{\mathscr T} _{R} [u]=h   \qquad {\rm at}\;\;  \Gamma_{\! R},  \label{eq2}
\end{align}
where  
\begin{equation}\label{Usuperscript}
\llbracket u \rrbracket:=u^--u^+,\quad  \llbracket \bs C\, \nabla u \rrbracket:={\bs n}\cdot ({\bs C}^- \nabla u^- -\bs C^+\nabla u^+), 
\end{equation}
 $u^\pm:=u|_{\Omega_\pm},$  $\bs C^{\pm}:=\bs C |_{\Omega_\pm}$ and $\bs n$ is the unit outer normal vector along $\Gamma.$

%
%
%
%
%


\subsection{Spectral-element scheme}\label{subsect:A} 
 Let $\Omega$  be a generic bounded domain, and $L^2(\Omega)$ be the space of square integrable functions with the inner product and norm denoted by $(\cdot,\cdot)_\Omega$ and $\|\cdot\|_\Omega$ as usual. 
 The Sobolev space $H^m(\Omega)$ with $m>0$ is defined as  in Admas \cite{Adams03} with the normal $\|u\|_{m,\Omega}.$ Define the trace integral 
 \begin{equation}\label{LineInner}
  \langle u,v \rangle_{\Gamma_{\!R}}:=\oint_{\Gamma_{\!R}} u \, \bar v\, {\rm d}\gamma.
  \end{equation}

A weak formulation of  \eqref{eq1}-\eqref{eq2} is to find $u\in H^1(B_{\!R})$  such that
\begin{equation} \label{weakform1}
\begin{split}
\mathscr {B}(u, v):&= ( {\bs C} \, \nabla u, \nabla v)_{B_{\!R}}-k^2(n u,v)_{B_{\!R}}-\langle {\mathscr  T}_{R}[u], v \rangle_{\Gamma_{\!R}} \\ 
& =\mathscr{F}(v):=-(f,v)_{B_{\!R}}+\langle h, v \rangle_{\Gamma_{\!R}}, \quad \forall v \in H^1(B_{\!R}),
\end{split}
\end{equation}
where  by  \eqref{DtNoperator}-\eqref{kernelK},
\begin{equation}\label{DtNterm}
\langle {\mathscr  T}_{R}[u], v \rangle_{\Gamma_{\!R}}=\frac{ R}{2\pi} \Sum  \mathcal{T}_m \Big(\int_0^{2\pi} u(R,\theta) e^{-\ri m\theta}{\rm d}\theta\Big)  \Big(\overline {\int_0^{2\pi}   v(R,\theta) e^{-\ri m\theta}{\rm d}\theta}\Big).
\end{equation}
\begin{rem}\label{exAsrmk} Recall  that (cf. \cite[P. 87]{Nede01}) 
\begin{equation}\label{bndTmk}
-\frac{|m|+1}{R}\le {\rm Re}(\mathcal{T}_{m})\le \frac 1 R,\quad 0<{\rm Im}(\mathcal{T}_{m})\le k,\;\;\; |m|=0,1,\cdots.  
\end{equation}   
We can claim  the unique solvability  of   \eqref{weakform1} from  \eqref{eqnA}-\eqref{Parameter2} and  \eqref{bndTmk} (cf. \cite{Nede01}). \qed
\end{rem}

For simplicity,  we assume that the scatterer  $\Omega_-$ is a polygonal domain or a disk, though our approach is  extendable  to more complicated domain.
We partition the computational domain $B_{\!R}$  into a finite number of non-overlapping 
 straight-sided or curvilinear quadrilateral  elements $\{\Omega^{e}\}_{e=1}^E,$ 
such that the  inner interfaces  are aligned with the  ``edges'' of  the elements.  In particular, we have
\begin{equation}\label{mesh}
\Gamma_{\!R}=\bigcup_{e=1}^{E_R}\Gamma_{\!\!R}^e= \bigcup_{e=1}^{E_{R}} \big[\theta_e,\theta_{e+1}\big] =[0,2\pi],\quad r=R,  
\end{equation}
where $\Gamma_{\!\!R}^e:=\Gamma_{\!R}\cap \bar\Omega^{e}\not=\emptyset$ for  all $e\in \{1,\cdots,E\},$ and $\theta_1=\theta_{E_R+1}$ (see Figure \ref{domaingraph} (a)). 
Let $\bs \chi^e:  Q:=(-1,1)^2\to\Omega^e$ be a one-to-one elemental mapping defined by 
\begin{equation}\label{map1}
\bs r={\bs x}=(x,y)=\bs\chi^e(\xi,\eta):=\big(\chi^e_1(\xi,\eta),\chi^e_2(\xi,\eta)\big),\quad\forall\, (\xi,\eta)\in Q.
\end{equation}
 Recall that one commonly-used  elemental mapping,  originally proposed by  Gordon and  Hall \cite{gordontransform1973}, transforms  $Q$ to any quadrilateral $\Omega^e$ with  straight or curved sides (see, e.g., \cite{fischer2002high,canuto2007spectral}). Here, we shall use a special Gordon and Hall transform in 
\eqref{ordersetting}-\eqref{GordonHall} below.

Denote by $\mathcal{P}_N$ the set of all polynomials of degree at most $N$ in $[-1,1].$ Introduce the spectral-element solution space
\begin{equation}\label{solusps}
V_N^E:=\big\{v\in C(B_{\!R})\,:\, v(\bs x)|_{\Omega^e}=v(\bs \chi^e)\in {\mathcal P}_N^2,\; 1\le e\le E\big\}.
\end{equation}
The spectral-element approximation of  \eqref{weakform1} is to find $u_N^E \in  V_{N}^E$ such that
\begin{equation}\label{discretescheme}
\mathscr{B}(u_N^E,v_N^E)=\mathscr{F}(v_N^E), \quad \forall\, v_N^E \in V_N^E.
\end{equation}
 In view of Remark \ref{exAsrmk}, we can show the well-posedness of  \eqref{discretescheme} as with  \eqref{weakform1}. 

\subsection{Seamless integration of SEM with DtN TBC}    As usual,  the continuous inner product $(\cdot,\cdot)_\Omega$ can  be evaluated by element-wise discrete inner product based on  tensorial Legendre-Gauss-Lobatto (LGL) quadrature, and likewise for the term $\langle h, v_N^E \rangle_{\Gamma_{\!R}}$ (see, e.g., \cite{fischer2002high}).   
However, much care is needed to deal with
 the term  $\langle {\mathscr  T}_{R}[u_N^E], v_N^E \rangle_{\Gamma_{\!R}},$ as the DtN operator is global, but the spectral-element solution is piecewise.
One can evaluate \eqref{mainint} by using the  fast Fourier transform (FFT),  but this requires  an intermediate interpolation to interplay between   spectral-element grids and Fourier points. Since  $u_N^E|_{\Gamma_{\!R}}\in C^0,$ 
a naive interpolation only results in a first-order convergence.  
 
 In what follows, we introduce an efficient semi-analytical means to compute $\langle {\mathscr  T}_{R}[u_N^E], v_N^E \rangle_{\Gamma_{\!R}}.$
 Let $\{\xi_j=\eta_j\}_{j=0}^N$ (in ascending order)  be the LGL  points in $[-1,1],$ and let $\{l_j\}_{j=0}^N$ be the associated Lagrange interpolating basis polynomials. Correspondingly,  the spectral-element grids and basis on $\bar \Omega^e$ are given by
\begin{equation}\label{newbasis}
\bs x_{ij}=\bs \chi^e(\xi_i,\eta_j),\quad  \psi_{ij}(\bs x)=l_i(\xi)l_j(\eta),\quad 0\le i,j\le N.
\end{equation}
Formally, we can write 
 \begin{equation}\label{uebasis0}
 u_N^E(x,y)\big|_{\Omega^e}= \sum_{i,j} \tilde u_{ij}^e\,  l_i(\xi)l_j(\eta), 
 \end{equation}
 where the unknowns $\{\tilde u_{ij}^e\}$ are determined by the scheme  \eqref{discretescheme}. 

\begin{figure}[htbp]
\begin{center}
 \subfigure[Curvilinear elements]{ \includegraphics[scale=.28]{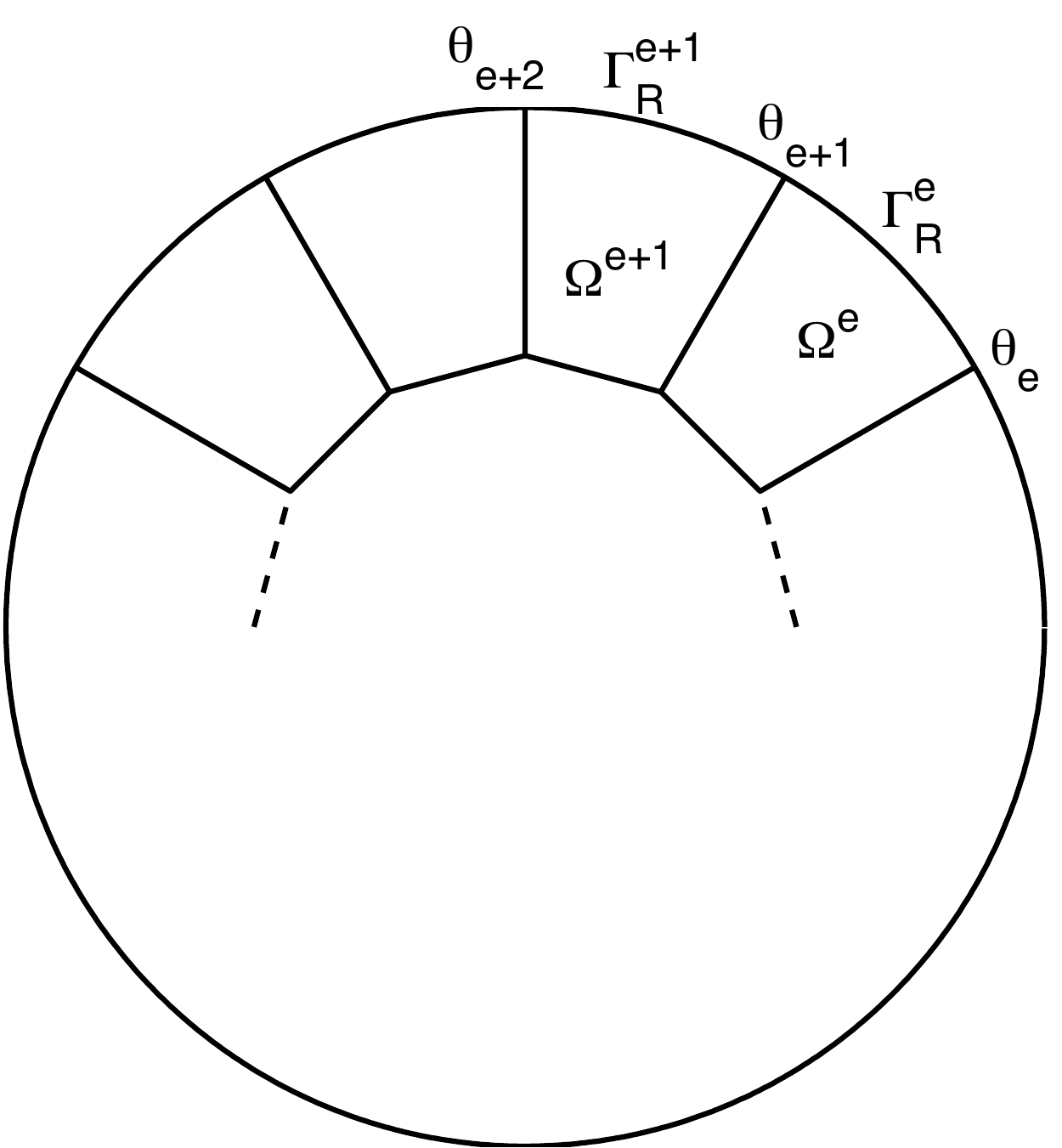}}\hspace*{6pt} 
 \subfigure[Mapped LGL points on $\Omega^e$]{ \includegraphics[scale=.31]{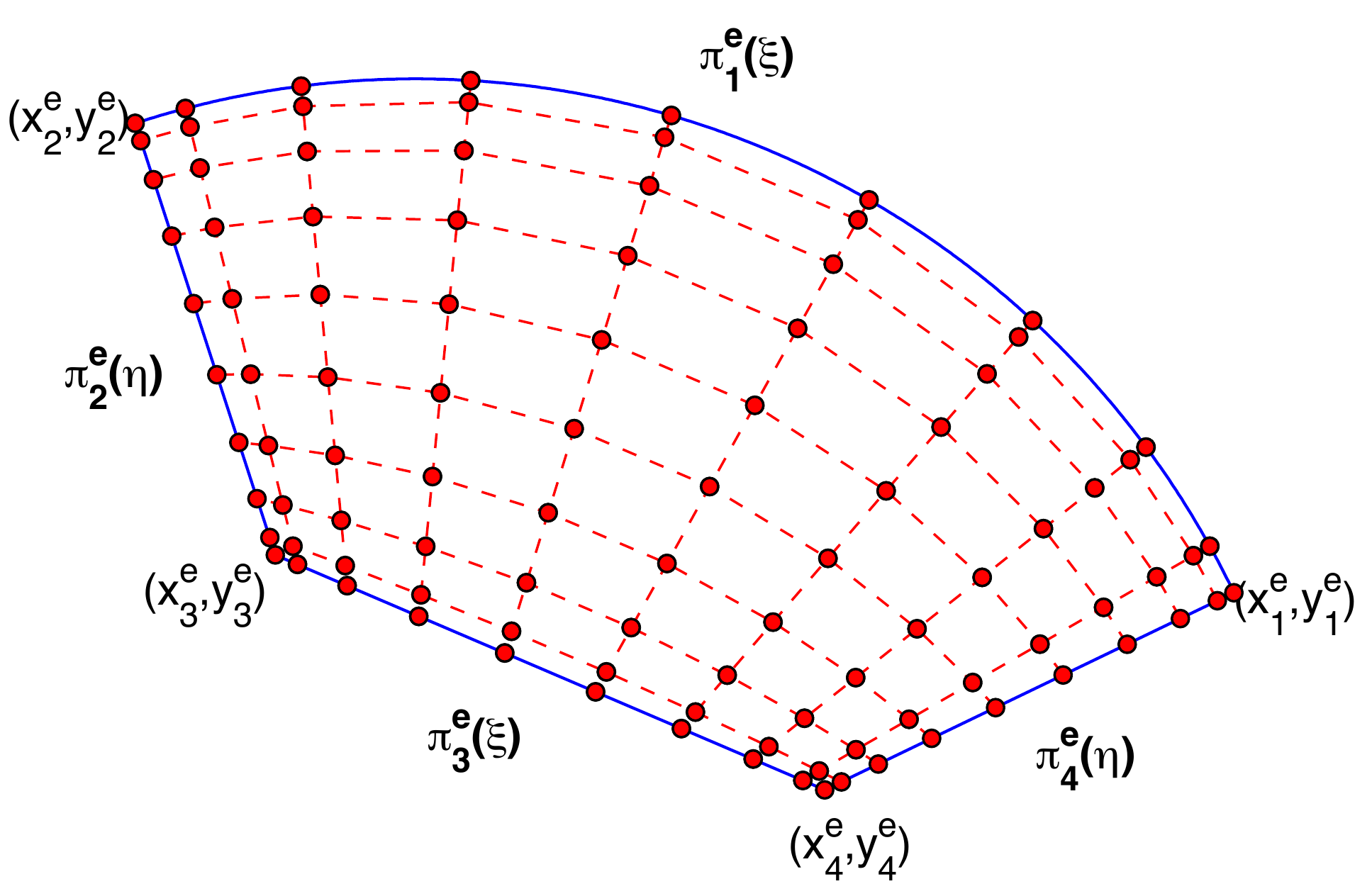}}\hspace*{8pt} 
 \subfigure[LGL points on  $Q$]{ \includegraphics[scale=.26]{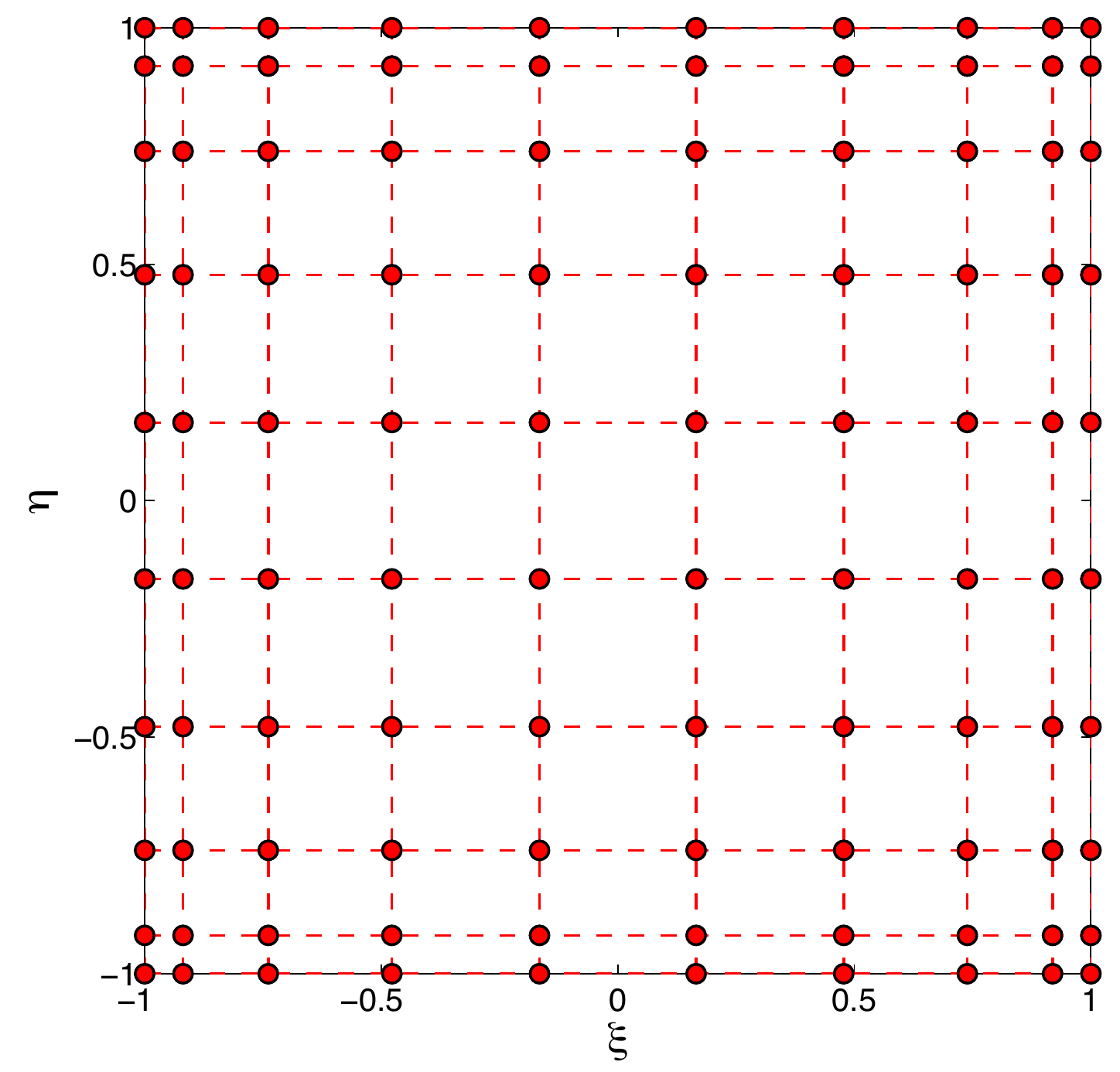}} 
  \caption{\small Curvilinear elements and tensorial  LGL points on the reference square and a curvilinear element via  the new elemental mapping.}
\label{domaingraph}
\end{center}
\end{figure}
 

 We now particularly look at the Gordon-Hall transform  for  a curvilinear element  $\Omega^e$ with vertices $\{(x_i^e,y_i^e)\}_{i=1}^4$ along $\Gamma_{\!R}$ (with three straight sides).  Let  $\{\bs \pi_j^e(t), t\in [-1,1]\}_{j=1}^4$ be, respectively, the parametric form of four sides  such that  
\begin{equation}\label{ordersetting}
\bs \pi_1^e(-1)= \bs \pi_4^e(1),\;\;  \bs \pi_1^e(1)= \bs \pi_2^e(1),\;\;  \bs \pi_2^e(-1)= \bs \pi_3^e(1),\;\; \bs \pi_3^e(-1)= \bs \pi_4^e(-1),
\end{equation}
 see  Figure \ref{domaingraph} (b). In this case,  the Gordon-Hall transform takes the form
\begin{equation} \label{GordonHall}
\begin{split}
{\bs x}={\bs \chi}^e(\xi,\eta)=& {\bs \pi}_1^e(\xi) \frac{1+\eta}{2}+{\bs \pi}_3^e(\xi)\frac{1-\eta}{2} +\frac{1+\xi}{2} {\bs \pi}_2^e(\eta) +  \frac{1-\xi}{2} {\bs \pi}_4^e(\eta)\\
&- \bigg( {\bs \pi}_1^e(-1) \frac{1-\xi}{2}       +{\bs \pi}_1^e(1)\frac{1+\xi}{2}\bigg) \frac{1+\eta}{2}  \\
&- \bigg( {\bs \pi}_3^e(-1) \frac{1-\xi}{2}       +{\bs \pi}_3^e(1)\frac{1+\xi}{2}\bigg) \frac{1-\eta}{2}\,, 
\end{split}
\end{equation}
where the edge $\eta=1$ of $Q$ is mapped to the arc $\Gamma_{\!\!R}^e=\{r=R,\; \theta\in (\theta_e,\theta_{e+1})\}$ of $\Omega_e,$ i.e.,  
\begin{equation}\label{paraformA}
\Gamma^e_{\!R}\,:\,\; x=\chi_1^e(\xi,1)=\pi_{11}^e(\xi),\;\; y= \chi_2^e(\xi,1)=\pi_{12}^e(\xi),\quad \forall\,\xi\in (-1,1).
\end{equation}
Accordingly,  the  spectral-element grids in polar coordinates on $\Gamma_{\!\!R}^e$  (see Figure \ref{domaingraph}) satisfy     
\begin{equation}\label{thetaej}
\cos \theta_j^e=R^{-1}{\pi_{11}^e(\xi_j)}\;\; {\rm or}\;\;  \sin \theta_j^e=R^{-1}{\pi_{12}^e(\xi_j)},  \quad 1\le j\le N.
\end{equation}

\vskip 2pt 
We now turn to  $\langle {\mathscr  T}_{R}[u_N^E], v_N^E \rangle_{\Gamma_{\!R}}$ in 
\eqref{discretescheme}.
Thanks to \eqref{DtNterm} and \eqref{uebasis0},   we need to evaluate 
\begin{equation}\label{mainint}
\begin{split}
\int_0^{2\pi} u_N^E(x,y)\big|_{\Gamma_{\!R}} e^{-\ri m\theta}\,{\rm d} \theta&=\sum_{e=1}^{E_R} \int_{\theta_e}^{\theta_{e+1}} u_N^E(x,y)\big|_{\Gamma_{\!R}^e} e^{-\ri m\theta}\,{\rm d} \theta\\
&= \sum_{e=1}^{E_R}\sum_i  \tilde u_{iN}^e \int_{-1}^1 l_i(\xi)  e^{-\ri m\theta(\xi)}\frac{{\rm d} \theta}{{\rm d} \xi}\, {{\rm d} \xi}.
\end{split}
\end{equation}
As the nodal basis $\{l_i\}$ can be represented in terms of  Legendre polynomials,  it suffices to compute 
\begin{equation}\label{mainint3}
{\mathbb I}_{nm}^e:=\int_{-1}^1 P_n(\xi)\, e^{-\ri m\theta(\xi)}\, \frac{{\rm d}\theta} {{\rm d}\xi}\, {\rm d}\xi, \;\;\;\; {\rm for}\;\;  n\ge 0, \;\;  |m|\ge 0,
\end{equation}
where $P_n$ is the Legendre polynomial of degree $n$, and   by \eqref{paraformA}, 
\begin{equation}\label{dthetax}
 \frac{{\rm d}\theta} {{\rm d}\xi}=\frac 1 R  \frac{{\rm d}\gamma} {{\rm d}\xi}=R^{-1}\sqrt{\big[\partial_\xi \pi_{11}^e(\xi)\big]^2+\big[\partial_\xi \pi_{12}^e(\xi)\big]^2}\,.
\end{equation}

It is seen that  the integrand is highly oscillatory for large $|m|,$ and the efficiency and accuracy in computing  ${\mathbb I}_{nm}^e$  essentially relies on the choice of the  parametric form for $\bs\pi_1(\xi).$  We next introduce a  parametric form that allows for exact evaluation of \eqref{dthetax} by analytic formulas (see Propositions \ref{newpar}-\ref{Imncomp}).   To stimulate  the idea, we first consider  a commonly-used parametric form.    

\subsubsection{A commonly-used parametric form for   $\bs \pi^e_1(\xi)$}   Following  the ideas of the cubed-sphere transformation (cf.   \cite{ronchi1996cubed,zhangj2011prolate}) and the ``ray" coordinates (cf.  \cite{karniadakis2013spectral}),   one can   project  the secant line: $(x_1^e,y_1^e)$, $(x_2^e,y_2^e)$ to the arc $\Gamma^e_{\!R}$ via the ``rays" from the origin.   This leads to the parameterisation: 
\begin{equation}\label{ArcMap1}
\bs \pi^e_1(\xi)=(\pi^e_{11}(\xi),\pi^e_{12}(\xi))=\bigg( \frac{R\,d_1(\xi)}{\sqrt{d_1^2(\xi)+d^2_2(\xi)}}, \frac{R\,d_2(\xi)}{\sqrt{d_1^2(\xi)+d^2_2(\xi)}}  \bigg),
\end{equation}
where
\begin{equation}\label{curvepara}
d_1(\xi)=\frac {x_2^e-x_1^e}2 \xi+\frac{x_1^e+x_2^e} 2,\quad d_2(\xi)=\frac {y_2^e-y_1^e}2 \xi+\frac{y_1^e+y_2^e} 2.
\end{equation}
Since $\cos \theta=R^{-1} \pi_{11}^e(\xi),$ we find   
\begin{equation}\label{cosinethegta}
\theta(\xi)=\begin{cases}
\alpha, \quad &{\rm if}\;\;\; \theta\in [0,\pi),\\[4pt]
2\pi-\alpha,\quad &{\rm if}\;\;\; \theta\in [\pi,2\pi),
\end{cases}
\quad \alpha:=\cos^{-1} \bigg(\frac{d_1(\xi)}{\sqrt{d_1^2(\xi)+d^2_2(\xi)}}\bigg)\,,
\end{equation}
and \eqref{dthetax} reads  
\begin{equation}\label{dthetax2}
 \frac{{\rm d}\theta} {{\rm d}\xi}=\frac{|x_1^ey_2^e-x_2^ey_1^e|}{2( d_1^2(\xi)+d^2_2(\xi))}, \quad \forall\, \xi\in [-1,1].
\end{equation}
  Inserting \eqref{cosinethegta} and \eqref{dthetax2}  into  \eqref{mainint3}, one immediately finds that  ${\mathbb I}_{nm}^e$ appears  complicated and 
  must be  evaluated numerically. However, the integrand  is highly oscillatory,  when $|m|$ is large. 

\subsubsection{A new parametric form for  $\bs \pi^e_1(\xi)$}\label{newsectA}  We next take a very different route to parameterise  $\Gamma^e_{\!R}.$  The essential idea is to look for 
\begin{equation}\label{ArcMap2}
\bs \pi^e_1(\xi)=\big(\pi_{11}^e(\xi), \pi_{12}^e(\xi)\big)=R (\cos  \theta, \sin \theta),\quad \theta\in [\theta_e,\theta_{e+1}],\;\;  \xi\in [-1,1],
\end{equation}
such that ${\rm d}{\gamma}=a\, {\rm d}\xi,$  that is, the arc length $\gamma$ is linear in $\xi.$ 
\begin{prop}\label{newpar} Let  $\Omega^e$ be the curvilinear element as in  {\rm Figure \ref{domaingraph} (b).}  Then  the new elemental mapping from the reference square $Q$ to $\Omega^e$ takes the form 
\begin{align}
x&=\pi_{11}^e(\xi)\frac {1+\eta} 2 +\frac{(1+\xi)(1-\eta)} 4 x_3^e+\frac{(1-\xi)(1-\eta)} 4 x_4^e,\label{xcurvetrans}\\
y&=\pi_{12}^e(\xi) \frac {1+\eta} 2 +\frac{(1+\xi)(1-\eta)} 4 y_3^e+\frac{(1-\xi)(1-\eta)} 4 y_4^e,\label{ycurvetrans}
\end{align}
where 
\begin{equation}\label{ArcMap2B}
\bs \pi^e_1(\xi)=\big(\pi_{11}^e(\xi), \pi_{12}^e(\xi)\big)=R\big(\cos(\hat \theta_e\xi+\beta_e) , \sin(\hat \theta_e\xi+\beta_e)\big),
\end{equation}
with 
\begin{equation}\label{hattheta}
\hat \theta_e=\frac{\theta_{e+1}-\theta_e}{2},\quad \beta_e=\frac{\theta_e+\theta_{e+1}}{2}.
\end{equation}
\end{prop}
\begin{proof} Let $\gamma=a\xi+b.$ 
The arc length along $\Gamma_{\!R}^e$ is  $\gamma=R(\theta-\theta_e),$ so we have 
\begin{equation}\label{gammaA}
a\xi+b=R(\theta-\theta_e). 
\end{equation}
Since $\theta=\theta_e$ (resp. $\theta=\theta_{e+1}$)  is mapped to $\xi=-1$ (resp. $\xi=1$), we find    
\begin{equation}\label{lineartransf}
a=b=\hat \theta_e R,\quad \theta=\hat \theta_e\xi+\beta_e,\quad \xi\in [-1,1].
\end{equation}
Inserting it into \eqref{ArcMap2} leads to the new parametric form \eqref{ArcMap2B}.  
Then we obtain   \eqref{xcurvetrans}-\eqref{ycurvetrans} from 
the equations of the straight sides, e.g.,  
\begin{equation}\label{Linear}
\bs \pi_4^e(\eta)=\frac{\bs x_1^e-\bs x_4^e} 2\eta +\frac{\bs x_1^e+\bs x_4^e} 2,\quad \eta \in [-1,1],  
\end{equation}
 and the Gordon-Hall transform \eqref{GordonHall}.  
\end{proof}

Observe that in distinctive contrast to  \eqref{cosinethegta}-\eqref{dthetax2},    the new transformation has a  linear dependence of  $\theta$ in  $\xi,$ so in \eqref{mainint3},
\begin{equation}\label{newrela}
\theta(\xi)=\hat \theta_e\xi+\beta_e,\quad \frac{{\rm d}\theta} {{\rm d}\xi}=\hat \theta_e.
\end{equation}   
This leads to the following analytic means for computing the integrals of interest.  
\begin{prop}\label{Imncomp} Under the new transformation in {\rm Proposition \ref{newpar}},  the integral in \eqref{mainint3} can be computed by 
\begin{equation}\label{analyticF1}
\begin{split}
& {\mathbb I}_{n0}^e=2\hat\theta_e\delta_{n0}; \quad {\mathbb I}_{nm}^e=\frac{2\hat \theta_eR}{{\rm i}^n}\sqrt{\frac{\pi}{2m\hat \theta_e}}J_{n+1/2}(m\hat \theta_e)\,
e^{-{\rm i}m\beta_e}, 
\end{split}
\end{equation}
and ${\mathbb I}_{n,-m}^e=({\mathbb I}_{nm}^e)^*$ for $n\ge 0$ and $m\ge 1, $
where $J_{n+1/2}$ is the Bessel function of the first kind, and $\hat \theta,  \beta_e$ are the same as in   \eqref{hattheta}. 
\end{prop}
\begin{proof}  By \eqref{mainint3} and \eqref{newrela}, 
\begin{equation}\label{mainint3new}
{\mathbb I}_{nm}^e=\int_{-1}^1 P_n(\xi)\, e^{-\ri m\theta(\xi)} \frac{{\rm d}\theta} {{\rm d}\xi}\, {\rm d}\xi= \hat \theta_e\, e^{-{\rm i}m\beta_e} \int_{-1}^1 
P_n(\xi) e^{-\ri  m \hat \theta_e  \xi}\, {\rm d}\xi.
\end{equation}
It is clear that for $m=0,$ we have ${\mathbb I}_{00}^e=2\hat \theta_e,$ and by the orthogonality of Legendre polynomials, we have 
${\mathbb I}_{n0}^e=0$ when $n\ge 1.$ Moreover, we have ${\mathbb I}_{n,-m}^e=({\mathbb I}_{nm}^e)^*,$ so we only 
need to compute the integrals with $m\ge 1.$   Recall the identity (cf. \cite{arfken2001mathematical})
\begin{equation}\label{Le2Bessel}
\int_{-1}^1 P_n(\xi)\,e^{-{\rm i}m  x \xi}\,{\rm d}\xi=\frac{1}{{\rm i}^n}\sqrt{\frac{2\pi}{m x}} J_{n+1/2}(m x),\quad {\rm for}\;\;  mx>0.
\end{equation}
Thus, \eqref{analyticF1} follows immediately. 
\end{proof}

\subsection{An illustrative numerical example}  As a by-product,  the new parameterisation provides an efficient  means to compute the Fourier coefficients via piecewise Legendre approximation. 
In a nutshell, we partition $[0,2\pi]$ into 
$\{[\theta_e,\theta_{e+1}]\}_{e=1}^{E_R},$  and approximate the underlying function on each subinterval by Legendre polynomials using  \eqref{ArcMap2B} so that 
the analytical formula \eqref{analyticF1} can be applied.  Here, we provide  an  example to illustrate this numerical-analytic approach.  

Consider the Fourier expansion of a plane wave (cf. \cite[P. 360]{Abr.S84}):
\begin{equation}\label{planeexp}
e^{{\rm i}k(x\cos \theta_0+y\sin\theta_0)}=\sum_{|m|=0}^{\infty} \hat g_m e^{{\rm i}m \theta}\;\;\;  {\rm with}\;\;\;  \hat g_m={\rm i}^m J_m(kR)e^{-{\rm i}m\theta_0},
 \end{equation}
 for some  constant $\theta_0$,  where  $J_m$ is the Bessel function of the first kind of  order $m$ as before.  Let $\hat g_{m,N}^{E_R}$ be the numerical approximation  to  $\hat g_m,$ and denote the error  ${\rm max}_{|m|\le M}|\hat g_m-\hat g_{m,N}^{E_R}|.$ Note that  $\hat g_m$ decays exponentially as $|m|$ increases.  We depict  the errors against $N$ (with fixed $E_R=4$)  in Figure \ref{uhatvsuNhat} (left),  
 and against the number of elements  $E_R$ (with fixed $N=10$) in  Figure \ref{uhatvsuNhat} (right), for  $R=1$, $M=20, k=10,20,30$  and $\theta_0={\pi}/{4}$.   We observe  exponential convergence  in both cases. 

\begin{figure}[h!]
  \centering
    \includegraphics[width=0.45\textwidth]{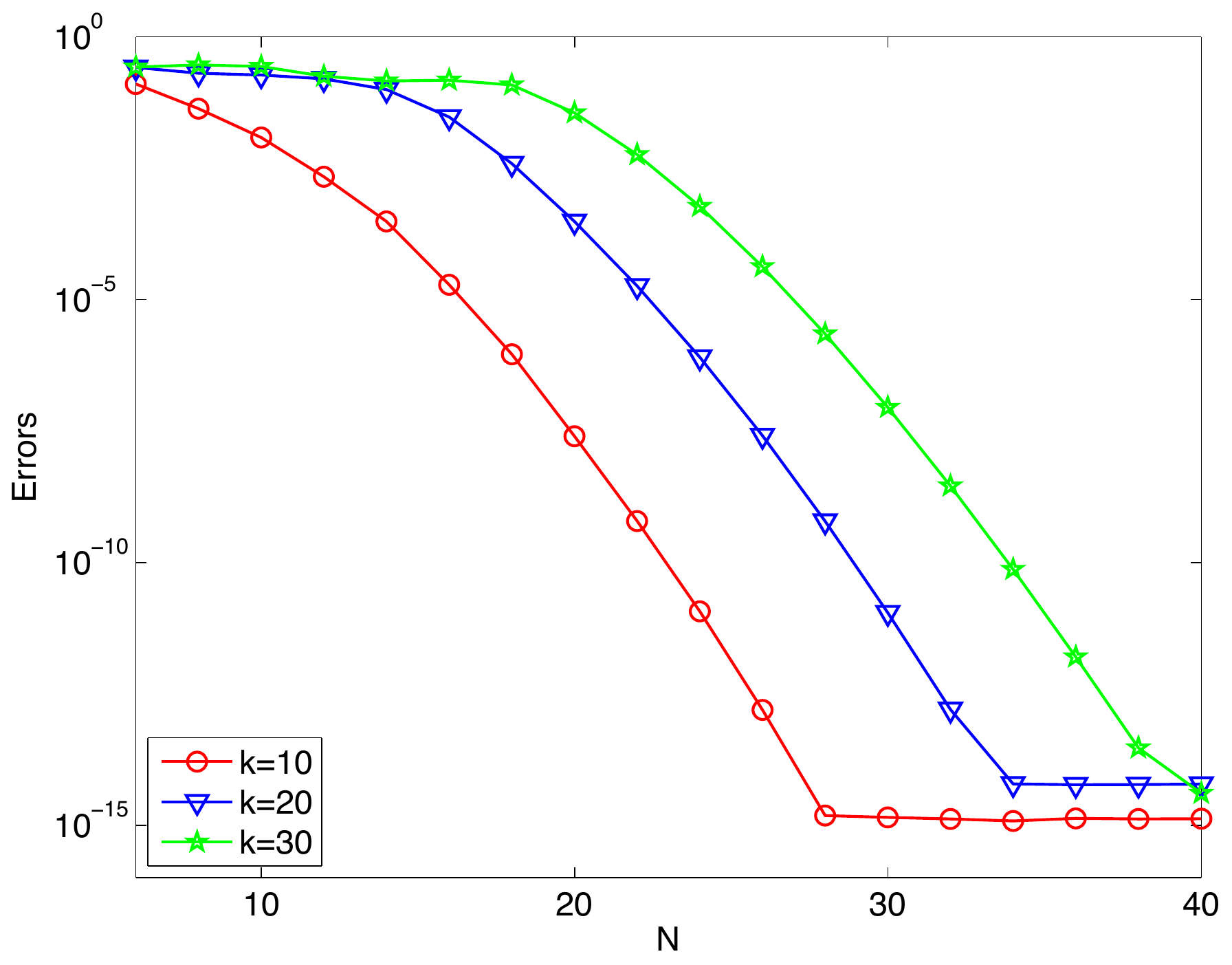}
     \includegraphics[width=0.45\textwidth]{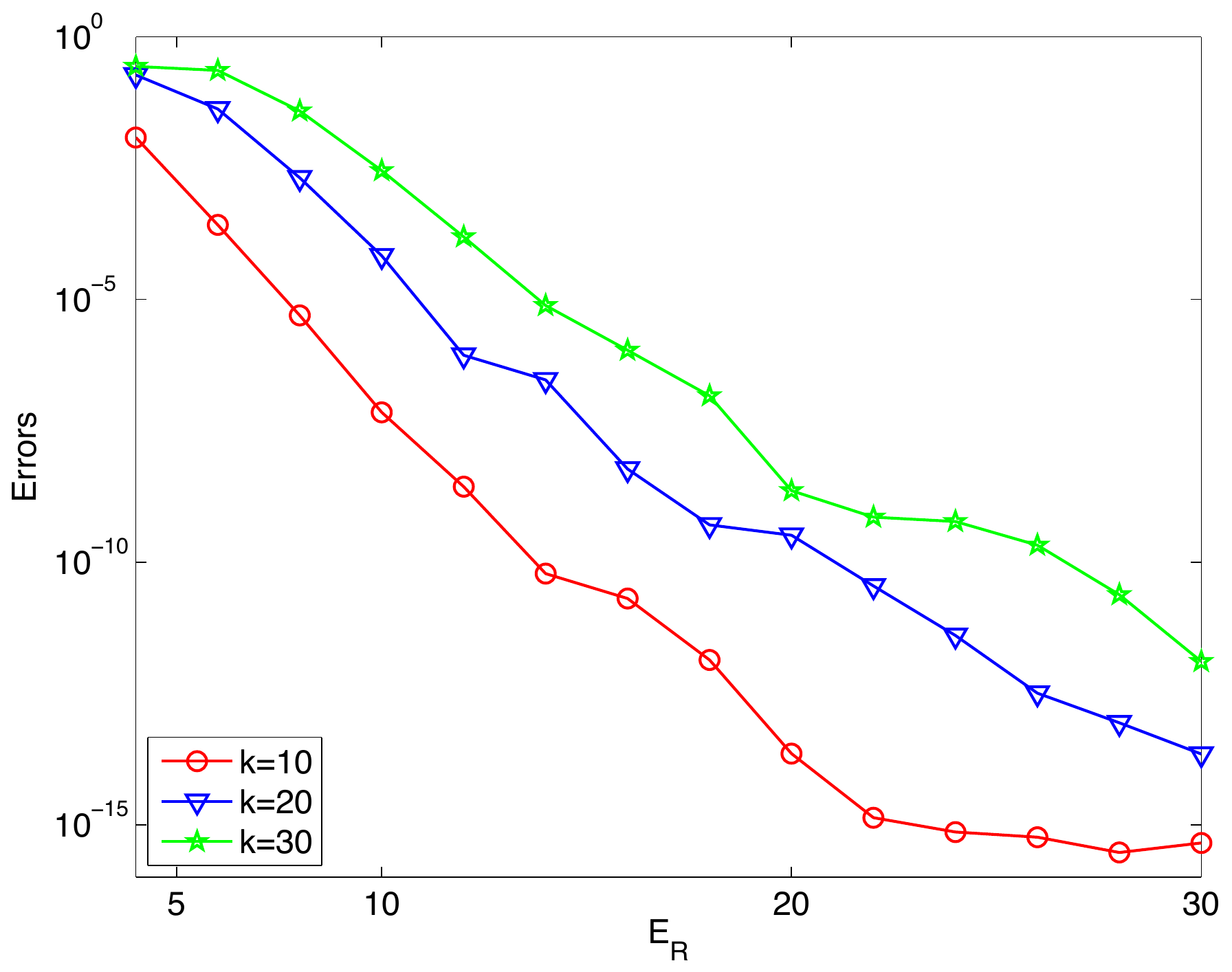}
   \caption{\small Numerical error  ${\rm max}_{|m|\le M}|\hat g_m-\hat g_{m,N}^{E_R}|$ with fixed $M=20$, $\theta_0={\pi}/{4}$ and  $k=10,20,30$. Left:  errors against $N$ with  $E_R=4$. Right: errors against $E_R$ with  $N=10$.}
    \label{uhatvsuNhat}
\end{figure}

\section{Accurate simulation of  polygonal invisibility cloaks} \label{sect:pcloak}

As already mentioned,  the invisibility cloak is one of the most exciting examples of 
transformation electromagnetics outlined in Subsection \ref{subsect:TE}.  
In this section, 
we apply the spectral-element solver proposed in Section \ref{sect3:sem}  to  simulate the  polygonal invisibility cloak, and numerically study the effects of defects, lossy media or dispersive media  in the cloaking layer.   
We particularly address the following  two  important  issues. 
\begin{itemize}
 \item[(i)] How to impose appropriate cloaking boundary conditions at the boundary of the cloaked region to 
 perfectly hide the objects inside the cloaked region? 
 \item[(ii)] How to efficiently treat the singular material parameters in spectral-element discretisation to ensure  accurate simulation? 
\end{itemize}


\subsection{Coordinate transformation and material parameters} 

As with Pendry et al. \cite{pendry.2006},  the underlying  coordinate transformation for a polygonal cloak  blows up  the origin $O$ in the original  $(\breve x,\breve y)$-coordinates  in  Figure \ref{concenfig} (a) to the polygonal domain  
$\Omega_-^{p}=A_{p}B_{p}\cdots F_{p}$ in Figure \ref{concenfig} (b), which forms the ``cloaked region".  
It is expected that  the waves from outside can not propagate  into $\Omega_-^p$ so that  any object inside  is concealed.    
 Accordingly,  the  polygonal domain 
 $\Omega_{-}$   (i.e., the polygon $AB\cdots F$) is compressed into the polygonal annulus $\Omega_-^{a}=\Omega_-\setminus \bar \Omega_-^{p}$, called the ``cloaking layer".

\begin{figure}[htbp]
 \subfigure[$(\breve x,\breve y)$-domain]{ \includegraphics[scale=.35]{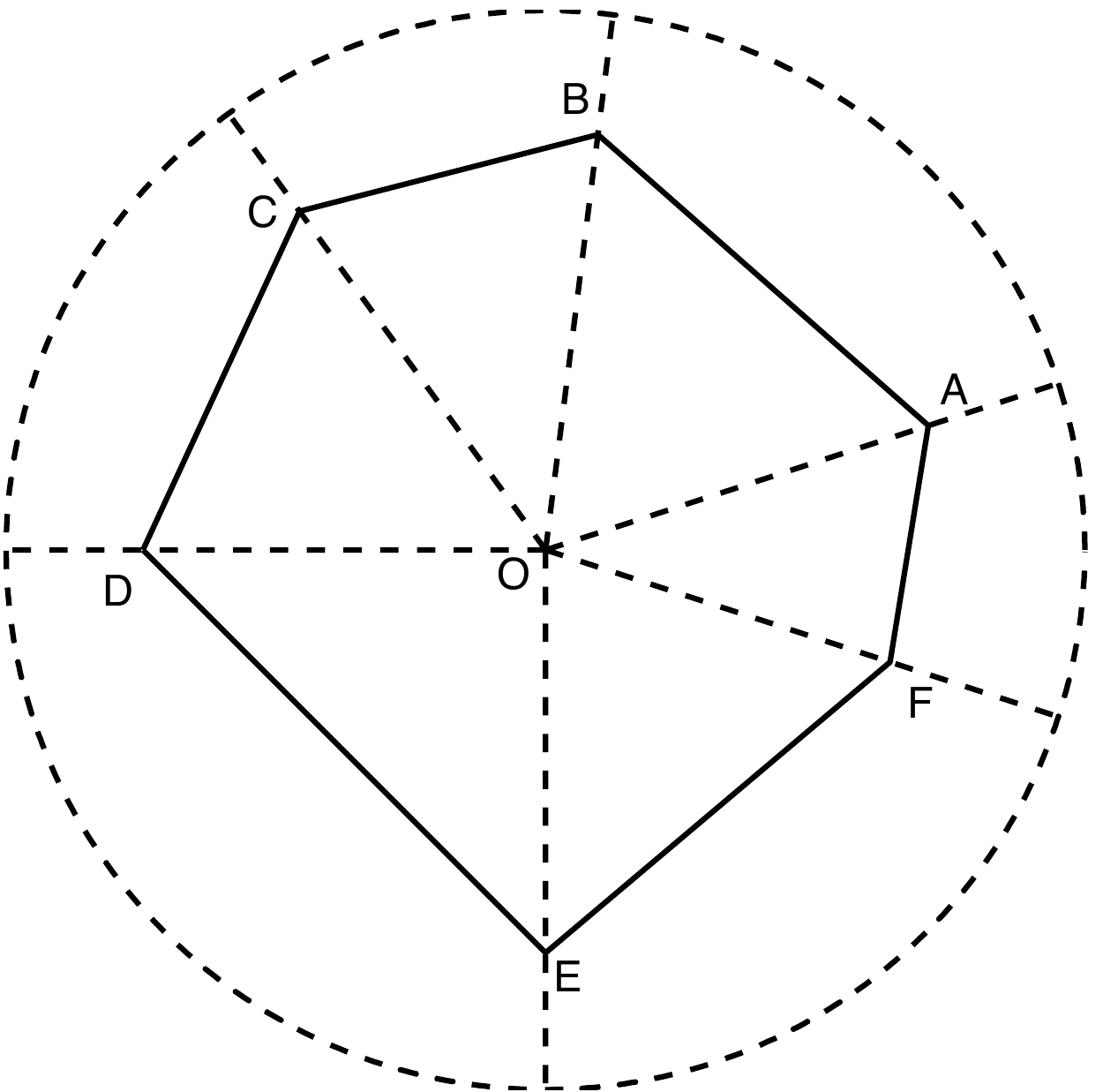}} \quad
 \subfigure[$(x, y)$-domain]{ \includegraphics[scale=.35]{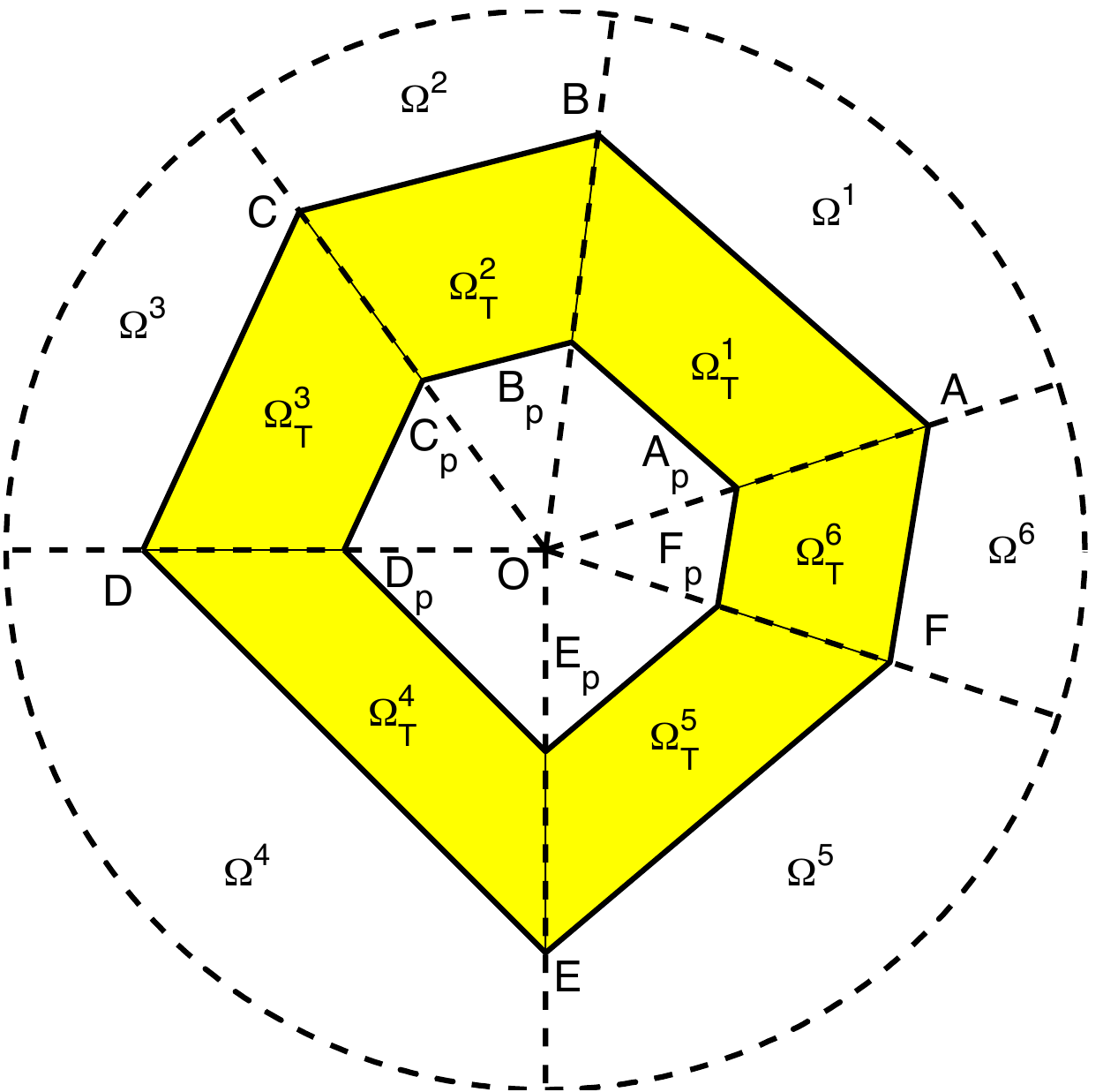}}\quad 
 \subfigure[$\bs \tau$ and $\bs n$]{ \includegraphics[scale=.3]{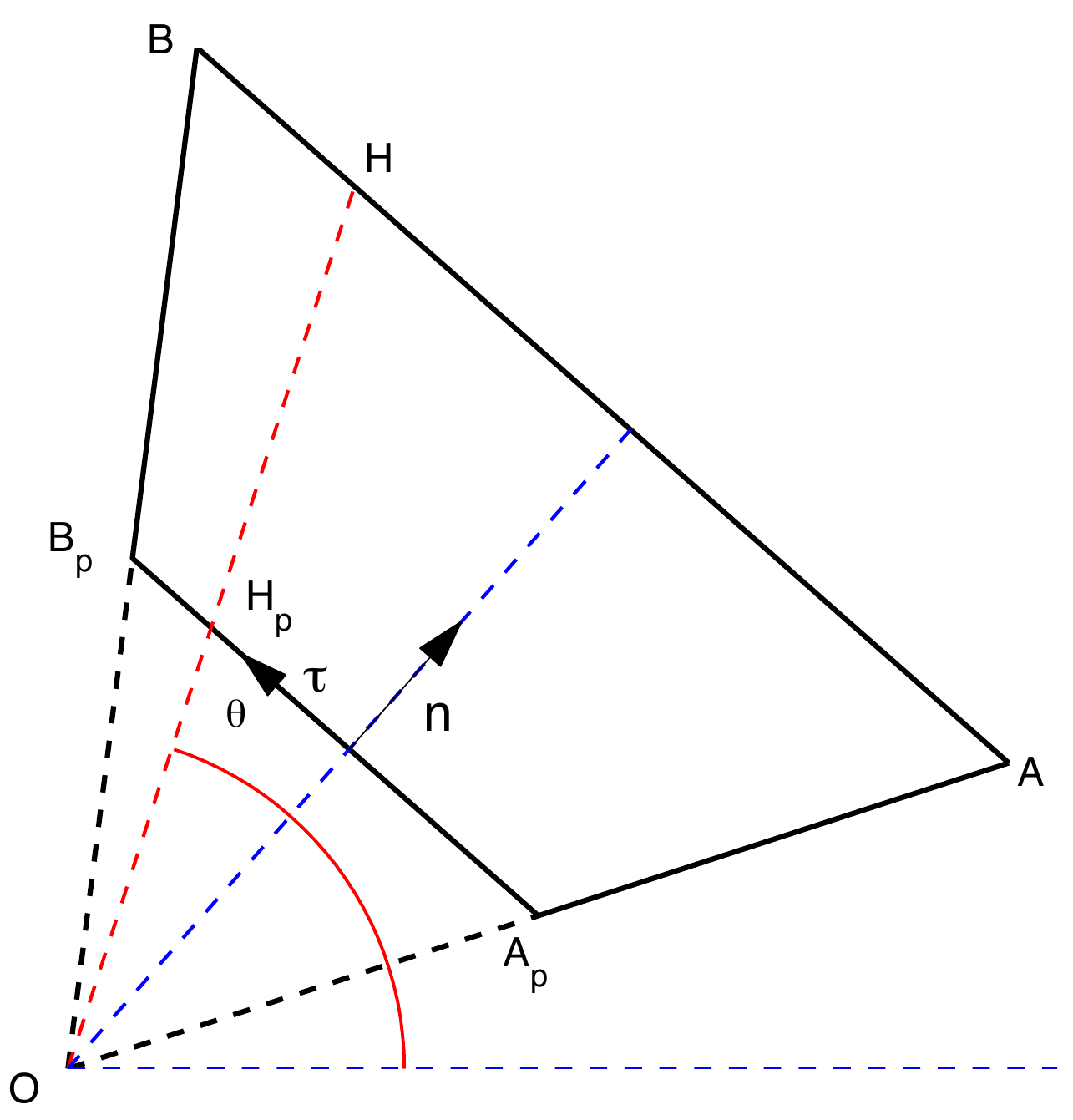}}\quad 
  \caption{\small Schematic geometry of a   polygonal  cloak. (a)  The polygonal domain in the original coordinates $(\breve x, \breve y).$ (b) Through  (\ref{cloak-tran}), the origin is spanned into the polygonal domain $\Omega_-^{p}=A_{p}B_{p}\cdots F_{p}$ that forms the cloaked region. Consequently, the original polygonal domain in (a) is compressed into the polygonal annulus 
  $\Omega_-^a$  (i.e., the shaded part) that  forms the cloaking layer.  (c) The ``local" coordinate system $(\bs \tau, \bs n).$}
\label{concenfig}
\end{figure}

In fact,  the coordinate transformations can be best  characterised in two  polar coordinates as depicted in the  diagram: 
\vspace*{6pt}
\begin{equation}\label{polartranA}
{\small \ovalbox{Original: $(\breve x,\breve y)$} \xleftarrow{\makebox[0.3cm]{}}\hspace*{-2pt}\xrightarrow{\makebox[0.3cm]{}} \ovalbox{Polar: $(\breve r,\breve \theta)$} \xrightarrow{\makebox[1.8cm]{Transform}} \ovalbox{Polar: $(r, \theta)$} \xleftarrow{\makebox[0.3cm]{}}\hspace*{-2pt}\xrightarrow{\makebox[0.3cm]{}} 
\ovalbox{Physical: $(x,y)$}}
\end{equation}
\vspace*{1pt} 

\noindent With this,  the coordinate transformation for a polygonal cloak  takes the form (see, e.g., \cite{zhangjj2008cloak}):
 \begin{equation}\label{cloak-tran}
 \begin{cases}
r=(1-\rho)\,\breve r+  \, R_1, \quad &  \breve r\in [0, R_2],\;\;\; r\in [R_1,R_2], \\[4pt]
\theta=\breve \theta,\quad &  \breve \theta, \theta  \in [0,2\pi),
\end{cases}
\end{equation}
where 
\begin{equation}\label{radiconst}
\rho=\frac{OA_{p}}{OA}=\frac{OB_{p}}{OB}=\cdots=\frac{OF_{p}}{OF},\quad  0<\rho <1,
\end{equation}
and $(R_1(\theta),\theta)$ (resp. $(R_2(\theta),\theta)$) is  the polar parametric form of the side of the cloaked region $\Omega_-^p$   (resp. the domain $\Omega_-$).    

We find it is more convenient to introduce a ``local" coordinate system to represent $R_i$ ($i=1,2$) when we impose the cloaking boundary conditions later on.   
Consider any side of the insidemost  polygonal domain $\Omega_-^{p},$ say $A_pB_p$ in Figure \ref{concenfig} (c), 
  with vertices  $(x_1, y_1)$ and $(x_2, y_2).$  Then its unit tangential vector $\bs \tau$ and unit normal vector $\bs n$ 
 are given by 
\begin{equation}\label{normaltangential}
{\bs \tau}=\frac{(x_2-x_1,y_2-y_1)}{\sqrt{(x_2-x_1)^2+(y_2-y_1)^2}}:=(\tau_1, \tau_2),\quad  {\bs n}=(\tau_2,-\tau_1),
\end{equation}
respectively, which form a ``local" coordinate system.  Noting that for any $(x,y)=(R_1\cos\theta,R_1\sin\theta)$ along this side, we have 
$(x,y)\cdot \bs n= (x_2,y_2)\cdot\bs n =(x_1,y_1)\cdot\bs n,$ and  
\begin{equation}\label{hR1}
R_1(\theta)=\frac{\tau_2 x_2-\tau_1 y_2}{\tau_2 \cos \theta -\tau_1 \sin \theta}, \quad R_2(\theta)=\rho^{-1} R_1(\theta).
\end{equation}

We can   use   \eqref{parameter3} and \eqref{cloak-tran}  to derive   the coefficients $\bs C$ and $n$ in the Helmholtz equation \eqref{eq1}.  More precisely,   the coefficients in  $B_{\!R}=\Omega_-^{p}\cup \bar \Omega_-^{a}\cup \Omega_+$  take different forms as follows.  
\vskip 5pt
\begin{itemize}
\item[(i)] In the cloaking layer $\Omega_-^{a}$,  the entries of $\bs C$ in \eqref{parameter2}-\eqref{parameter3} are given by 
(see Appendix \ref{AppendixA}):
\begin{align}
C_{11}&= \frac{r-R_1}{r}\frac{x^2}{r^2}+\frac{1}{r(r-R_1)}\bigg( \frac{{\rm d} R_1}{{\rm d} \theta}\frac{x}{r}-y   \bigg)^2\,, \label{ck11}\\[6pt]
C_{22}&=\frac{r-R_1}{r}\frac{y^2}{r^2}+\frac{1}{r(r-R_1)}\bigg(\frac{{\rm d} R_1}{{\rm d} \theta}\frac{y}{r}+x          \bigg)^2\,, \label{ck22}\\[6pt]
C_{12}&=\frac{r-R_1}{r}\frac{xy}{r^2}+\frac{1}{r(r-R_1)}\bigg( \frac{{\rm d} R_1}{{\rm d} \theta} \frac{x}{r}-y   \bigg)\bigg( \frac{{\rm d} R_1}{{\rm d} \theta} \frac{y}{r}+x   \bigg),\label{ck12}
\end{align}
and 
\begin{equation}\label{nkcase2}
n= \frac{r-R_1}{r(1-\rho)^2}\,.  
\end{equation}

\vskip 5pt 
\item[(ii)] In  $\Omega_+= B_R\setminus\bar \Omega_-$, we have  $\bs C=\bs I_2$ and $n=1.$ 
\vskip 5pt
\item[(iii)] 
Following   \cite{pendry.2006}, we set ${\bs C}=\bs I_2$ and $n=1$ in $\Omega_-^{p}.$
\end{itemize}
\begin{rem}\label{newrmk} It is seen from \eqref{hR1} that $R_1$ has a different representation in different trapezoids  $\{\Omega_T^i\}_{i=1}^6$ in Figure \ref{concenfig} (b), so  the entries of $\bs C$ and the coefficient $n$ are piecewise functions, which might be not continuous across  the  sides in the  radial direction, e.g., $A_pA$. \qed  
\end{rem}

\subsection{Transmission conditions  and  new CBCs}  For notational convenience,  we partition the disk $B_R$ into  a finite number of non-overlapping subdomains  based on the nature of  material parameters.  Take the setting in Figure \ref{concenfig} as an example, and decompose 
\begin{equation}\label{notadecomp} 
B_{\!R}=\Omega_-^{p}\cup \bar \Omega_-^{a}\cup \Omega_+,\quad \bar \Omega_-^{a}=\cup_{i=1}^6 \bar  \Omega_T^i,\quad  \bar \Omega_+=\cup_{i=1}^6 \bar  \Omega^i,
\end{equation}
where $\{\Omega_T^i\}$ are trapezoids. Denote  by $\Gamma_{\!\pm}^p $ the  ``outside"  and ``inside"  boundaries of the cloaked region $\Omega_-^p,$ respectively. 
We further denote by $\Gamma^a:=\bar \Omega_+ \cap \bar \Omega_-^{a},$ i.e., the outer boundary of $\Omega^a.$

We next impose interface conditions:  
 (i)  at $\Gamma^a$ and sides in the radial direction of the trapezoids  $\{\Omega_T^i\};$
  and (ii) at the cloaking boundary $\Gamma_{\!\pm}^p.$
In the former case,    the material parameters are bounded below and above, so 
we impose the usual  transmission conditions. 
Recall that under the TE polarisation, $\bs E=(0,0,u)^t,$ so we have 
\begin{equation}\label{tranmissionA}
\llbracket u \rrbracket=\llbracket \bs C \nabla u \rrbracket=0\; \quad \text{for case (i)}.
\end{equation}

In the latter case,   the material parameters degenerate at $\Gamma^p_{\!+}$ (cf. \eqref{ck11}-\eqref{ck12}),  so the above transmission conditions are not applicable.  In fact, the tangential component of $\bs E$ is not continuous across the cloaking boundary.  Indeed, 
 it is the singular non-isotropic medium in the cloaking layer  $\Omega_-^{a}$ that  offers  the possibility of achieving invisibility within the innermost polygonal domain  $\Omega_-^{p}$.  
In practice, the perfect conduct shell, i.e., PEC (under TM polarisation) or PMC (under TE polarisation) was enforced  at $\Gamma_{\!+}^p$ in simulations (see, e.g., \cite{cummer2006full}  for circular cylindrical cloaks, \cite{jiang2008elliptical, kwon2008two} for elliptic cloaks,   \cite{zhangjj2008cloak, Wu2009material} for polygonal cloaks, and \cite{zhao2008full, Li2012developing} for time domain cloaks).  However, as commented in \cite{cummer2006full}, such a condition was sufficient but not necessary. Indeed, it was  shown in \cite{yang2014accurate}, the PEC or PMC  could not lead to an independent, meaningful boundary condition for circular cylindrical cloaks  in polar coordinates.

Following the spirit of \cite{yang2014accurate},  we next introduce the essential ``pole"  conditions associated with  the singular transformation \eqref {cloak-tran}. 
 As we will see, it is important  to use the  ``local" coordinate system $(\bs \tau,\bs n)$ in 
\eqref{normaltangential}  to decompose the differential operators and then carefully study the singularity.  
\begin{prop}\label{essentialpole} Let  $(\bs \tau,\bs n)$  be the ``local" coordinate system in 
\eqref{normaltangential}.  The CBCs take the form 
\begin{equation}\label{newCBC}
\nabla_{\!\bs \tau}u^{+}=\bs \tau \cdot \nabla u^+=0\;\;  {\rm at}\;\; \Gamma_{\!+}^p; \quad   \nabla_{\!\bs n}u^{-}=\bs n \cdot \nabla u^-=0  \;\; {\rm at}\;\; \Gamma_{\!-}^p, 
\end{equation}
where $u^{+}=u|_{B \setminus \bar \Omega_-^p}$ and $u^{-}=u|_{\bar \Omega_-^p}.$
\end{prop}
\begin{proof} 
In order to achieve spectral accuracy in simulations involving singular transformations, we follow   \cite{yang2014accurate} (also see \cite{ShenTaoWang2011}) and require that  the well-behaved and finite electromagnetic fields in the original coordinates $\breve {\bs r}=(\breve x,\breve y)$  must be well-behaved and finite  in the new coordinates $\bs r=(x,y)$.  We therefore apply this principle to the magnetic field in the cloaking layer, and show that   the essential ``pole" condition of the transformation \eqref {cloak-tran} takes the form 
\begin{equation}\label{CBC}
\nabla_{\!\bs \tau}u^{+}=0 \;\;\;\; {\rm at}\;\;\; \Gamma_{\!+}^p.
\end{equation} 
Recall that by  \eqref{magnetic}, $\bs H^+=\bs H|_{B \setminus \bar \Omega_-^p}$ can be expressed as
\begin{equation} \label{magnetic2}
{\bs H}^+=(H_1^+, H_2^+,0)^t=\frac{1}{{\rm i} \omega \mu_0}\big(C_{12} u_x^++C_{22} u_y^+, -C_{11}u_x^+-C_{12}u_y^+,0\big)^t. 
\end{equation}
Then  we have 
\begin{equation}\label{poinval}
u_x^+= \tau_1 \nabla_{\!\bs \tau} u^+- \tau_2 \nabla_{\!\bs n} u^+,\quad u_y^+=\tau_2 \nabla_{\!\bs \tau} u^+ + \tau_1 \nabla_{\!\bs n} u^+.
\end{equation}
Inserting \eqref{ck11}-\eqref{ck12} and \eqref{poinval} into \eqref{magnetic2} and collecting the terms,  we obtain 
\begin{align}
\ri \omega \mu_0H_1^+&= -\Big(\frac{r-R_1}{r}\tau_1+\frac{{\rm d}R_1}{{\rm d}\theta}\frac{\tau_2}{r}\Big) \nabla_{\!\bs n}u^+ +\frac{{\rm d}R_1}{{\rm d}\theta}\frac{\tau_1}{r}\nabla_{\!\bs \tau}u^+ +\frac{\tau_2}{r-R_1}\Big(r+\Big(\frac{{\rm d}R_1}{{\rm d}\theta}\Big)^2\frac{1}{r}\Big)  \nabla_{\!\bs \tau}u^+, \label{H1}\\
\ri \omega \mu_0H_2^+ &=-\Big(\frac{r-R_1}{r}\tau_2-\frac{{\rm d}R_1}{{\rm d}\theta}\frac{\tau_1}{r}\Big) \nabla_{\!\bs n}u^+ +\frac{{\rm d}R_1}{{\rm d}\theta}\frac{\tau_2}{r}\nabla_{\!\bs \tau}u^+ -\frac{\tau_1}{r-R_1}\Big(r+\Big(\frac{{\rm d}R_1}{{\rm d}\theta}\Big)^2\frac{1}{r}\Big)  \nabla_{\!\bs \tau}u^+.\label{H2}
\end{align}
By  \eqref{hR1}, 
\begin{equation}\label{dRdtheta}
 \frac{{\rm d}R_1}{{\rm d}\theta}=R_1\dfrac{\tau_1 x + \tau_2 y}{\tau_2 x- \tau_1 y},
\end{equation}
which is uniformly bounded in $\bar \Omega_-^a$ (note: $\tau_2 x- \tau_1 y=0$ implies  the side of the polygonal passes through  the origin, which is not possible). Thus, $\bs H^+$ is a finite field in $\bar \Omega_-^a$, if and only if  the first condition in  \eqref{newCBC} holds.  

Let $\bs H^-=\bs H|_{ \bar \Omega_-^p}.$ From \eqref{normaltangential}, \eqref{ck11}-\eqref{ck12}  and \eqref{magnetic2},   we obtain  
\begin{equation}\label{casesA}
\begin{split}
\bs n\times\big(\bs H^+ -\bs H^-\big)\big|_{r=R_1}&=\frac{1}{{\rm i} \omega \mu_0} \big(0,0,-\bs n \cdot (\bs C \nabla u^+)+   \nabla_{\bs n} u^-\big)^t \big|_{r=R_1}  \\
&=\frac{1}{{\rm i} \omega \mu_0} \Big(0,0, -\frac{r-R_1}{r}\nabla_{\!\bs n}u^+ +\frac{1}{r}\frac{dR_1}{d \theta}\nabla_{\!\bs \tau}u^+ +\nabla_{\bs n}u^- \Big)^t\Big|_{r=R_1},
\end{split}
\end{equation}
where the restriction at $r=R_1$ means that we approach the cloaking boundary from the inside and the outside. 
Suppose that $\nabla_{\bs n} u^+$ is finite. Then by   \eqref{CBC},  
\begin{equation}\label{simplfyA}
\bs n\times\big(\bs H^+ -\bs H^-\big)\big|_{r=R_1}=\frac{1}{{\rm i} \omega \mu_0} \big(0,0,  \nabla_{\bs n}u^- \big)^t\big|_{r=R_1}.
\end{equation} 
By imposing the  second condition in \eqref{newCBC}, we find 
\begin{equation}\label{ncurlE}
\bs H^-\times \bs n=\frac 1 {\ri \omega\mu_0}\big(\nabla \times \bs E^-\big) \times \bs n=0\quad  {\rm at}\;\; \Gamma_-^{p},
\end{equation}
and by \eqref{casesA},   the tangential component of the magnetic field $\bs H$ is continuous across the cloaking boundary. 
\end{proof}

\begin{rem}\label{newcaseA}  Weder \cite{Weder08} proposed CBCs for point-transformed   3D invisibility cloaks: 
\begin{align}
&\bs E^+\times \bs n=\bs H^+\times \bs n=\bs 0\quad {\rm at}\;\; \partial K_+; \label{ncurlEWeder}\\ 
&(\nabla \times \bs E^-)\cdot\bs n=(\nabla\times \bs H^-)\cdot \bs n= 0\quad {\rm at}\;\; \partial K_-, \label{ncurlEWeder2} 
\end{align}
where $K$ is the cloaked region.  These allowed for the decoupling of the governing equations of the inside and outside, and the spherical cloak was considered as a  particular  application.   It is important to remark that \eqref{ncurlEWeder}-\eqref{ncurlEWeder2} are not applicable to the 2D polygonal cloak, as 
$$\bs E^+\times \bs n=(\tau_1,\tau_2,0)^t\, u^+\not =\bs 0\quad {\rm at}\;\; \Gamma_{\!+}^p.$$
Indeed, $\bs E^+=(0,0,u^+)^t$ does not vanish at the outer cloaking boundary.  Moreover, the condition \eqref{ncurlEWeder2}  is different from \eqref{ncurlE}.  
Notably,  the CBCs in \eqref{newCBC}  also leads to the decoupling of the inside and outside, as we will see below.   
\qed 
\end{rem}

With the new CBCs in \eqref{newCBC},  the governing equations   can be  decoupled into 
\begin{align}
& \nabla \cdot({\bs C}(\bs r)\, \nabla u^+(\bs r))+ k^2n(\bs r) u^+(\bs r)=0 \quad  {\rm in}\;\; B_{\!R} \setminus \bar \Omega_-^p, \label{eq1conL}   \\
& \text{Transmission conditions in \eqref{tranmissionA}}; \quad 
\nabla_{\!\bs \tau}u^+=0 \;\;\;  {\rm at}\;\; \Gamma_{\!+}^p,\label{eq11conL}\\
&\partial_{r} u^+ -{\mathscr T} _{R} [u^+]=h   \quad {\rm at}\;\;  \partial B_{\!R},  \label{eq2conL}
\end{align}
and
\begin{align}\label{inner2} 
& \Delta u^- +k^2 u^- =0  \quad  {\rm in}\;\;\Omega_-^{p}; \quad \nabla_ {\bs n} u^-=0   \quad  {\rm at}\;\; \Gamma_-^p. 
\end{align}
\begin{rem}\label{markA}
It is standard to show that \eqref{inner2} has a unique solution $u^-(\bs r)\equiv 0,$ if  $k^2$ is not an eigenvalue of $-\Delta$ in $\Omega_-^p$ with homogeneous Neumann boundary condition.   \qed
\end{rem}

\subsection{Treatment of singularities  in spectral-element discretisation}    
For accurate simulation,  it is advisable to build the boundary   condition $\nabla_{\!\bs \tau}u^+=0$ at $\Gamma_{\!+}^a$ in the spectral-element solution space.  Naturally, we   adopt the partition in \eqref{notadecomp},
and  consider for example $\Omega^{1}_{T}=A_pB_pBA,$ where  $A_pB_p$ 
  has vertices  $(x_1, y_1)$ and $(x_2, y_2).$  Suppose that   $A_pB_p$  is mapped to $\eta=-1$ via  the Gordon-Hall transform \eqref{GordonHall}, namely, 
  \begin{equation}\label{tranA}
{\bs r}={\bs \chi}^e(\xi,-1)
=\Big(\frac{x_2-x_1}{2}\xi+\frac{x_1+x_2}{2}, \frac{y_2-y_1}{2}\xi+\frac{y_1+y_2}{2}\Big),\quad \xi\in [-1,1],
\end{equation}
which yields 
\begin{equation}\label{partialderi}
\partial_x \xi=\frac{2}{x_2-x_1},\;\; \partial _y \xi=\frac{2}{y_2-y_1},\;\; \partial_x \eta=\partial_y \eta=0.
\end{equation}
One  verifies readily  that
\begin{align*}
0&=\nabla_{\bs \tau}u^+\big|_{A_pB_p} 
=
\big(\tau_1 \partial_x \xi+\tau_2 \partial_y \xi\big)\partial_{\xi}u^e(\xi,-1)+\big(\tau_1 \partial _x \eta +\tau_2 \partial_y\eta\big) \partial_{\eta}u^e(\xi,-1)\\
&=\frac{4}{\sqrt{(x_1-x_2)^2+(y_1-y_2)^2}}\partial _{\xi} u^e(\xi,-1), 
\end{align*}
where $u^e=u^+({\bs \chi}^e(\xi,-1)).$
This leads to the corresponding boundary condition in  $(\xi,\eta)$-coordinates: 
\begin{equation}\label{polerefer}
\partial _{\xi} u^e(\xi,-1)=0,\;\;\;  \xi \in [-1,1].
\end{equation}

 Accordingly, we modify the approximation space in  \eqref{solusps} as
 \begin{equation}\label{solusps2}
 \begin{split}
V_N^E=\big\{v\in C(B_R \setminus \bar \Omega_-^p)\,:\, & v(\bs r)|_{\Omega^e}=v(\bs \chi^e(\xi,\eta))\in {\mathcal P}_N^2,\; 1\le e\le E\;\; {\rm and} \\
& \nabla_{\bs \tau} v(\bs r)|_{\Omega^e_T\cap \Gamma_{\!+}^p}=\partial_\xi v(\bs \chi^e(\xi,-1))=0,\;\; 1\le e\le E_T \big\},
\end{split}
\end{equation}
where $E=12$ and $E_T=6$ for the setting in Figure \ref{concenfig}. To meet the boundary condition at $\Gamma_{\!+}^p$, we modify the tensorial nodal basis in \eqref{newbasis}: 
\begin{align}\label{modibasis} 
 \psi_{00}=l_0(\eta),\quad  \psi_{ij}=l_i(\xi)l_j(\eta),\quad 0\le i\le N,\;\; 1\le j \le N,
\end{align}
and 
one verifies readily that 
\begin{align}\label{modibasis2} 
 \partial_\xi \psi_{00}=0,\quad  \partial_\xi \psi_{ij}\big|_{\eta=-1}=l_i'(\xi)l_j(-1)=0,\quad 0\le i\le N,\;\; 1\le j \le N.
\end{align}
With such a modification,  
the singularity can be absorbed by the basis and  the spectral-element scheme \eqref{discretescheme} can be implemented as usual. 

\subsection{Simulation results for perfect polygonal cloaks} 
We  now provide some numerical results and compare them with results obtained from the finite-element-based  COMSOL Multiphysics   package. 
  Assume  that the incident source  is a TE plane wave with an incident angle $\theta_0$:
\begin{equation} \label{incident1}
u_{\rm in}(r,\theta)=e^{{\rm i}kr\cos(\theta-\theta_0)}=\sum_{|m|=0}^{\infty} {\rm i}^m  J_m(kr)e^{{\rm i}m(\theta-\theta_0)},
\end{equation}
so  $h$ in \eqref{uDtN} takes the form
$$
h=\partial_{r} u_{\rm in}-{\mathscr T} _{R}[u_{\rm in}]={\rm i}k\cos(\theta-\theta_0)u_{\rm in}-\sum_{|m|=0}^{\infty}{\rm i}^m J_m(kR){\mathcal T}_m e^{-{\rm i}m \theta_0},
$$
where ${\mathcal T}_m$ is defined in \eqref{kernelK}. In the following tests, we consider a pentagonal cloak where the polar coordinates  of $A,B,\cdots, E$ are 
$\theta=\pi/ 5, 3 \pi/7,  8 \pi/9,  6 \pi/5, 7 \pi/4,$ respectively, and $r=0.7$ for all vertices.   
We take   $\rho=0.7$  in \eqref{radiconst} and $R=1.0$, and     truncate  the infinite series in \eqref{incident1} with a cut-off number  $M=60.$  
\begin{figure}[htbp]
 {~}\hspace*{-16pt}\subfigure[FEM electric field]{ \includegraphics[scale=.23]{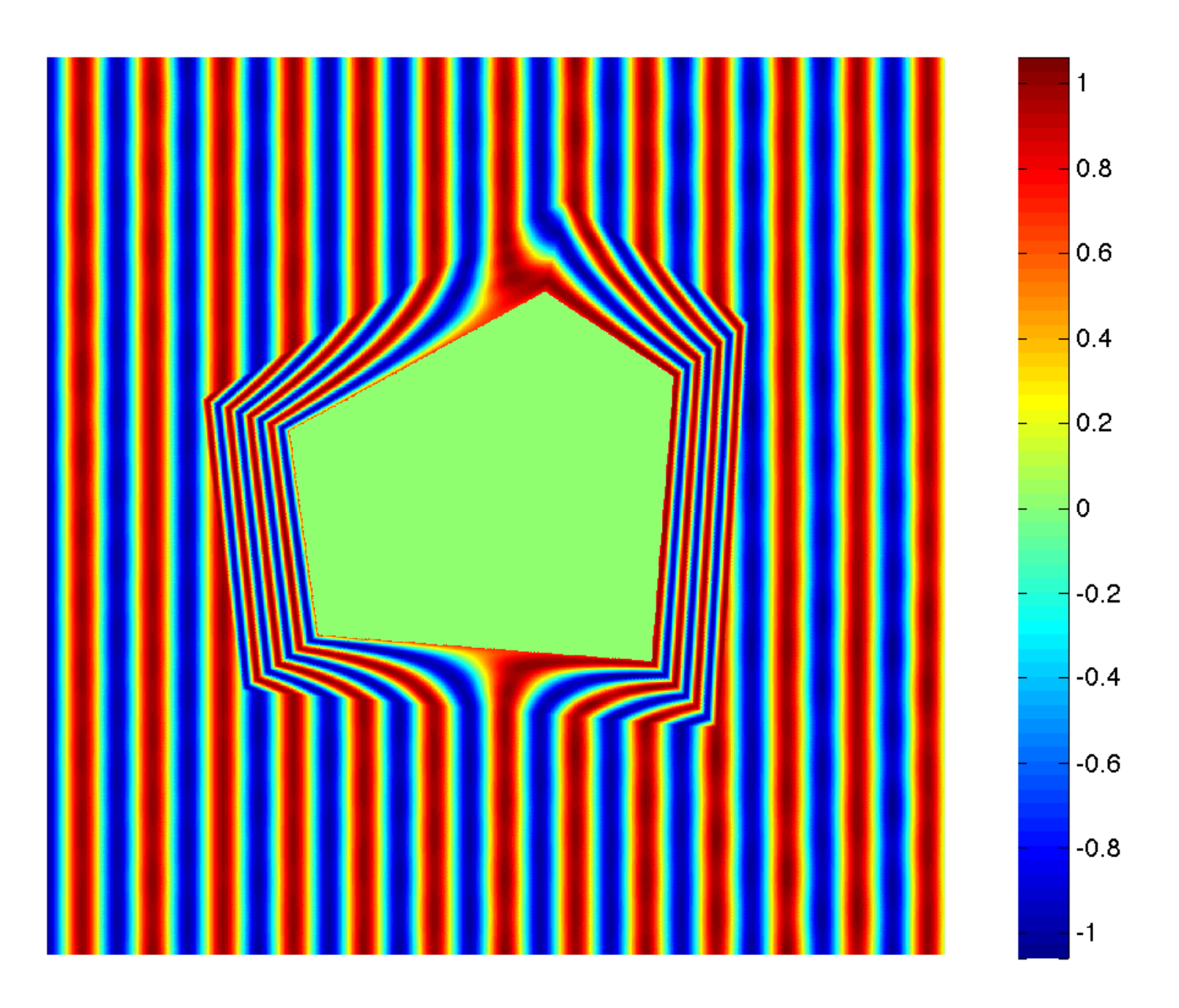}} \hspace*{-12pt}
\subfigure[Contour (FEM)]{\includegraphics[scale=.27]{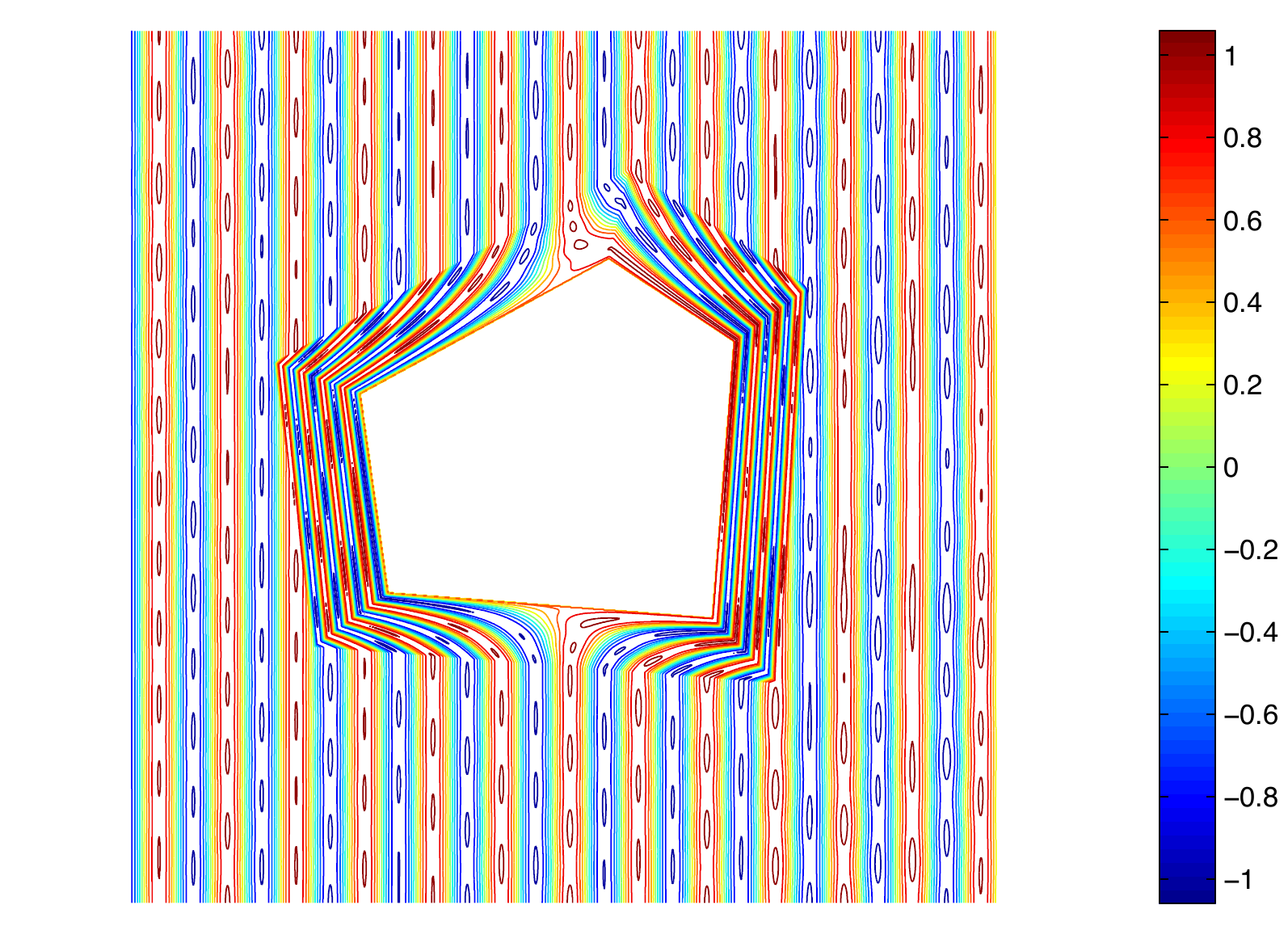}} \hspace*{-10pt} 
 \subfigure[SEM electric field]{ \includegraphics[scale=.23]{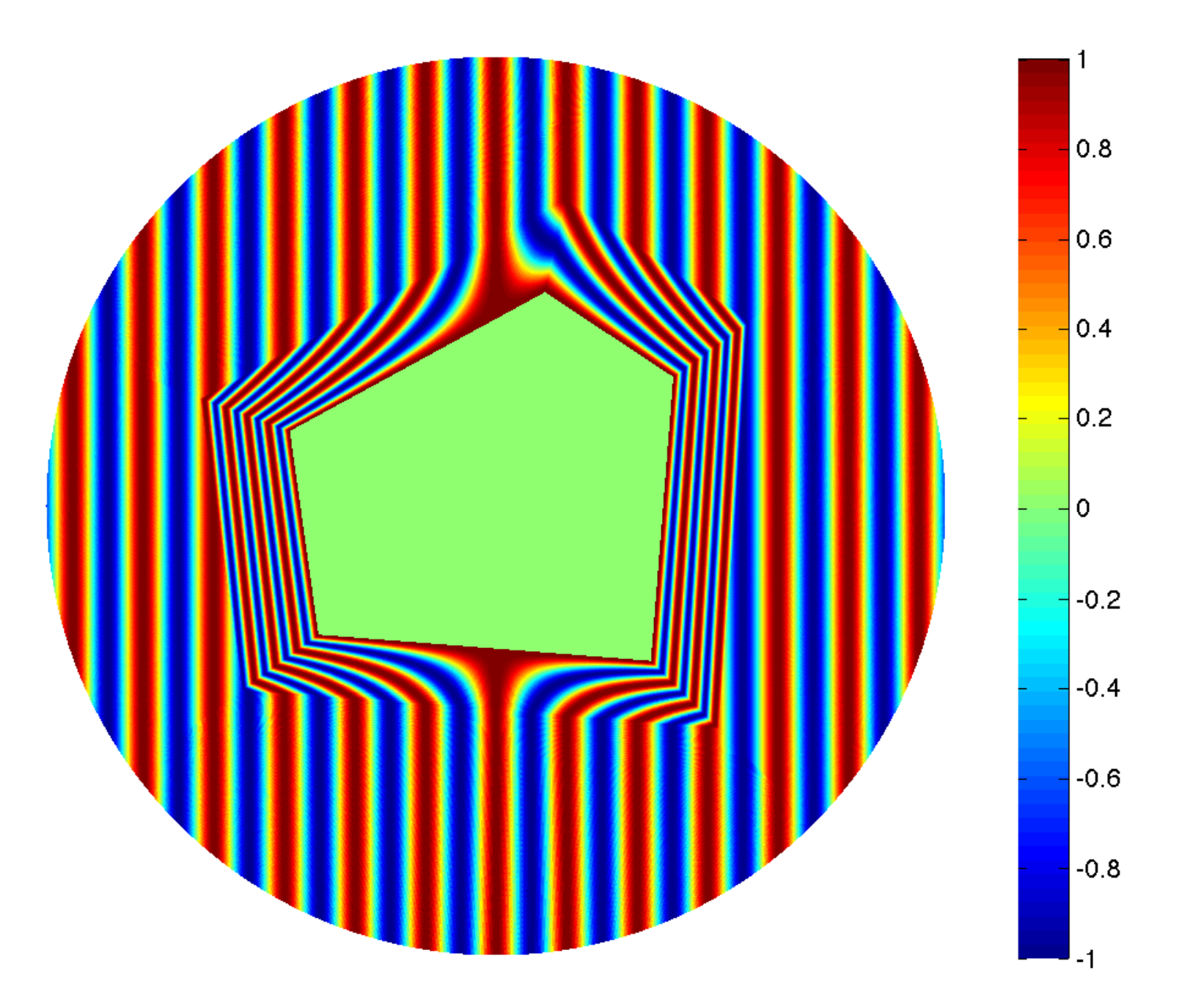}} \hspace*{-10pt}
 \subfigure[Contour (SEM)]{ \includegraphics[scale=.23]{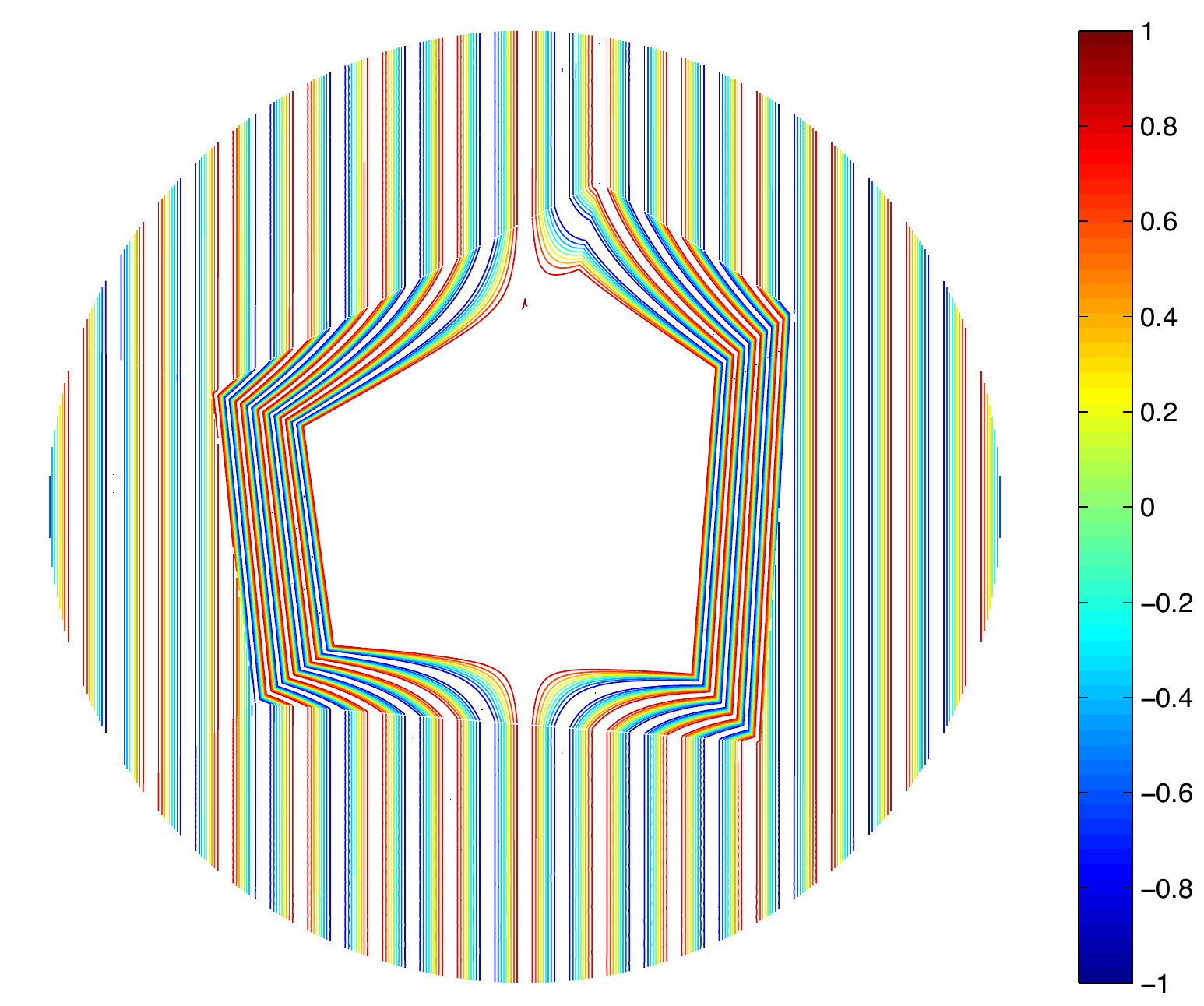}}

\subfigure[Profiles of  the real  and imaginary  parts of the electric field along  $\theta=0$ by SEM]{ \includegraphics[width=0.94\textwidth,height=0.25\textwidth]{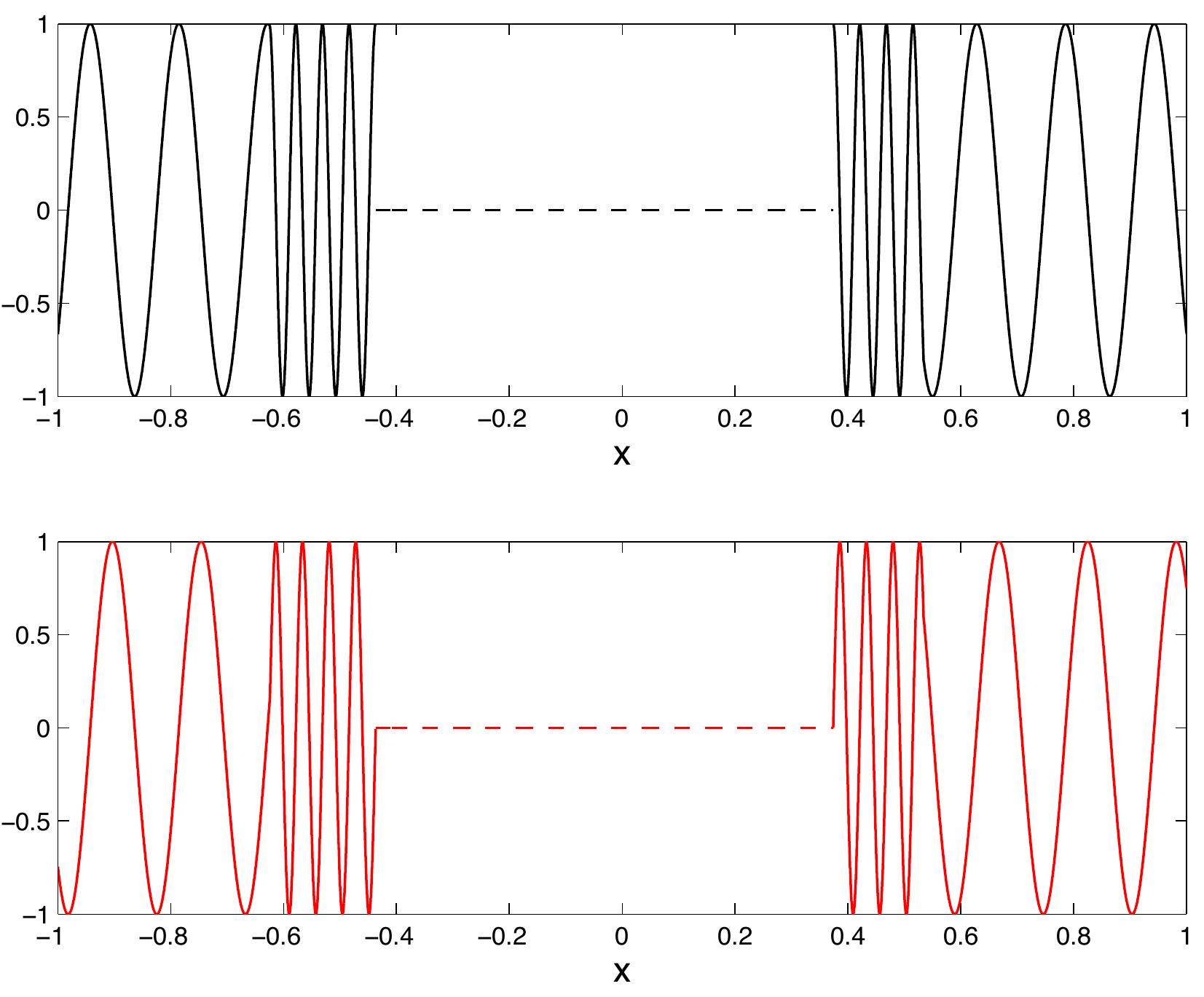}} 
  \caption{\small A comparison study: SEM versus FEM, where $\theta_0=0$ and $k=40.$ In SEM simulation,     $40\times 40$-grid is used for each element with a total DOF: $16,000.$      In FEM simulation,   
    a total DOF:  $726,933$ is used.}
\label{SemA}
\end{figure}
 
%
%

 In Figure \ref{SemA}, we plot the electric field distribution (real part)  for $k=40$  obtained  by  (i) SEM  with $N=40$ in each element and  the degree of freedom (DOF): $16,000$,  and (ii)  FEM  with 
 with a total DOF:  $726,933$ (in order to obtain reasonable results).  Observe that 
the magnitude of the field obtained from FEM is about  $1.05$ (see the colour bar of  (a)), which is expected  to be $1$ as shown in (c). As a result,  the wavefront of the field  in the  contour appears blurred, while that of the SEM is very accurate.   Indeed, the use of exact DtN boundary condition and new CBCs allows us to simulate the   ideal cloak very accurately.  
 
 We further challenge  SEM with higher wavenumber $k=80$ and oblique incident angle $\theta_0=\pi/4$ (see Figure \ref{SemB}).  
We depict in (a) the electric field distribution (real part) with cut-off number $M=80$ and $N=80$ in each element with the same geometric setting in Figure \ref{SemA} (c). Again, the highly oscillatory oblique incident wave is perfectly steered by the cloaking layer and completely shielded from the cloaked region. Apart from plotting
the electric-field distributions, we  also depict the  time-averaged Poynting vector (cf. \cite{orfanidis2002electromagnetic}):
${\bs S}={\rm Re}\,\{{\bs E}\times {\bs H}^*\}/2,$
which indicates the directional  energy flux density.   In (b), we depict the associated Poynting vector fields. We find that the waves are again  steered smoothly around the polygonal cloaked region without reflecting and scattering.

\begin{figure}[htbp]
 {~}\hspace*{-16pt}\subfigure[Electric field]{ \includegraphics[scale=0.26]{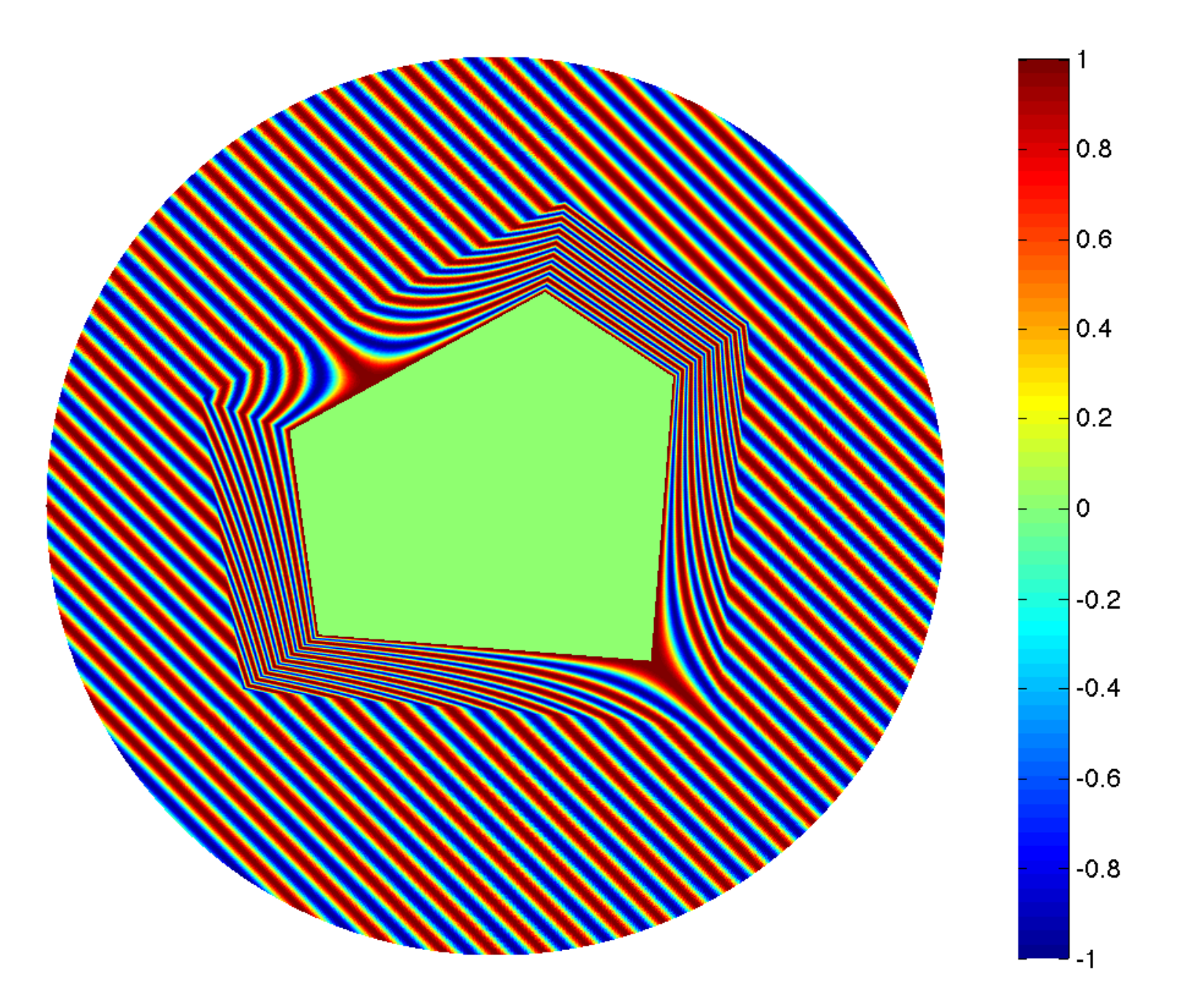}}  \; 
\subfigure[Poynting vector]{\includegraphics[scale=0.27]{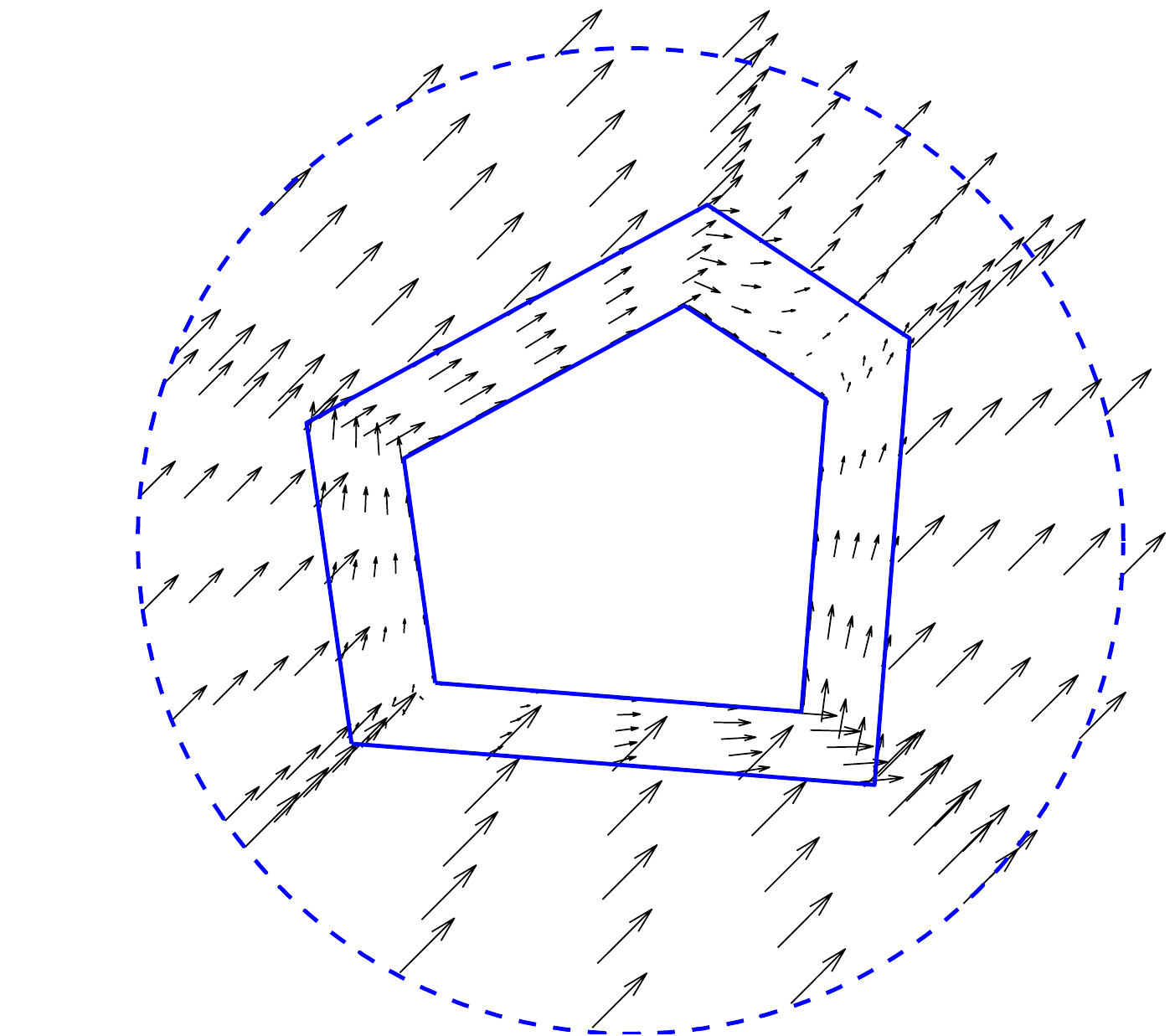}}\quad  
 \subfigure[External source]{ \includegraphics[scale=0.26]{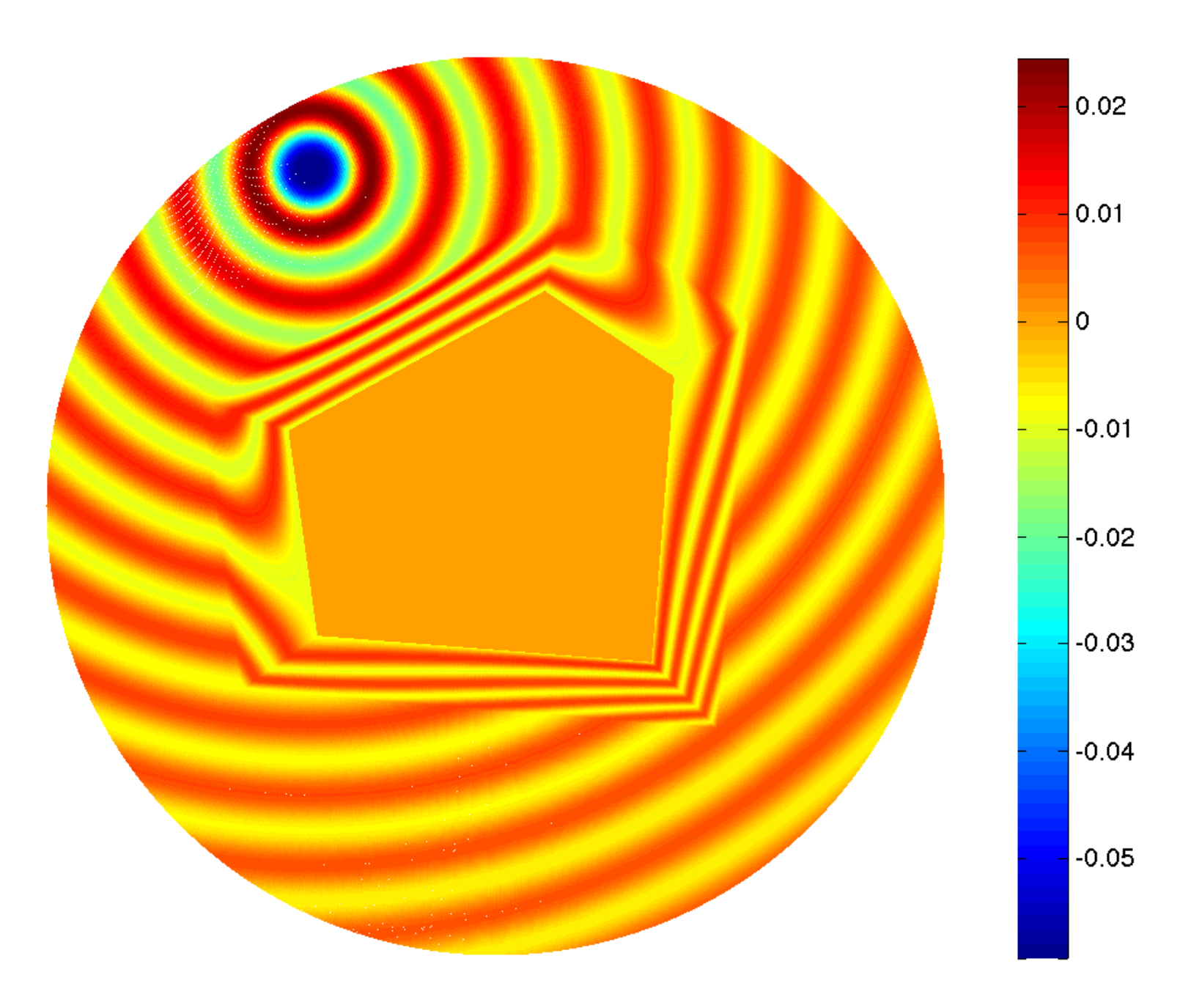}}
   \caption{\small SEM with high frequency wave and external source. (a) Real part of the electric field distribution and (b) the related Poynting vector, where $k=80$, $\theta_0=\pi/4.$ (c) Real part of the electric field  with external source (\ref{source_example1}).   }
\label{SemB}
 \end{figure}

We now add  an external source,  compactly supported in $\Omega_+$ as  the wavemaker, and turn off the incident wave. More precisely, we modify   \eqref{eq1conL} and \eqref{eq2conL} as
\begin{equation}\label{sourceEq}
 \nabla \cdot({\bs C}(\bs r) \nabla u^+(\bs r))+ k^2n(\bs r) u^+(\bs r)=f(\bs r) \;\;  {\rm in}\;\; \Omega_+; \quad\partial_{r} u^+ -{\mathscr T} _{R} [u^+]=0   \quad {\rm at}\;\;  \Gamma_{\!R}.
\end{equation}
In practice, we use the Guassian function in Cartesian coordinates:
\begin{equation}\label{source_example1}
f(\bs r)=\alpha\; {\rm exp}\Big({-\frac{(x-\beta)^2+(y-\kappa)^2}{2 \gamma^2}}\Big),
\end{equation}
where $\alpha,\beta,\kappa,\gamma$ are tuneable  constants. 
To this end, we take $\alpha=100$, $\beta=-0.41$, $\kappa=0.75$ and $\gamma=0.04,$ so the source at $B_{\!R}$ is nearly zero.  The plot of the electric field distributions in Figure \ref{SemB} (c) is computed from  SEM with $k=40$, $M=60$ and $N=40$ in each element.  The anti-plane   waves generated by the source are smoothly bent and the cloak does not produce any scattering.  Observe that the waves seamlessly pass through the outer artificial boundary without any reflecting.

\subsection{Numerical study of effects of defects, lossy media and dispersive media}
\label{newsectBA} 
We next demonstrate that the proposed SEM provides a reliable tool to study the effect of defects and sensitivity to the variation of the media within the cloaking layer.  Interesting  investigation (mostly  from analytic point of view) has been devoted to the circular and spherical cloaks (see, e.g.,   \cite{cummer2006full,zhang07Interaction,okada2012fdtd,argyropoulos2010dispersive,zhang2008rainbow}), but the tools appear non-trivial to be extended to the polygonal cloaks. 


\subsubsection{Defects in the cloaking layer} We consider the influence of defects to the perfect polygonal cloak.
As illustrated in Figure \ref{SemC} (a)-(b),  a rectangular defect with length $a$ and width $b$ is embedded into the cloaking layer.
We set $\bs C=\bs I_2$ and $n=1$ within the defect, so  
the traditional transmission condition \eqref{eq11} can be imposed  at four sides, if the defect is not aligned with the cloaking boundary.  
In Figure \ref{SemC} (a), we depict   the electric field distribution with defect $a=b=0.06$ obtained by the proposed SEM with $k=40$, $\theta_0=0$, $M=60$ and $N=40.$  Observe that even with such a small defect, the electric field distribution is apparently disturbed, especially for the forward-scattering region, and the magnitude increases approximately  up to $1.5$ (note: it is $1$ for the perfect cloak). Also notice that in the back-scattering region, the waves  appear not significantly affected, so the cloaking effects seem still good.  In Figure \ref{SemC} (b),  we enlarge  the defect and set $a=0.24,$ $b=0.06.$
The field in both the back and forward scattering regions is deteriorated more.

\begin{figure}[h!]
{~}\hspace*{-16pt} \subfigure[Defect with $a=b=0.06$]{ \includegraphics[scale=.32]{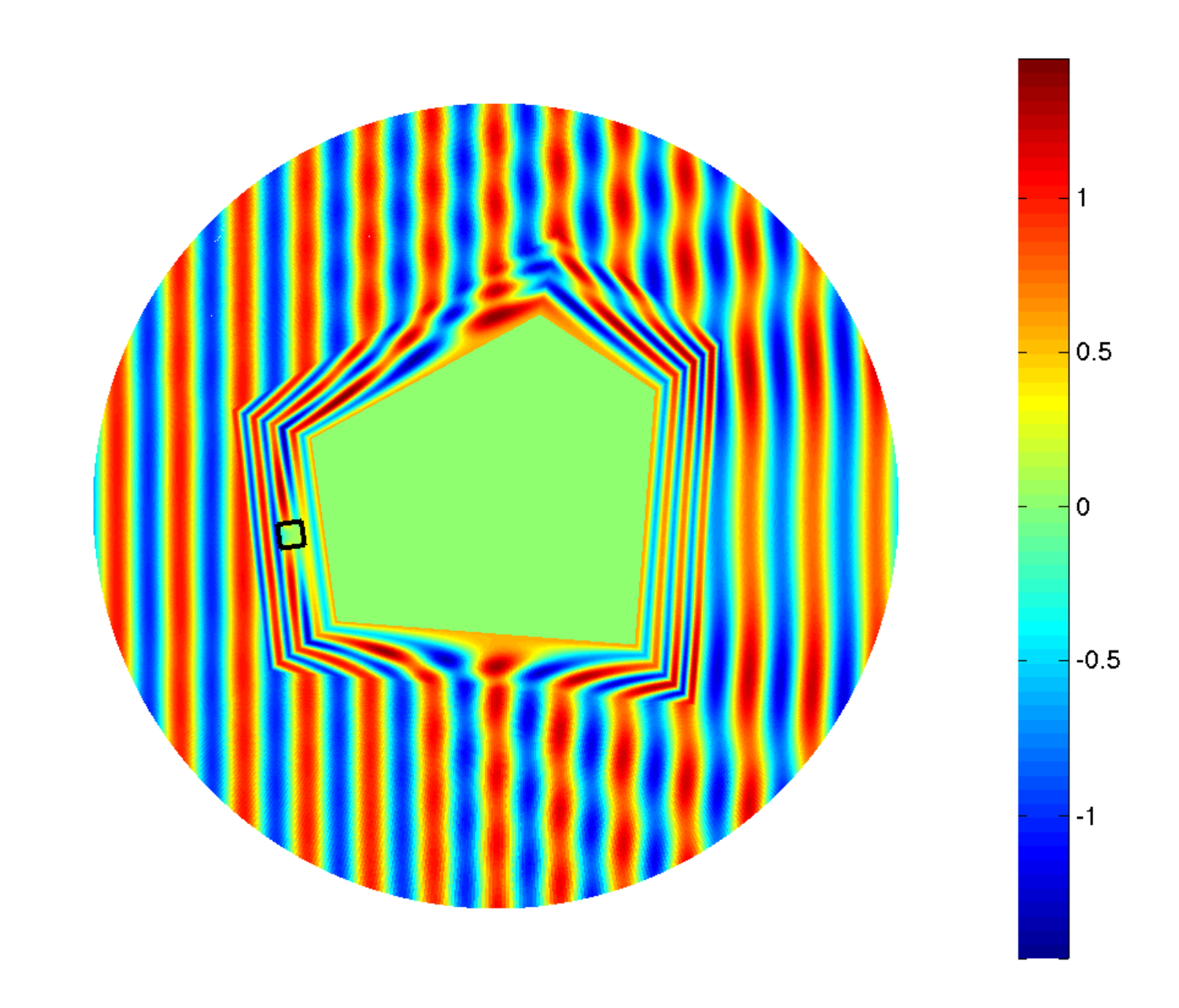}} \hspace*{-10pt}
\subfigure[Defect with $a=0.24,\,b=0.06$]{\includegraphics[scale=.33]{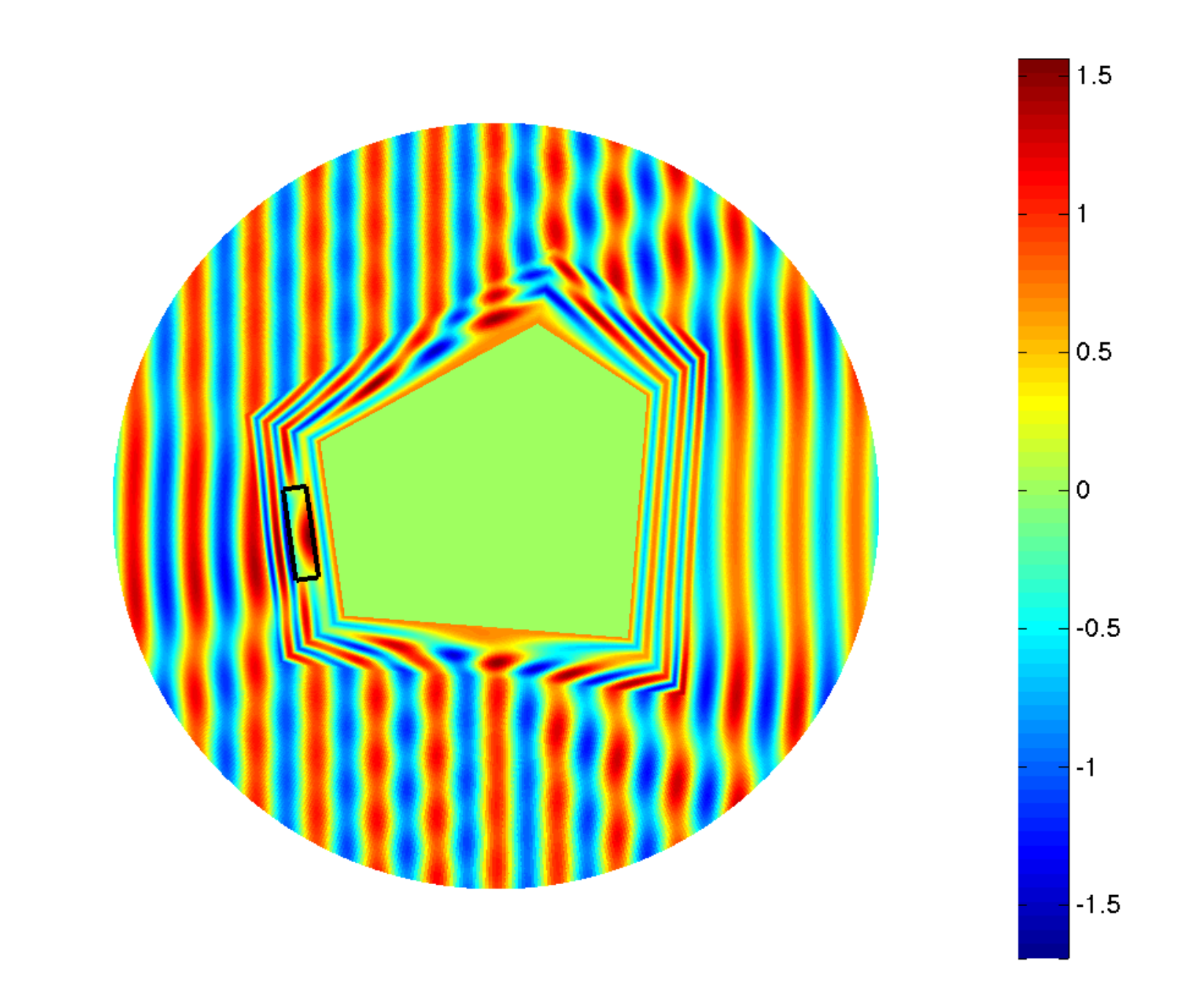}} \hspace*{-10pt} 
 \subfigure[Loss tangent $\tan \delta=0.01$]{ \includegraphics[scale=.30]{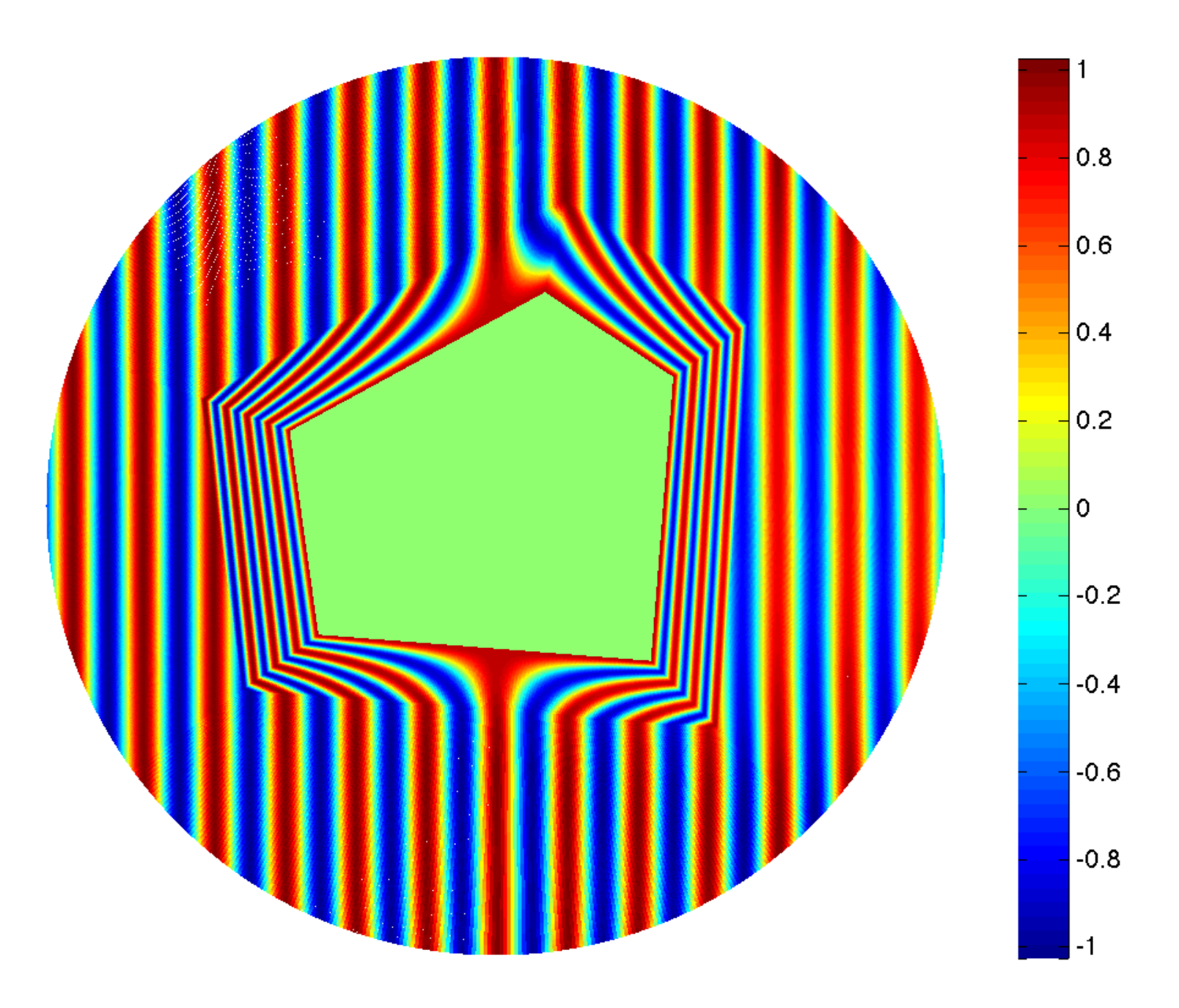}} \hspace*{-10pt}
 \subfigure[Loss tangent $\tan \delta=0.05$]{ \includegraphics[scale=.30]{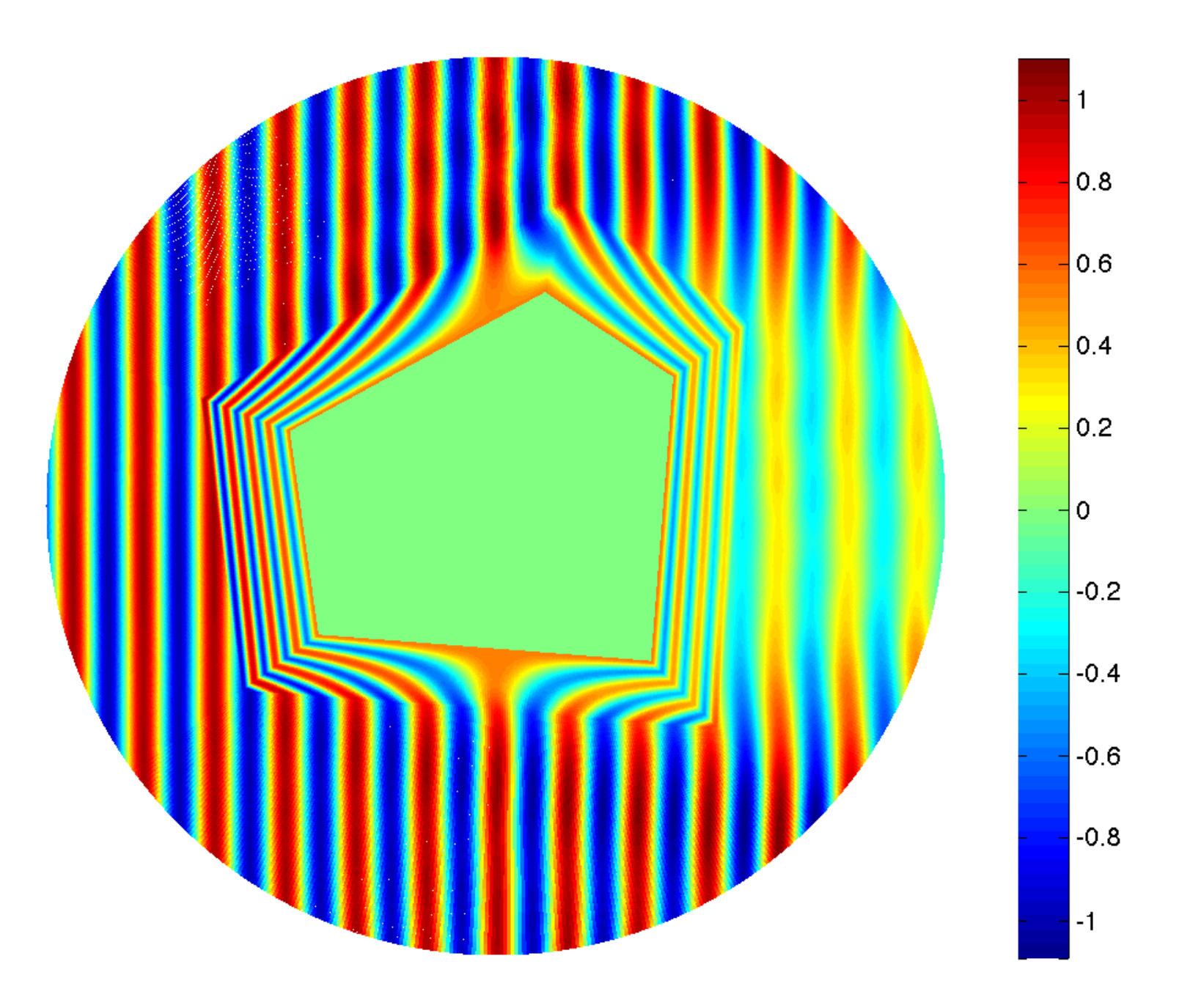}}
 \subfigure[Dispersion with $k=39$]{ \includegraphics[scale=.30]{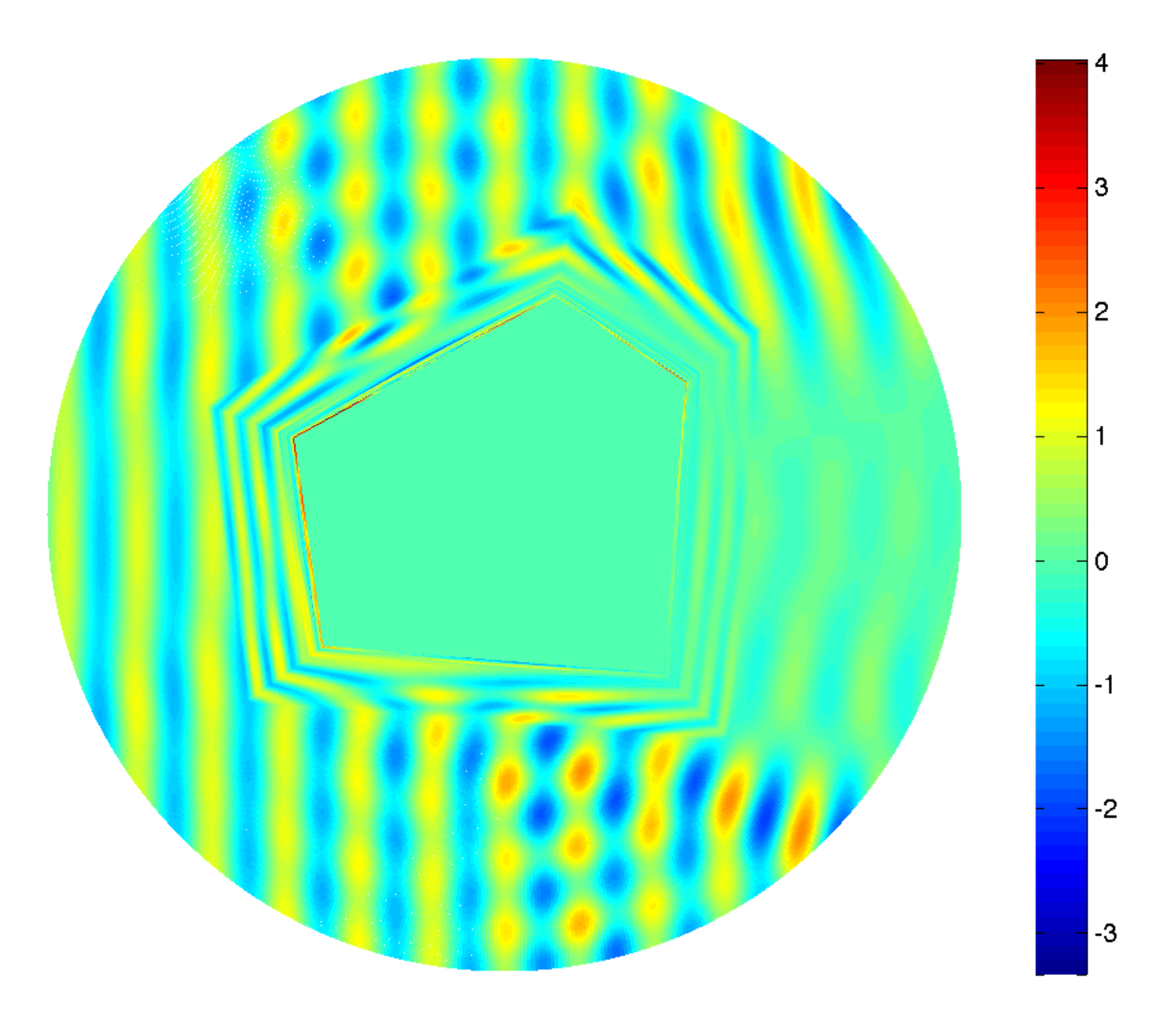}} \hspace*{-10pt}
 \subfigure[Dispersion with $k=41$]{ \includegraphics[scale=.30]{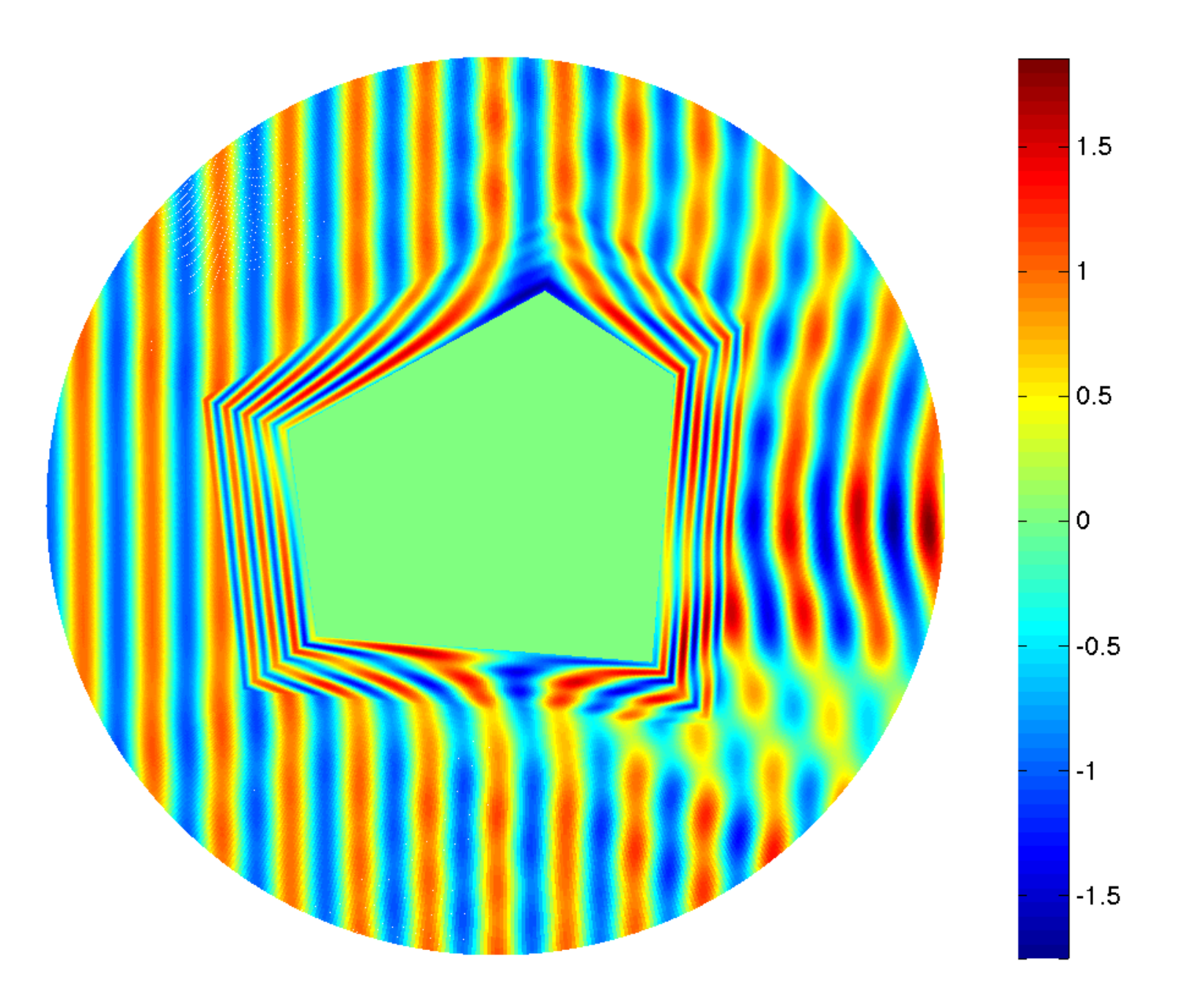}}
  \caption{\small The effects of defects, loss and dispersion on the polygonal invisibility cloak.  
Defects:  real part of the electric field distributions with homogeneous defects with sizes (a) $a=b=0.06$ and (b) $a=0.24$, $b=0.06$. Loss: real part of the electric field distributions with loss tangents (c) $\tan \delta=0.01$ and (d) $\tan \delta=0.05$. Dispersion: fix $k_c=40$, the real part of the electric field distributions with wave numbers (e) $k=39$ and (f) $k=41$.}
\label{SemC}
 \end{figure}

\subsubsection{Lossy media in the cloaking layer} Similar to the setting in \cite{cummer2006full,zhang07Interaction}   for the circular and spherical cloaks,  we replace the media in cloaking layer by a electric-lossy medium  (cf. \cite{orfanidis2002electromagnetic}). 
More precisely, the real electric permittivity $\bs \epsilon$ in \eqref{newMaxwell} is replaced by  a complex electric permittivity  
$(1+{\rm i} \tan \delta)\bs \epsilon,$   
where $\tan \delta$ is a tunable constant termed as loss tangent to quantify the absorptive property of the medium.  
Note that this replacement only brings about the modification of the two-dimensional Helmholtz equation \eqref{HelmMeta} as  
$$\nabla \cdot ({\bs C}(\bs r)\;\nabla u(\bs r))+k ^2(1+\ri \tan \delta) n(\bs r)\, u(\bs r)=0.$$
In Figure \ref{SemC}, we depict the electric field distributions with loss tangent   $\tan \delta=0.01, 0.05$ 
in (c) and (d), respectively, where we take  $k=40$, $\theta_0=0$, $M=60$ and $N=40$ in each element. Observe that
in the first case,  the effect of the loss is almost imperceptible. As we enlarge the loss tangent to $0.05$ in (d), 
 the cloaking effect appears  good in the backscattering region but is apparently deteriorated in the forward-scattering region, which  is inevitable because the lossy medium absorbs the forward-travelling wave power.  We point out that 
similar phenomena were observed for   the circular  and spherical cloaks in \cite{cummer2006full,zhang07Interaction}.

\subsubsection{Drude model and dispersive media in the cloaking layer}  
Based on the form-invariant coordinate transformation,   the polygonal cloak can perfectly conceal arbitrary objects inside the interior polygonal  domain. 
However, as an  ideal cloak, the material parameters are dispersive, and perfect invisibility can only be achieved for a single 
 frequency, known as the ``cloaking frequency"  (cf. \cite{pendry.2006,zhang2008rainbow,chen2007extending}). 
 It is of much physical relevance to study the response of an ideal cloak to a non-monochromatic  
 electromagnetic wave passing  through such a dispersive cloak.  The investigation along this line has been very limited to mostly 
 analytic treatments of circular and spherical cloaks  (cf.  \cite{argyropoulos2010dispersive, zhang2008rainbow}). 
We demonstrate that the proposed SEM offers an accurate means to 
 understand some interesting phenomena of a nonmonochromatic wave interacting with a polygonal cloak.  

Following the procedure in \cite{okada2012fdtd}, 
we start with  diagonalizing  the symmetric matrices $\bs \epsilon$ and  $\bs \mu$  in  \eqref{parameter2}-\eqref{parameter3}, i.e., 
 \begin{equation} \label{diagonalize}
\bs \epsilon=\bs \mu ={\bs P} {\bs \Lambda} {\bs P}^t,\quad \bs \Lambda={\rm diag}(\lambda_1,\lambda_2,\lambda_3),
\end{equation}
where $ {\bs P}=(P_{ij})_{1\leq i,j \leq 3}$ is an orthonormal matrix (with $P_{j3}=P_{3j}=0$ for $j=1,2,$ and $P_{33}=1$), and
 the eigenvalues are 
\begin{equation}\label{eigenM}
\lambda_1=\frac{C_{11}+C_{22}+\sqrt{(C_{11}+C_{22})^2-4}}{2},\quad \lambda_2=\frac{C_{11}+C_{22}-\sqrt{(C_{11}+C_{22})^2-4}}{2},\quad \lambda_3=n.
\end{equation}
From \eqref{ck11}-\eqref{ck22},  we have 
\begin{equation}\label{Cineq}
C_{11}+C_{22}=\frac{r-R_1}{r}+\frac{1}{r(r-R_1)}\Big( r^2+\Big(\frac{dR_1}{d\theta}\Big)^2    \Big)\geq \frac{r-R_1}{r}+\frac{r}{r-R_1}\geq 2,
\end{equation}
which implies  $\lambda_1>1.$ 
However,  $\lambda_2$ and $\lambda_3$ are less than $1$ for some    $r\in (R_1, R_2)$.  
Based on the principle in \cite{okada2012fdtd,zhang2008rainbow},  we modify  $\lambda_2$ and $\lambda_3$ by 
using the  Drude model (cf. \cite{orfanidis2002electromagnetic}).  More precisely,  let $\omega_c>0$ be the ``cloaking frequency", 
and  define 
\begin{equation}\label{drudemodel}
\tilde \lambda_i:=\tilde \lambda_i(\bs r, \omega)=1-\frac{\omega_{p,i}^2}{\omega(\omega+\ri \gamma_i)},  \;\;\; {\rm with}\;\;\; 
\omega_{p,i}^2:=\omega_c(\omega_c+{\rm i}  \gamma_i)(1-\lambda_i),\quad i=2,3,
\end{equation}
where $\{\gamma_i\}_{i=2}^3$  are given collision frequencies,  and $\{\omega_{p,i}\}_{i=2}^3$ are known as  the plasma frequencies.  For notational convenience, we define 
\begin{equation}\label{collision}
\beta_i:=\frac{\omega_c(\omega_c+{\rm i}  \gamma_i)}{\omega(\omega+{\rm i}  \gamma_i)}, \;\;\;  {\rm so}\;\;\; 
\tilde \lambda_i =1+\beta_i (\lambda_i-1),\quad i=2,3.
\end{equation}

Denoting $ \widetilde{\bs \Lambda}={\rm diag}(\tilde \lambda_1,\tilde \lambda_2,\tilde \lambda_3)$ with  $\tilde \lambda_1=\lambda_1,$
we then replace the  material parameters $\bs \epsilon$ and $\bs \mu$ in \eqref{newMaxwell}, respectively, by  
\begin{equation}\label{tildeepsimu}
\tilde{ \bs \epsilon}=\tilde {\bs \mu}={\bs P} \widetilde{\bs \Lambda} {\bs P}^t=
\begin{bmatrix}
\widetilde{\bs C} & \bs 0^t\\[2pt]
               \bs 0      & \tilde n
\end{bmatrix},
\end{equation}
where by a direct calculation,  we have 
\begin{equation}\label{tildeC}
\widetilde{ \bs C}=
\begin{bmatrix}
\widetilde C_{11} & \widetilde C_{12}\\[2pt]
\widetilde C_{12} & \widetilde C_{22}
\end{bmatrix}
=
\bs C+(1-\beta_2)(1-\lambda_2)
\begin{bmatrix}
P_{12}^2 & P_{12}P_{22}\\[2pt]
P_{12}P_{22}& P_{22}^2 
\end{bmatrix},
\end{equation}
and
\begin{equation}\label{dispersiven}
\tilde n=1+\beta_3 (\lambda_3-1)=\tilde \lambda_3.
\end{equation}
One verifies readily from \eqref{tildeepsimu} that 
\begin{equation}\label{dettildeC0}
{\rm det}(\tilde{ \bs \epsilon})={\rm det}(\tilde {\bs \mu})=\tilde\lambda_1 \tilde\lambda_2\tilde\lambda_3= \lambda_1 \tilde\lambda_2\,\tilde n=
{\rm det}(\widetilde {\bs C})\,\tilde n.
\end{equation}
Thus, using the fact $\lambda_1 \lambda_2=1$ (cf. \eqref{eigenM}), we obtain from  \eqref{collision}  and \eqref{dettildeC0} that 
\begin{equation}\label{dettildeC}
{\rm det}(\widetilde {\bs C})= \lambda_1\tilde \lambda_2=\lambda_1(\beta_2\lambda_2+1-\beta_2)=\beta_2+(1-\beta_2)\lambda_1.
\end{equation}
Accordingly, we find that the counterpart of \eqref{parameter2s} becomes 
\begin{equation}\label{muepsinv}
\tilde {\bs \mu}^{-1}=\tilde {\bs \epsilon}^{-1}
=\begin{bmatrix}
\widehat C_{22} & -\widehat C_{12}& 0\\[1pt]
-\widehat C_{12} & \widehat C_{11}&0\\[1pt]
0& 0 & \tilde n^{-1}
\end{bmatrix},\;\;\; {\rm where}\;\;\;  \widehat C_{ij}= \frac {\widetilde C_{ij}} {\beta_2+(1-\beta_2)\lambda_1}, 
\end{equation}
for $i,j=1,2.$ Using \eqref{magnetic} with $\widehat C_{ij}$ and $\tilde n$ in place of $C_{ij}$ and $n$,  we obtain the new model defined in the cloaking layer: 
\begin{equation} \label{HelmMetahat}
\nabla \cdot \big(\widehat {\bs C}(\bs r, \omega)  \;\nabla u(\bs r)\big)+k ^2\,\tilde  n(\bs r, \omega)\, u(\bs r)=0,
\end{equation}
where $ \widehat {\bs C}=(\widehat C_{ij})_{1\le i,j\le 2},$ and  $k=\omega \sqrt{\epsilon_0 \mu_0}$ as before. 
\begin{rem}\label{conditionA} Observe from \eqref{collision} that if $\omega=\omega_c,$ then $\beta_i=1$ and $\lambda_i=\tilde \lambda_i$ for $i=1,2.$ Thus, in this case,  \eqref{HelmMetahat} reduces to \eqref{HelmMeta} and  $\widehat {\bs C}(\bs r, \omega_c)$ is singular at the cloaking boundary $r=R_1$. 
However,  if $\omega\not =\omega_c$ (so $\beta_2\not =1$), then   $\widehat {\bs C}(\bs r, \omega)$ becomes regular at $r=R_1.$ Indeed, 
by  \eqref{muepsinv}, 
\begin{equation*}
\widehat C_{ij}= \frac {(r-R_1)\widetilde C_{ij}} {\beta_2(r-R_1)+(1-\beta_2)(r-R_1)\lambda_1}.
\end{equation*}
In fact, one can verify that if $\beta_2\not =1,$ 
$$\lim_{r\to R_1} (r-R_1)\big\{\bs C, \lambda_1,\lambda_2\big\} \;\; \text{all exist}.$$  
Thus, we can claim from \eqref{tildeC} and the above  that $\widehat {\bs C}(\bs r, \omega)$ is well-defined at $r=R_1.$ In view of this,  the CBCs can not be applied. Here, we follow \cite{argyropoulos2010dispersive,zhang2008rainbow}   and  impose a PMC shell instead.   \qed 
\end{rem}

In the computation, we  take  $\omega_c=k_c/\sqrt{\epsilon_0 \mu_0}$ with  $k_c=40,$ and $\gamma_i/\sqrt{\epsilon_0 \mu_0}=0.0001$ for  $i=2,3$. 
In Figure \ref{SemC} (e)-(f), we plot the electric field distributions with $k=39$ and $k=41$ illuminated by plane wave in \eqref{incident1} with incident angle $\theta_0=0$ and the cut-off number $M=60$ and $N=45$ in each element.
In contrast with Figure \ref{SemA} (c) (where perfect cloaking effect can be obtained for $k_c=40$), 
 we observe from  Figure \ref{SemC} that  the electric field distributions are affected  and distorted in both cases (i.e., $k=k_c\pm 1$),  
 in particular, more severely when  $k<k_c.$ 
 Indeed, similar to the phenomena observed  in  \cite{zhang2008rainbow,argyropoulos2010dispersive} for circular and spherical cloaks,   the incident wave with frequency slightly deviated below  $k_c$, the field after the wave passes the cloak is dissipated and a large shadow appears in the forward scattering region. While for the incident wave with frequency slightly deviated above $k_c$, the field in  most part of the cloaking layer does not change much,  except for  the part close to the cloaking boundary $r=R_1$,  and the field behind the cloak is reinforced. This can be regarded as  the frequency shift effect as in    \cite{zhang2008rainbow}. 
\section{Accurate simulation of    electromagnetic concentrators  and rotators}\label{sect:App}

%
%

In this section, we further apply the efficient spectral-element solver to accurately simulate the electromagnetic concentrators and rotators. 

\subsection{Polygonal concentrators}\label{sect:concentrator}
The electromagnetic concentrator aims at intensifying electromagnetic waves in a certain region, which play an important role in the harnessing of light in solar cells or similar devices, where high field intensities are needed.  

Here, we are interested in the polygonal concentrator with a  configuration similar to  the polygonal cloak  
illustrated in Figure \ref{concenfig2} (b),  where  EM waves
are expected to be 
concentrated  in the interior convex polygonal region $\Omega_-^{p}$. It is accomplished by a coordinate transformation that maps the ``polygonal annulus"    in Figure \ref{concenfig2} (a) to the ``polygonal annulus"    in Figure \ref{concenfig2} (b), where the interior portion of the latter has  larger area.  More precisely, the polygonal concentrator is mapped from the same structure in Figure \ref{concenfig2} (b) but with a different ratio (see Figure \ref{concenfig2} (a)): 
 \begin{equation}\label{radiconst2}
 \breve \rho=\frac{OA_{o}}{OA}=\frac{OB_{o}}{OB}=\cdots, \quad 0<\rho <\breve \rho<1,
\end{equation}
where $\rho$ is defined in \eqref{radiconst}. Then the ratio $\rho/ \breve \rho$ is known as the rate of concentration. For notational convenience, we define
\begin{equation}\label{taudefine}
\varrho:=1-\dfrac{1-\rho}{1-\breve \rho}.
\end{equation}

\begin{figure}[htbp]
 \subfigure[$(\breve x,\breve y)$-domain]{ \includegraphics[scale=.27]{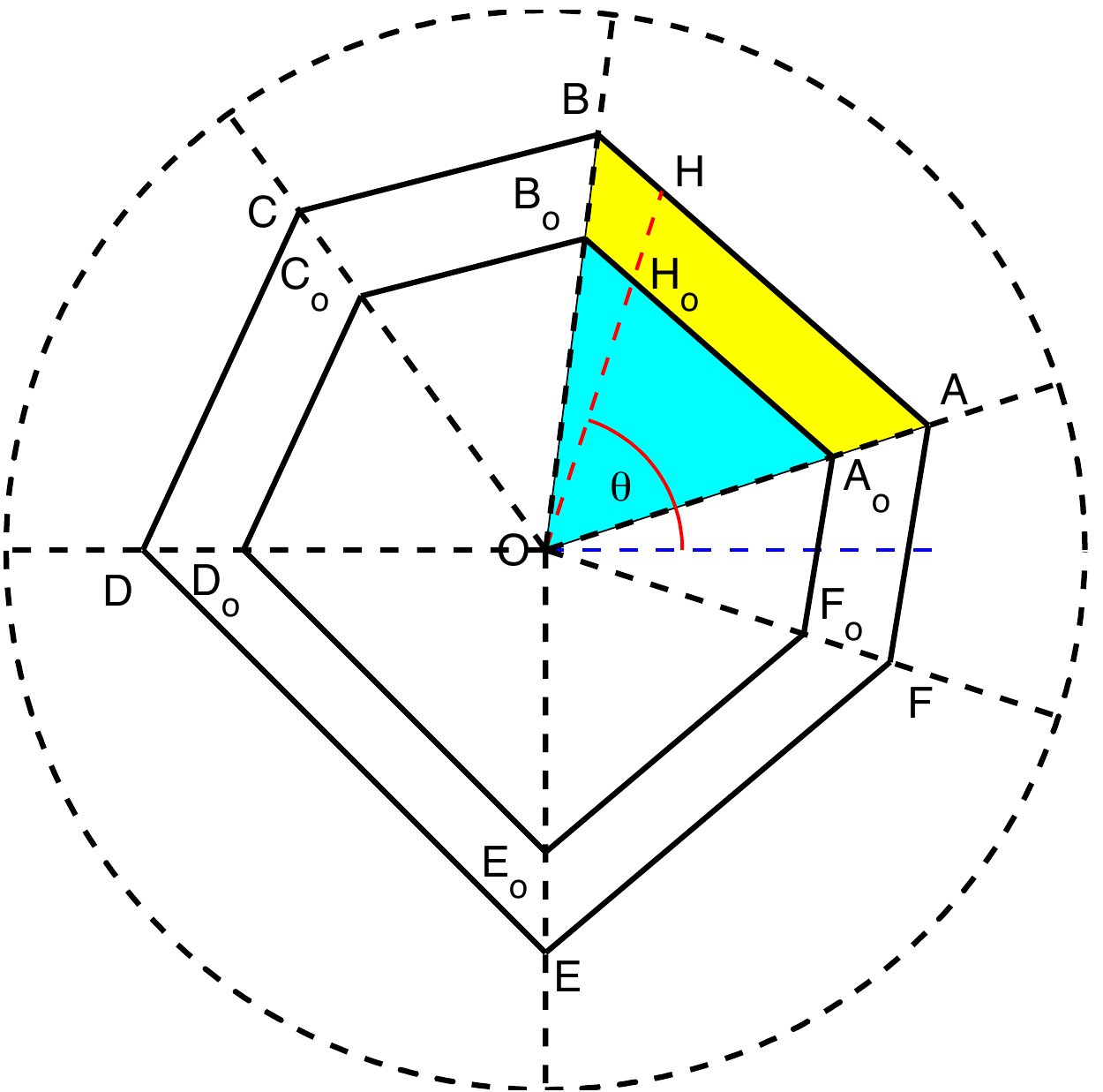}} 
 \subfigure[$(x, y)$-domain]{ \includegraphics[scale=.27]{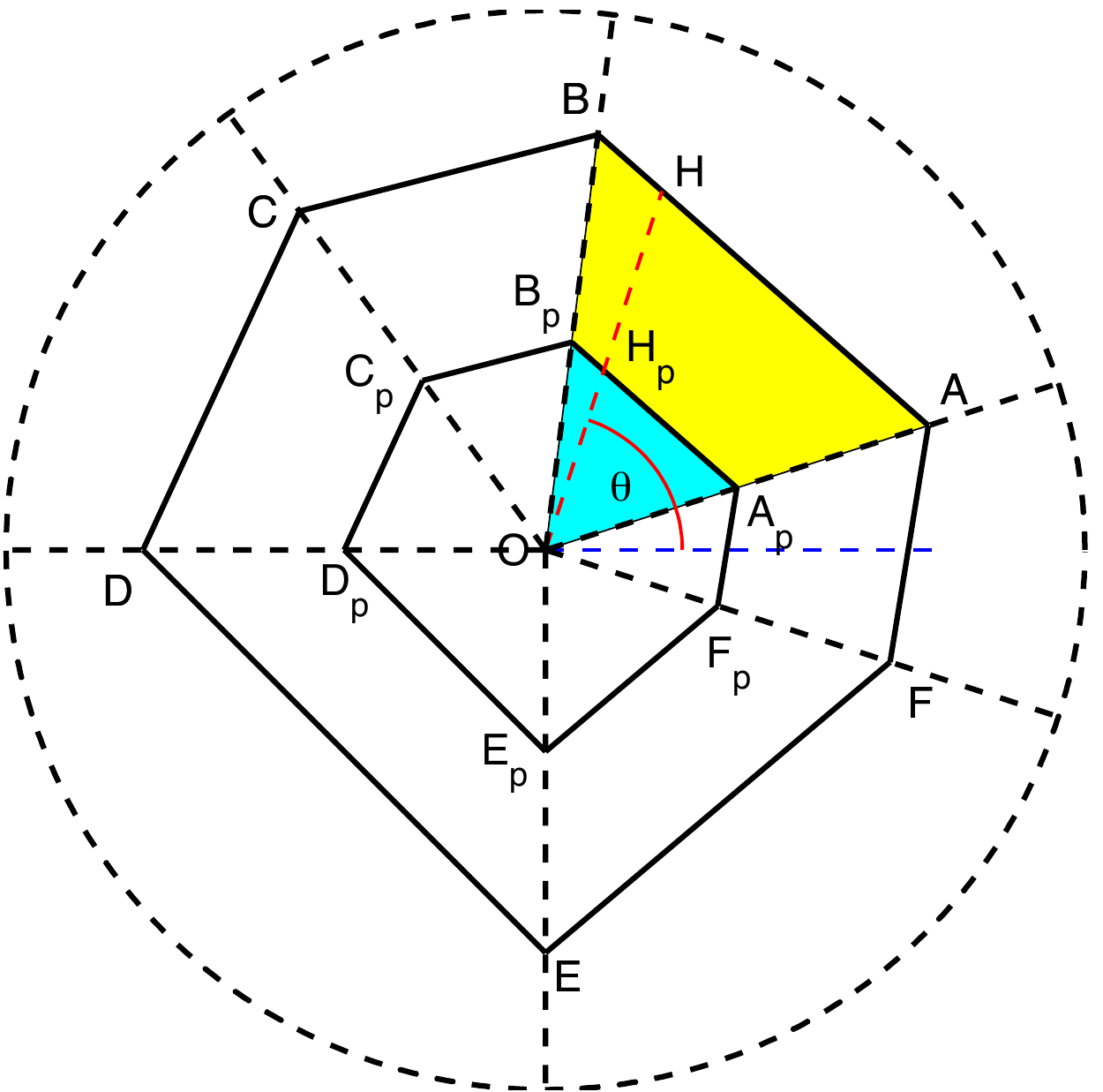}}
  \subfigure[Circular rotator]{ \includegraphics[scale=.26]{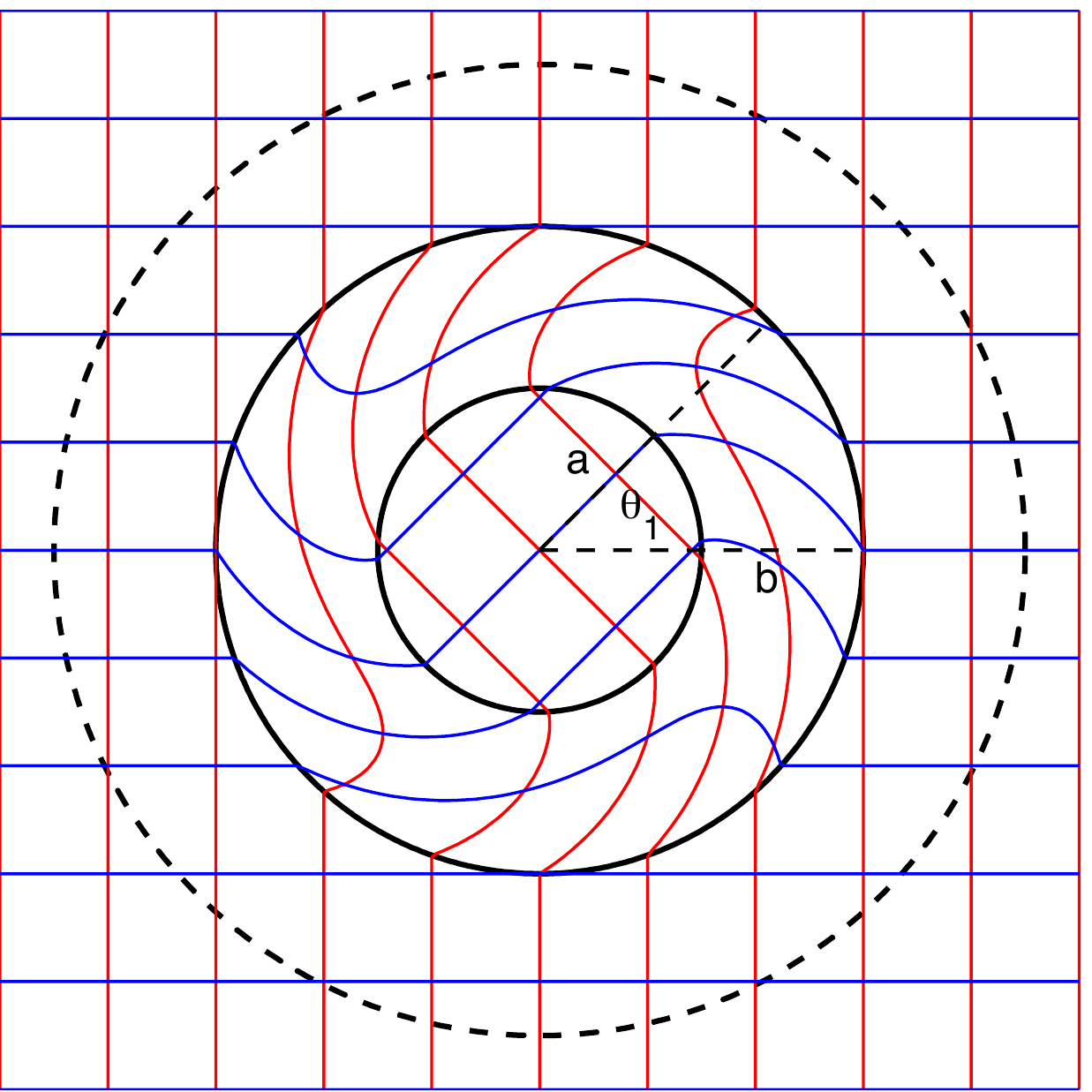}} 
 \subfigure[Partition of $B_R$]{ \includegraphics[scale=.26]{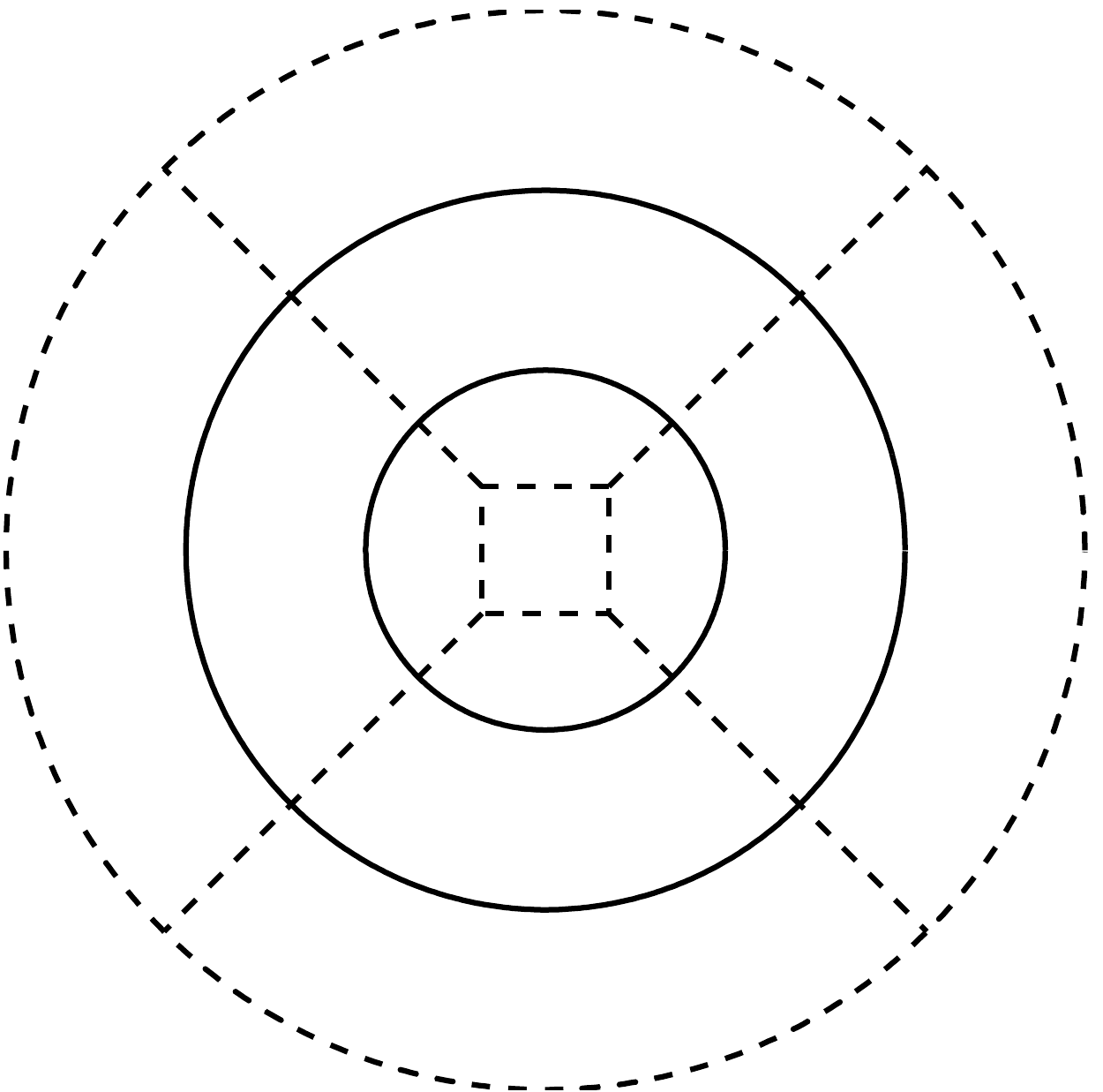}}
  \caption{\small Schematic geometry of a   polygonal  concentrator and a circular rotator. (a)  The polygonal domain in the original coordinates $(\breve x, \breve y).$ (b) Through the coordinate transformation (\ref{concen-tran}), the polygonal domain $\Omega_-^{o}=A_oB_o\cdots F_o$ is compressed into the polygonal domain $\Omega_-^{p}$ that forms the concentration region. Consequently, the original polygonal annulus domain $\Omega_- \setminus \bar \Omega_-^{o}$ in (a) is expanded into the polygonal annulus 
  $\Omega_-^a$. (c) Through the coordinate transformation \eqref{rotat-tran}, points in the circular annulus $a<r<b$ are rotated with a fixed angle $\theta_1$. (d) The computational mesh for the circular rotator. }
\label{concenfig2}
\end{figure}

The corresponding coordinate transformation takes the form (see, e.g., \cite{jiang2008design}):  

\vskip 4pt 
\begin{itemize}
\item[(i)] In  $\Omega_-$,  
\begin{equation}\label{concen-tran}
\begin{cases}
r=\dfrac {\rho} {\breve \rho} \,\breve r,\quad   & \breve r\in [0,\breve R_1],\;\; r\in [0, R_1],\\[8pt]
r=(1-\varrho)\,\breve r+ \varrho \, R_2,
\quad & \breve r\in [\breve R_1, R_2],\;\; r\in [R_1,R_2],
\end{cases} 
\end{equation}
and $\theta=\breve \theta\in [0,2\pi);$ 
\vskip 7pt 
\item[(ii)] In  $\Omega_+=B_R\setminus \bar \Omega_-$, the transformation is identity: $r=\breve r, \; \theta=\breve\theta.$
\end{itemize}
\vskip 5pt 
Here, $R_i,\; i=1,2,$ are the same as in \eqref{hR1} and $\breve R_1=\breve \rho/ \rho  R_1.$  Using \eqref{parameter3}, we can derive the coefficients $\bs C$ and $n$ as follows (see Appendix \ref{AppendixA}):
\vskip 5pt
\begin{itemize}
\item[(i)$_a$] In  $\Omega_-^{\rm p}$, 
\begin{equation}\label{cn1}
 {\bs C}={\bs I}_2,\quad  n={\breve \rho^2}/{\rho^2}\,.
\end{equation}
\vskip 4pt 
\item[(i)$_b$] In  $\Omega_-^{\rm a}$, 
 \begin{align}
C_{11}&=\frac{r-\varrho R_2}{r}\,\frac{x^2}{r^2}+\frac{\varrho}{r-\varrho R_2} \bigg(\frac{\varrho }{r}\frac{{\rm d} R_2}{{\rm d}\theta}\frac{x^2}{r^2}- \frac {2xy}{r^2} \bigg)  \frac{{\rm d} R_2}{{\rm d}\theta}+
\frac{r}{r-\varrho R_2}\,\frac{y^2}{r^2},\label{cn2a}\\
C_{22}&=\frac{r-\varrho R_2}{r}\,\frac{y^2}{r^2}+\frac{\varrho}{r-\varrho R_2} \bigg(\frac{\varrho}{r}\frac{{\rm d} R_2}{{\rm d}\theta}\frac{y^2}{r^2}+ \frac {2xy}{r^2} \bigg)  \frac{{\rm d} R_2}{{\rm d}\theta}+
\frac{r}{r-\varrho R_2}\,\frac{x^2}{r^2},\label{cn2b} \\
C_{12}&=\bigg( \frac{r-\varrho R_2}{r}- \frac{r}{r-\varrho R_2}  \bigg)\,\frac{xy}{r^2} +  \frac{\varrho}{r-\varrho R_2} \bigg( \frac{x^2}{r^2}+\frac{\varrho}{r}\frac{{\rm d} R_2}{{\rm d}\theta}\frac{xy}{r^2}-\frac{y^2}{r^2}  \bigg) \frac{{\rm d} R_2}{{\rm d}\theta}, \label{cn2c}
\end{align}
and 
\begin{equation}\label{ncase2}
n=\frac{r-\varrho R_2}{r (1-\varrho)^2}\,.  
\end{equation}
\vskip 5pt 
\item[(ii)] In  $\Omega_+$, we have  $\bs C=\bs I_2$ and $n=1.$ 
\end{itemize}
\vskip 5pt

In summary, the governing equation for the polygonal concentrator reads 
\begin{align}
& \nabla \cdot({\bs C}(\bs r)\, \nabla u(\bs r))+ k^2n(\bs r) u(\bs r)=0 \quad  {\rm in}\;\; B_R, \label{eq1con}   \\
& \llbracket u \rrbracket=\llbracket \bs C\, \nabla u \rrbracket=0\quad {\rm at}\;\; \Gamma^p \cup \Gamma^a, \label{eq2con0}\\
& \partial_{r} u-{\mathscr T} _{R} [u]=h   \quad {\rm at}\;\; \partial B_R.  \label{eq2con}
\end{align}
Note  that in the interior polygon,  \eqref{eq1con}  becomes  the Helmholtz equation: 
\begin{equation}\label{omegapa}
 \Delta u+\frac {\breve \rho^2}{\rho^2} k^2 u=0 \quad {\rm in}\;\;   \Omega_-^{p},
\end{equation}
where the ratio ${\breve \rho^2}/{\rho^2}>1$ is a constant. 
Therefore, the coordinate transformation enlarges  the wavenumber $k$ that produces the effect of concentration.   

We can  implement the spectral-element solver   based on the partition of the computational domain 
as with  the previous application. However, different from the previous case, the interior region is part of the computational domain, where  a normal transmission condition is imposed along its boundary  (see \eqref{eq2con0}).    Below, we provide some numerical results with the setting:  a square concentrator centred at the origin with length of each side $1.2$  and the parameters:  $\rho=1/3,\; \breve \rho=2/3$ in \eqref{radiconst} and \eqref{radiconst2} and $R=1.0.$ 
We set  the cut-off number $M=60$ in the DtN operator.  In Figure \ref{SemD}, we depict the electric field distributions and the associated time averaged Poynting vectors illuminated with different incident angles  ((a)-(b): $\theta_0=0$ and (c)-(d): $\theta_0=\pi/4$) with $k=40$ and the grid $N=50$ in each element.  It can be seen that the electric field and energy flux are smoothly concentrated into the inner concentration region $\Omega_-^p$, and the field outside is not affected regardless of the incident angle  on the concentrator.


\begin{figure}[htbp]
 {~}\hspace*{-16pt}\subfigure[Electric field with $\theta_0=0$]{ \includegraphics[scale=0.24]{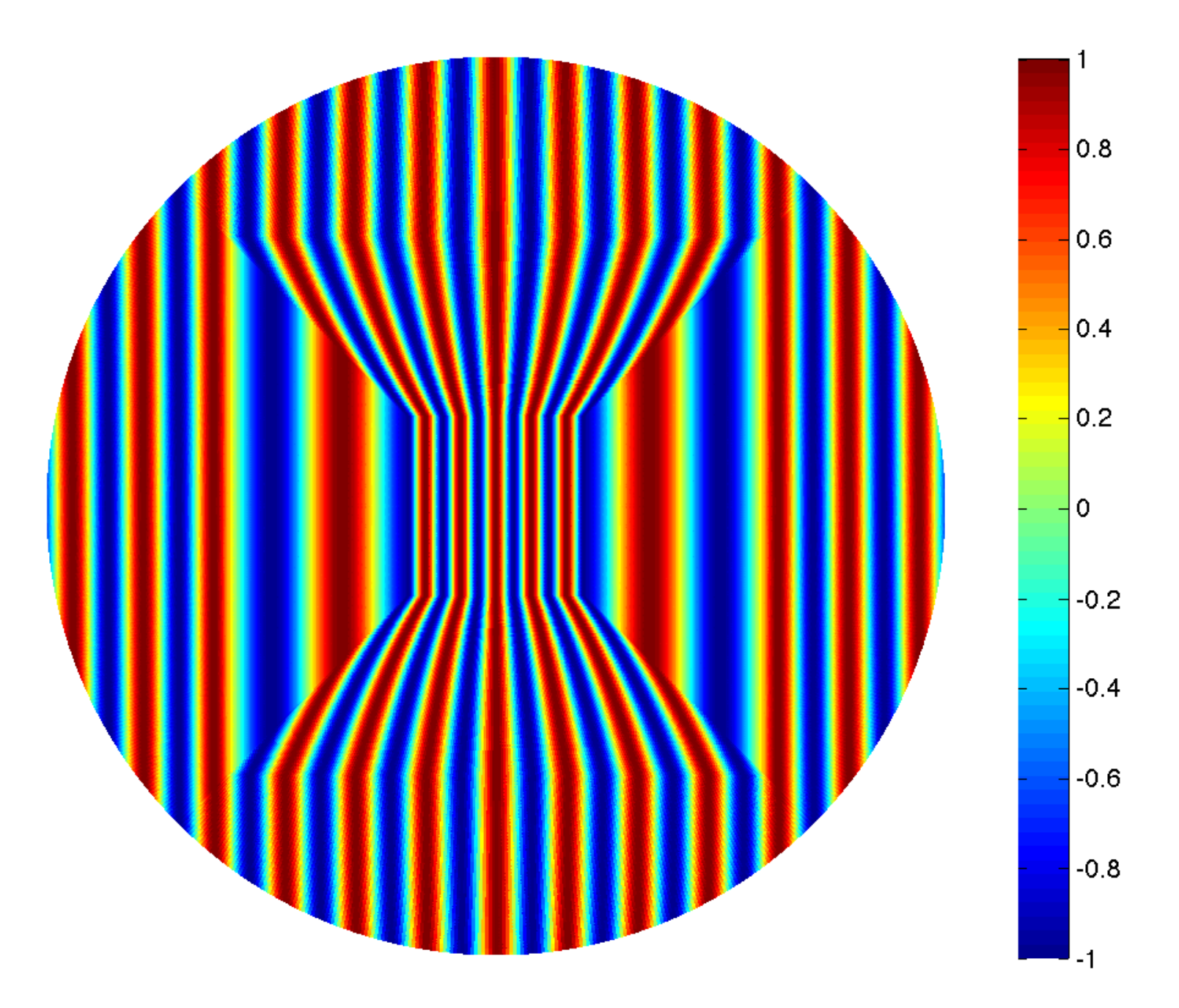}} \hspace*{-4pt}
\subfigure[Poynting vector]{\includegraphics[scale=0.25]{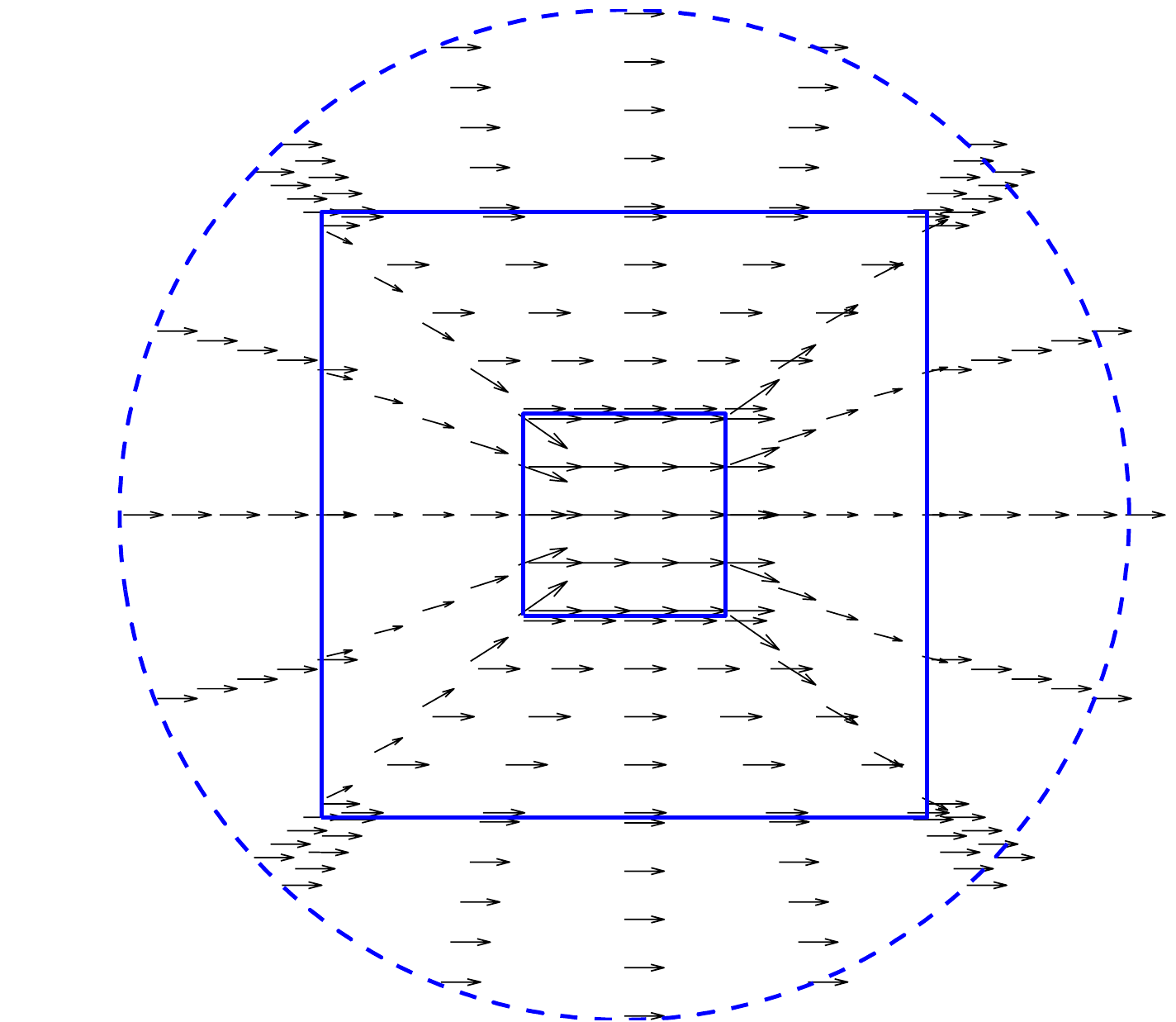}} \hspace*{-4pt} 
 \subfigure[Electric field with $\theta_0=\pi/4$]{ \includegraphics[scale=0.24]{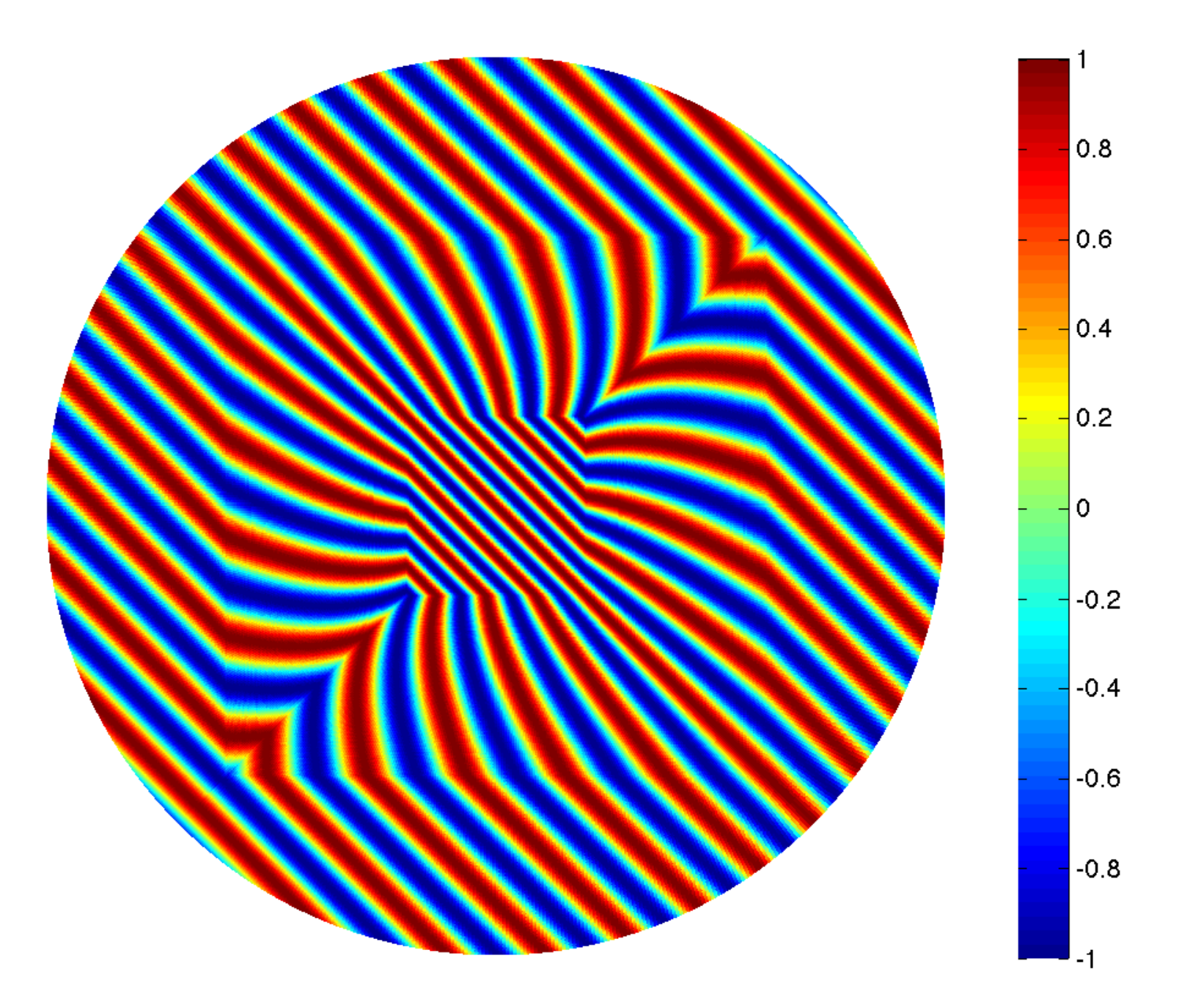}} \hspace*{-4pt}
 \subfigure[Poynting vector]{ \includegraphics[scale=0.25]{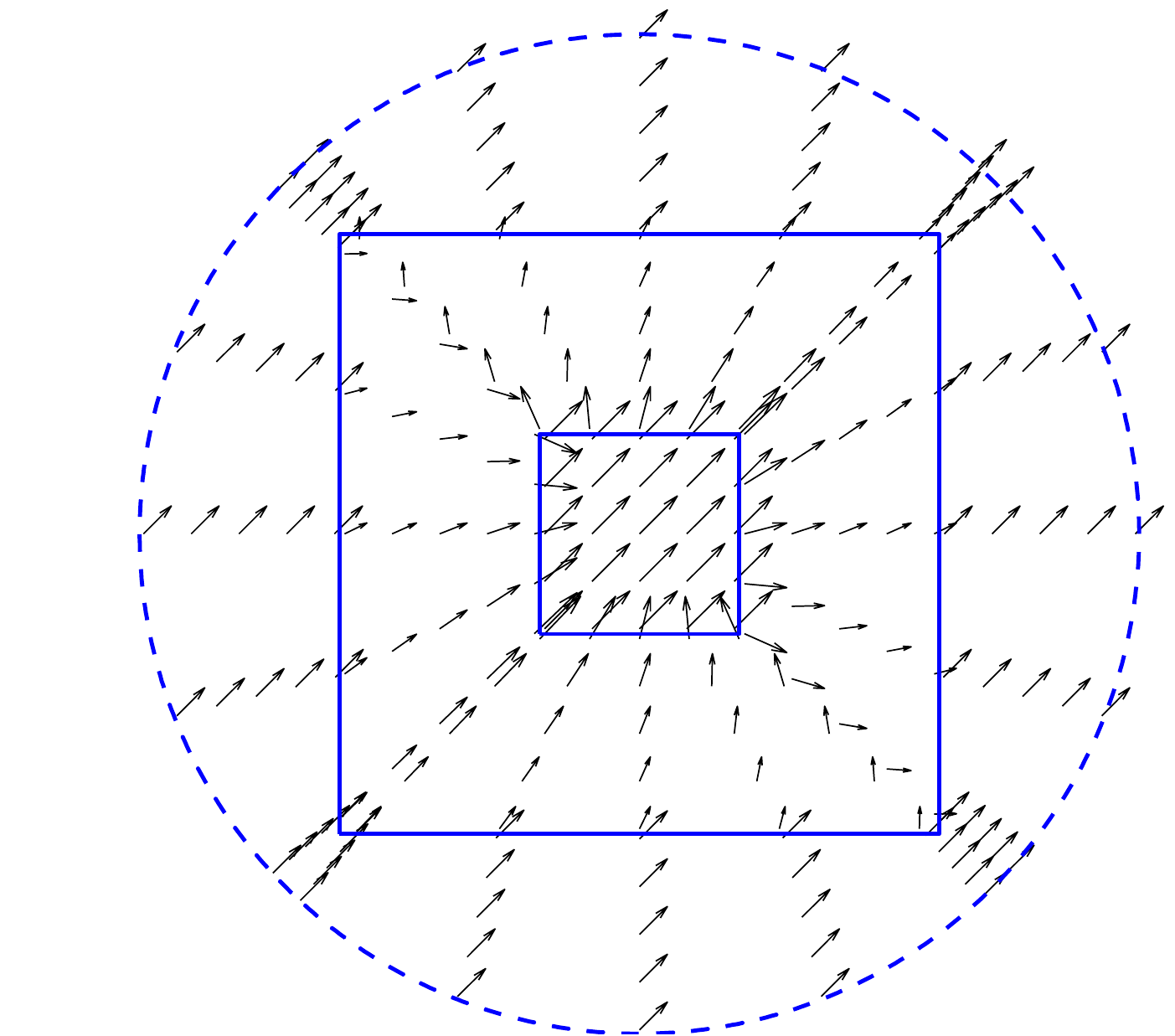}}
   \caption{\small The real part of the electric field distributions and associated Poynting vectors for square concentrators with (a)-(b): $\theta_0=0$ and (c)-(d): $\theta_0=\pi/4$, respectively.   }
\label{SemD}
\end{figure}

\subsection{Circular rotators} In contrast with the invisibility cloak and concentrator, the electromagnetic rotator is based upon a coordinate transformation of the angular variable rather  than the radial variable in \eqref{polartranA}. 

As illustrated in Figure \ref{concenfig2} (c), the domain $\Omega_{-}$  is a disk of radius $r=b,$ 
which encloses a concentric disk of radius $r=a<b.$   The waves are expected to rotate with a fixed angle $\theta_1$ in the interior disk.  This can be realised by the coordinate transformation (cf. \cite{chen2007rotator}):  
\begin{equation}\label{rotat-tran}
\begin{cases}
r=\breve r,\;\;\;  \theta= \breve \theta+\theta_1, \quad   &0<\breve r <a,\;\;\breve \theta,\,\theta \in [0,2\pi),\\[8pt]
r=\breve r,\;\;\;   \theta=\breve \theta +   \dfrac{s(b)-s(\breve r)}{{s}(b)-s(a)} \theta_1,
\quad &  a < \breve r<b,\;\;\breve \theta,\,\theta \in [0,2\pi),
\end{cases} 
\end{equation}
where $s$ is any smooth function such that $s(b)\not =s(a).$ As before, the transformation  is  
identity exterior to $\Omega_-.$

%


%
%
Define 
\begin{equation}\label{consta}
\kappa=  \dfrac{s(r)'}{{s}(b)-s(a)}\, \theta_1.
\end{equation}
Working out the material parameters as before (see Appendix \ref{AppendixA}), we have 
\begin{equation}\label{cr2}
{\bs C}=\frac{1}{r^2}
\begin{bmatrix}
r^2+2\kappa xy+\kappa^2y^2&-\kappa x^2-\kappa^2xy+\kappa y^2\;\\[3pt]
-\kappa x^2-\kappa^2xy+\kappa y^2 & r^2-2\kappa xy+\kappa^2x^2\;
\end{bmatrix},\;\;\;  n=1\quad {\rm for}\;\; a<r<b, 
\end{equation}
and 
\begin{equation}\label{cr1}
 {\bs C}={\bs I}_2,\quad  n=1 \quad {\rm for}\;\; 0<r<a,\;\; b<r<R.
\end{equation}
In Figure  \ref{concenfig2} (d), we illustrate a partition of the computational domain $B_R.$ Together with the standard transmission conditions in \eqref{Usuperscript}, we can implement the spectral element scheme as the previous cases with a similar partition 
of the computational domain (see Figure \ref{concenfig2} (d)).   In the computation, we  set  $a=0.3,\; b=0.7$ and $s(r)=r$ in \eqref{rotat-tran} and choose  $R=1.0$, $M=60$ and $N=40$ in each element. In Figure \ref{SemE}, we fix $k=40$, and plot the electric field distribution (real part) and the corresponding time averaged Poynting vector with the same incident angle $\theta_0=0$ and different rotation angles ((a)-(b): $\theta_1=\pi/4$, (c)-(d): $\theta_1=3\pi/4$). We find that the electric field distribution rotates its direction by $\pi/4$ in Figure \ref{SemE} (a) and the power flux (b) flows with the same direction in the closed region $r<a$. It can be observed in Figure \ref{SemE} (c)-(d) that even for a very sharp rotation angle $\theta_1=3\pi/4$, the field rotates exactly by  $3\pi/4$ angle without introducing any scattering wave outside.


\begin{figure}[htbp]
 {~}\hspace*{-16pt}\subfigure[Electric field with $\theta_1=\pi/4$]{ \includegraphics[scale=0.24]{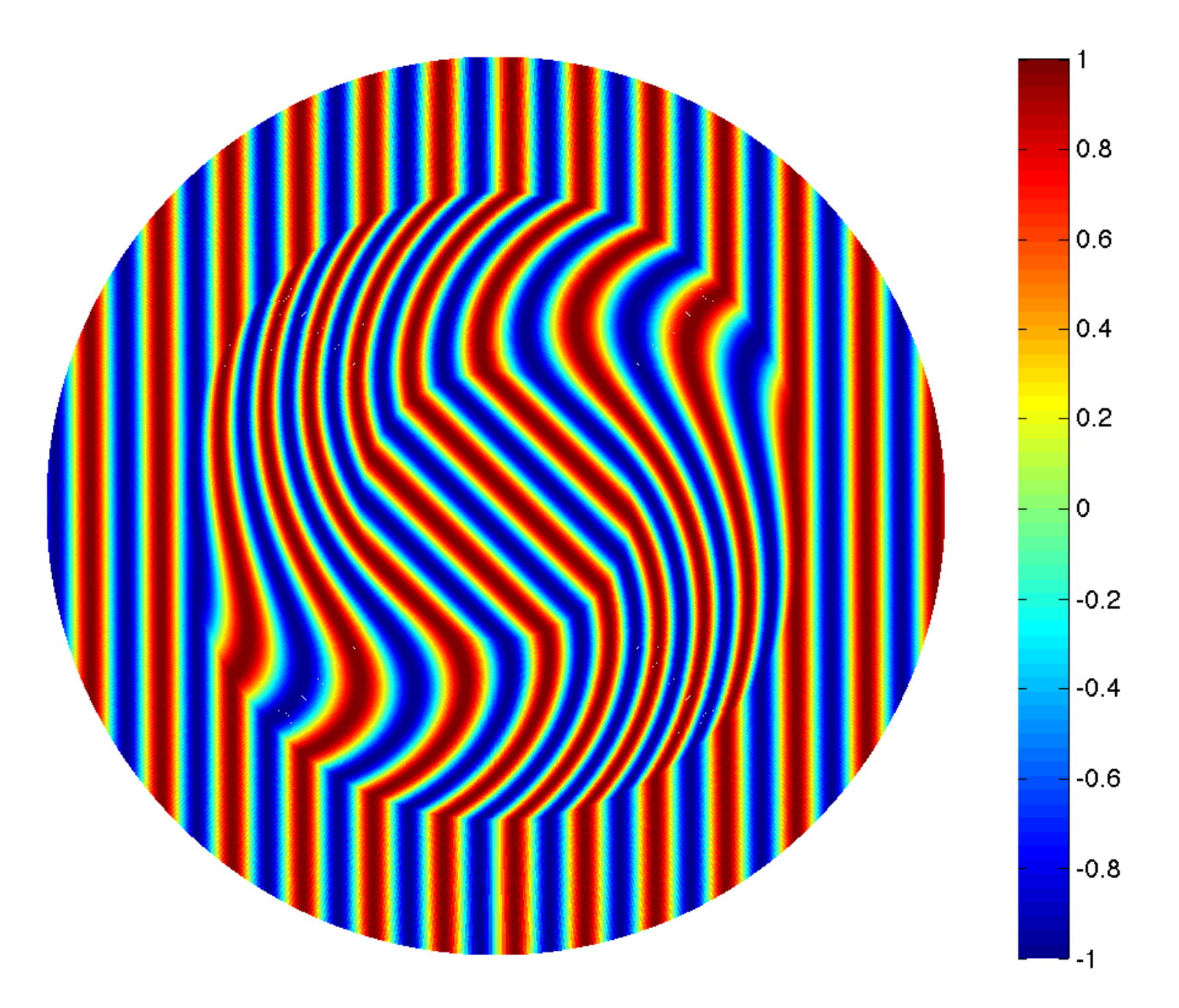}} \hspace*{-4pt}
\subfigure[Poynting vector]{\includegraphics[scale=0.25]{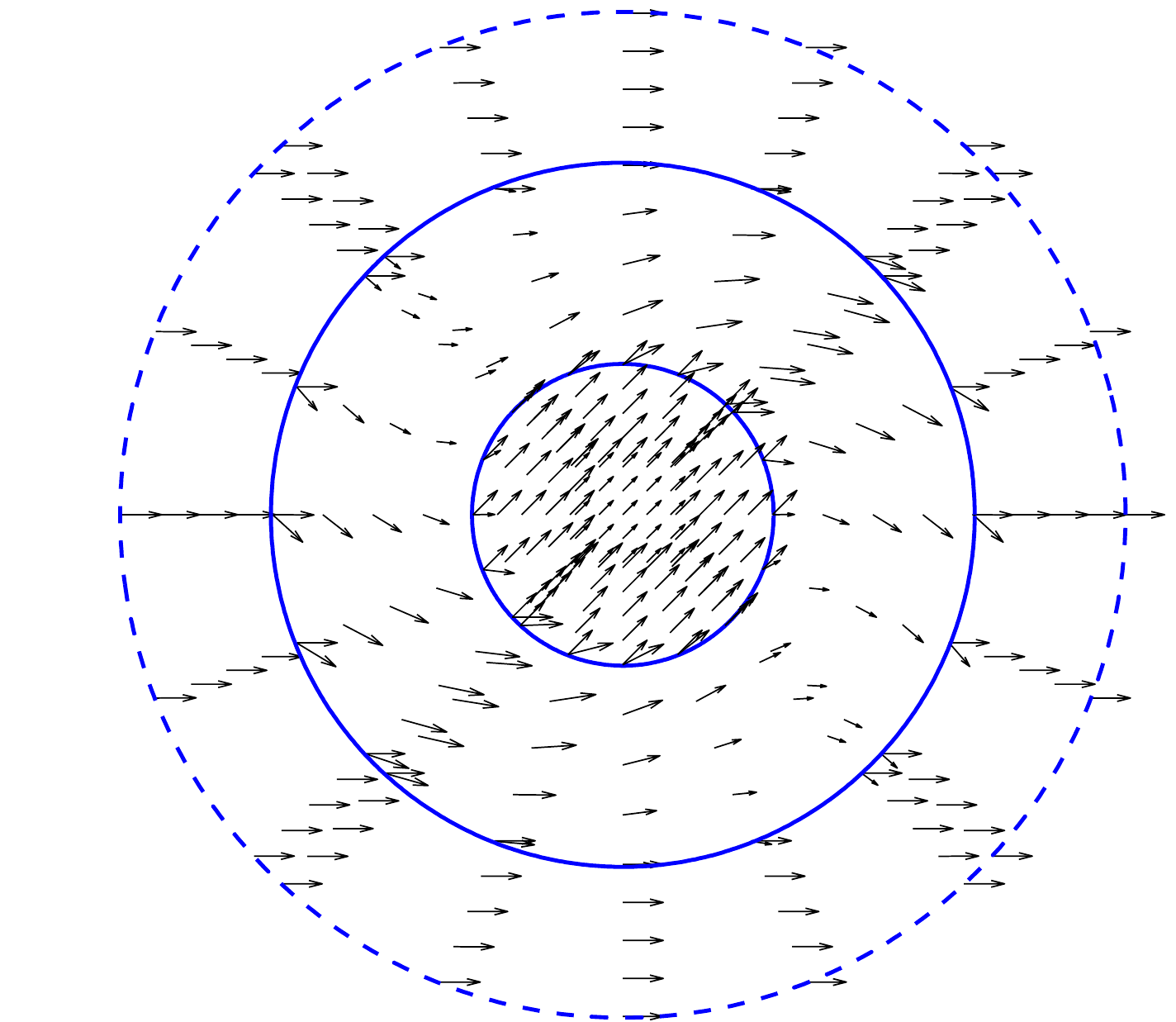}} \hspace*{-4pt} 
 \subfigure[Electric field with $\theta_1=3\pi/4$]{ \includegraphics[scale=0.24]{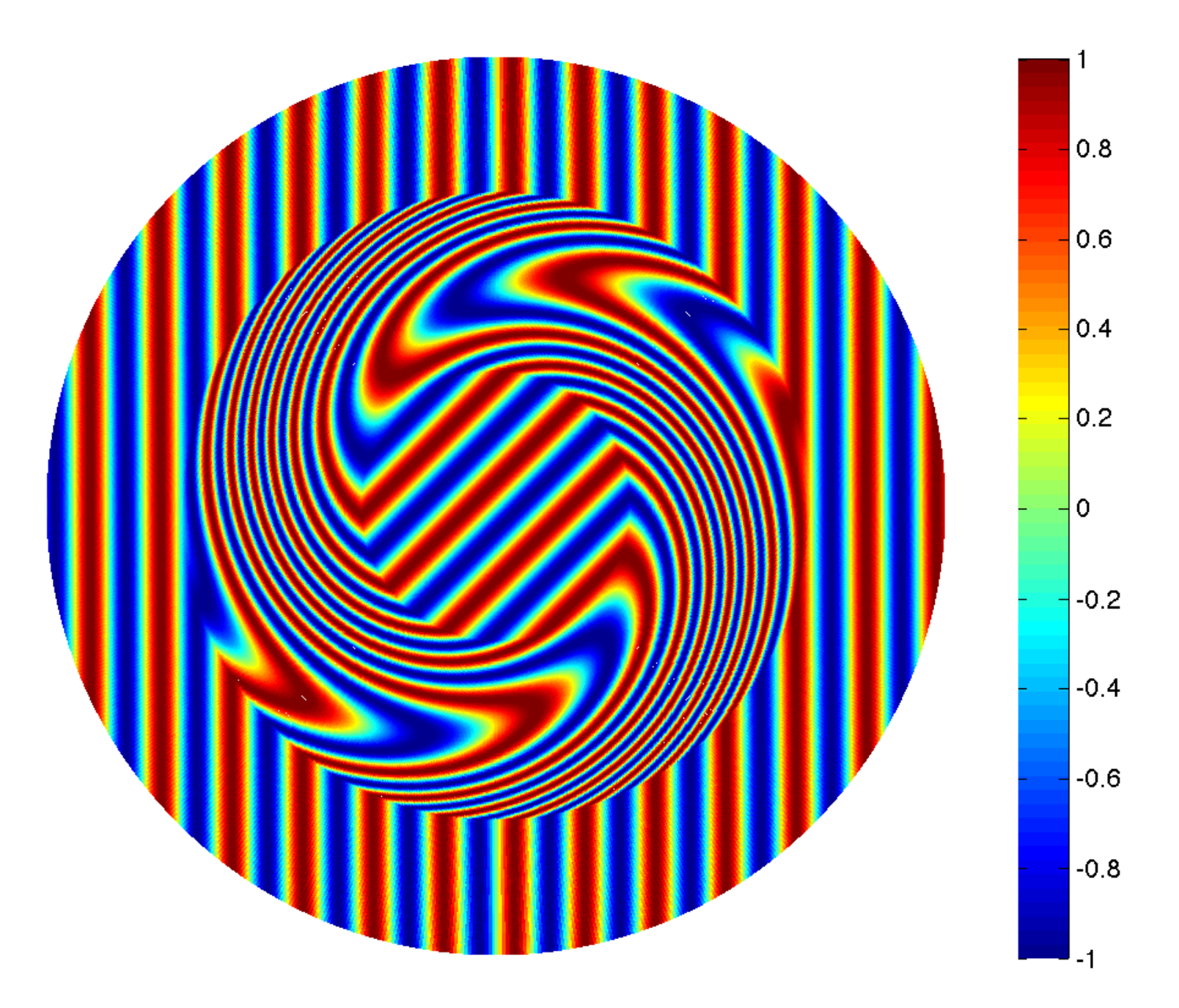}} \hspace*{-4pt}
 \subfigure[Poynting vector]{ \includegraphics[scale=0.25]{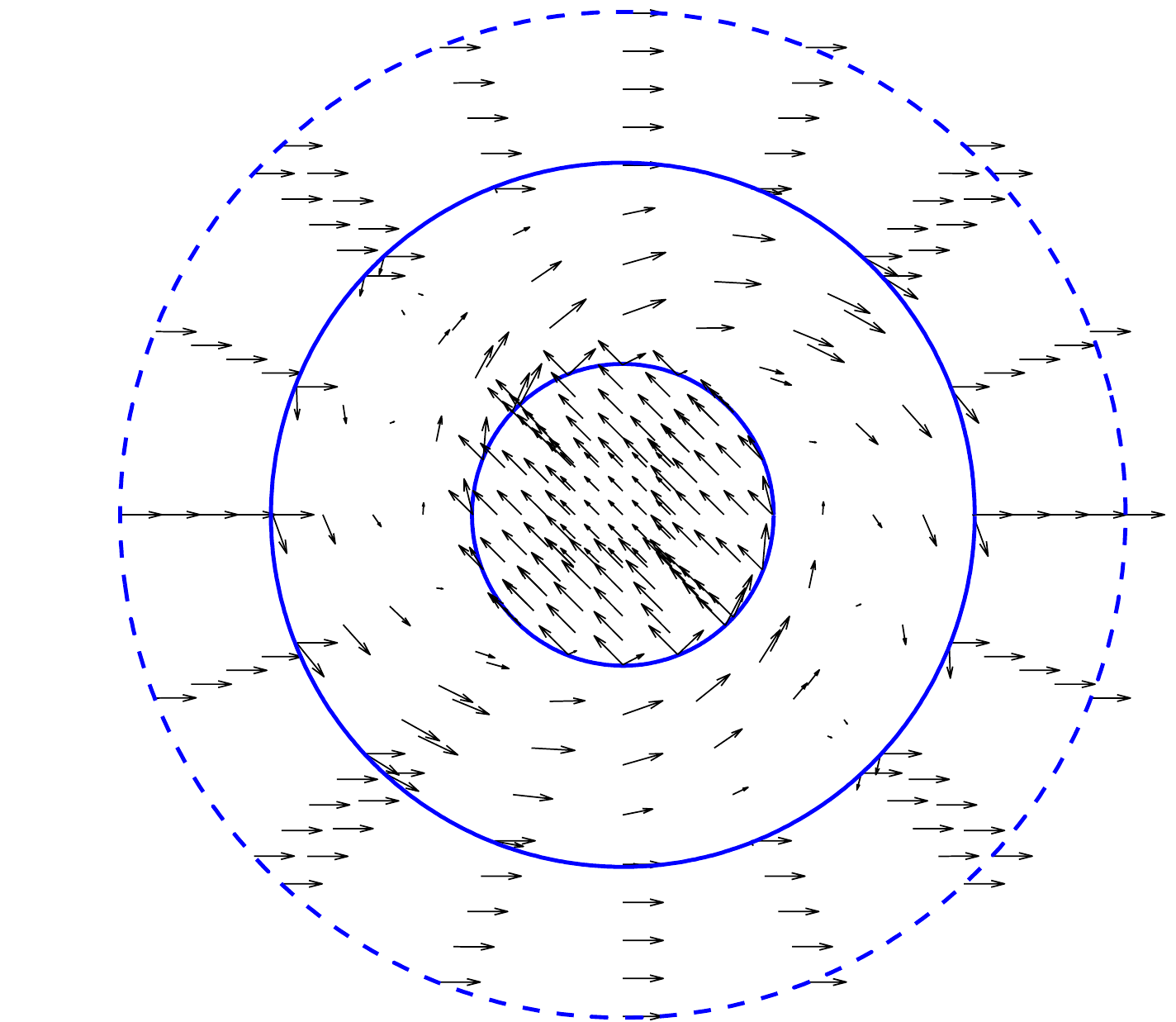}}
   \caption{\small The real part of the electric field distributions and associated Poynting vectors for circular rotators with (a)-(b): $\theta_1=\pi/4$ and (c)-(d): $\theta_1=3\pi/4$, respectively. }
\label{SemE}
\end{figure}

\vskip 20pt
\noindent\underline{\large\bf Concluding remarks}
\vskip 10pt

In this paper, we presented an accurate and efficient spectral-element solver for time-harmonic Helmholtz equations in general inhomogeneous and anisotropic media. We focused on several applications arisen from transformation electromagnetics which included the polygonal invisibility cloaks and concentrators, and  circular rotators.  We introduced new ideas of  how to seamlessly  integrate local elements and global DtN boundary condition.  This can shed light on the three-dimensional simulation where 
the DtN  BC involves spherical harmonic expansions.  We proposed new cloaking boundary conditions for accurate simulation of perfect polygonal invisibility cloaks.  The proposed method also provided a reliable alternative to the analytic tools to study the interesting phenomena when defects and other media were embedded or placed in a perfect cloak.

\vskip 10pt

\appendix
\renewcommand{\thesection}{\Alph{section}}
\renewcommand{\theequation}{\thesection.\arabic{equation}}
\section{Derivation of material parameters for polygonal cylindrical  cloaks, concentrators and circular rotators}\label{AppendixA}
As  illustrated by \eqref{polartranA} in Section \ref{sect:pcloak},  the transformation is given from polar coordinates  $(\breve r,\breve \theta)$ (of the original Cartesian coordinates $(\breve x,\breve y)$) to polar coordinates  $(r,\theta)$ (of the physical  Cartesian coordinates $(x,y)$), so by the  chain rule,  the Jacobian matrix $ \bs J_{\! {\rm cn}}$ can be computed by 
\begin{equation}\label{Jcnexp}
 \bs J_{\! {\rm cn}}=
\begin{bmatrix}
\partial_{\breve x} x & \partial_{\breve y} x  \\[1pt]
\partial_{\breve x} y & \partial_{\breve y} y  \\[1pt]
\end{bmatrix}
=
\begin{bmatrix}
\partial_{r} x & \partial_{\theta} x  \\[1pt]
\partial_{r} y & \partial_{\theta} y  \\[1pt]
\end{bmatrix}
\begin{bmatrix}
\partial_{\breve r} r & \partial_{\breve \theta} r  \\[1pt]
\partial_{\breve r} \theta & \partial_{\breve \theta} \theta  \\[1pt]
\end{bmatrix}
\begin{bmatrix}
\partial_{\breve x} \breve r & \partial_{\breve y} \breve r  \\[0.5 pt]
\partial_{\breve x} \breve \theta & \partial_{\breve y} \breve \theta  \\[0.5pt]
\end{bmatrix}.
\end{equation}
It is clear that 
\begin{equation}\label{r2x}
\begin{bmatrix}
\partial_{r} x & \partial_{\theta} x  \\[1pt]
\partial_{r} y & \partial_{\theta} y  \\[1pt]
\end{bmatrix}
=
\begin{bmatrix}
\cos \theta & -r\sin \theta  \\[1pt]
\sin \theta & r \cos \theta  \\[1pt]
\end{bmatrix},
\quad 
{\rm det}\bigg(   
\begin{bmatrix}
\partial_{r} x & \partial_{\theta} x  \\[1pt]
\partial_{r} y & \partial_{\theta} y  \\[1pt]
\end{bmatrix}
\bigg)
=r,
\end{equation}
and
\begin{equation}\label{x2r}
\begin{bmatrix}
\partial_{\breve x} \breve r & \partial_{\breve y} \breve r  \\[0.5 pt]
\partial_{\breve x} \breve \theta & \partial_{\breve y} \breve \theta  \\[0.5pt]
\end{bmatrix}
=
\begin{bmatrix}
\cos \breve \theta & \sin \breve \theta  \\[0.1pt]
-\sin \breve \theta / \breve r & \cos \breve \theta / \breve r  \\[0.1pt]
\end{bmatrix},
\quad 
{\rm det}\bigg(  
\begin{bmatrix}
\partial_{\breve x} \breve r & \partial_{\breve y} \breve r  \\[0.5 pt]
\partial_{\breve x} \breve \theta & \partial_{\breve y} \breve \theta  \\[0.5pt]
\end{bmatrix} 
\bigg)
=\frac{1}{\breve r}.
\end{equation}
Denote 
\begin{equation}
{\bs { \tilde  J}}:=
\begin{bmatrix}
\partial_{\breve r} r & \partial_{\breve \theta} r  \\[1pt]
\partial_{\breve r} \theta & \partial_{\breve \theta} \theta  \\[1pt]
\end{bmatrix},
\end{equation}
which is determined by specific  coordinate transformations. With a direct calculation,  we derive from  \eqref{parameter3}  the material parameters $\bs C$ and $n$ for the  general transformation in \eqref{polartranA}:
\begin{align}
C_{11}&=\frac{\breve r}{r^3 {\rm det}(\bs { \tilde  J})} \Big( \big(x\partial_{\breve r} r-ry\partial_{\breve r} \theta \big)^2+\frac{1}{\breve r^2}\big(x\partial_{\breve r} \theta -ry\partial_{\breve \theta} \theta \big)^2    \Big)        \,, \label{ck11g}\\[6pt]
C_{22}&=\frac{\breve r}{r^3 {\rm det}(\bs { \tilde  J})} \Big( \big(y\partial_{\breve r} r+rx\partial_{\breve r} \theta \big)^2+\frac{1}{\breve r^2}\big(y\partial_{\breve r} r +rx\partial_{\breve \theta} \theta \big)^2    \Big) \,, \label{ck22g}\\[6pt]
C_{12}&=\frac{\breve r}{r^3 {\rm det}(\bs { \tilde  J})} \Big( xy\Big(\partial^2_{\breve r} r+\frac{1}{\breve r^2}\partial^2_{\breve \theta} r -r^2\partial^2_{\breve r} \theta-\frac{r^2}{\breve r^2}  \partial^2_{\breve \theta} \theta  \Big)+\frac{(x^2-y^2)r}{\breve r^2}\big(\breve r^2  \partial_{\breve r} r \partial_{\breve r} \theta  + \partial_{\breve \theta} r    \partial_{\breve \theta} \theta\big)
\Big),\label{ck12g}
\end{align}
and 
\begin{equation}\label{ng}
n=\frac 1{{\rm det}(\bs J_{\! {\rm cn}})}= \frac{\breve r}{r {\rm det}(\bs {\tilde J})}\,.  
\end{equation}

\vskip 6pt 

For the polygonal cylindrical cloak, we derive from \eqref{cloak-tran} that in $\Omega_-^{a}$,  
\begin{equation}\label{br2r}
\breve r=\dfrac{1}{1-\rho} (r-  \, R_1),\quad \breve \theta =\theta,
\end{equation}
which leads to  
\begin{equation}\label{cloakJcn}
\bs { \tilde  J}=
\begin{bmatrix}
1-\rho & \partial_{ \theta } R_1   \\[1pt]
0 & 1  \\[1pt]
\end{bmatrix},\quad
\rm {det}(\bs { \tilde  J})=1-\rho.
\end{equation}
Recall that  $\partial_{ \theta } R_1 $ can be worked out  by  \eqref{hR1}. 
Thus, the material parameters $\bs C$ and $n$ in  \eqref{ck11}-\eqref{nkcase2} can be obtained by substituting \eqref{br2r}-\eqref{cloakJcn} into \eqref{ck11g}-\eqref{ng}.

 Similarly, for the polygonal cylindrical concentrator, we obtain from  the transformation  \eqref{concen-tran}   that in $\Omega_-^{\rm a}$, 
\begin{equation}\label{br2rconcen}
\breve r=\frac{1}{1-\varrho}(r-\varrho R_2),\quad \breve \theta=\theta,
\end{equation}
and
\begin{equation}\label{cloakJcnconcen}
\bs { \tilde  J}=
\begin{bmatrix}
1-\varrho & \varrho \partial_{  \theta } R_2 \\[1pt]
0 & 1  \\[1pt]
\end{bmatrix},\quad
\rm {det}(\bs { \tilde  J})=1-\varrho \,.
\end{equation}
Inserting \eqref{br2rconcen} and \eqref{cloakJcnconcen} into \eqref{ck11g}-\eqref{ng}, we obtain $\bs C$ and $n$ in \eqref{cn2a}-\eqref{ncase2}.

For the circular rotators, we  have  from \eqref{rotat-tran} that    
\begin{equation}\label{br2rrotator}
\breve r= r,\;\;\;   \breve \theta= \theta -   \dfrac{s(b)-s( r)}{{s}(b)-s(a)} \theta_1,
\end{equation}
and 
\begin{equation}\label{Jcnrotator}
\begin{bmatrix}
\partial_{\breve r} r & \partial_{\breve \theta} r  \\[1pt]
\partial_{\breve r} \theta & \partial_{\breve \theta} \theta  \\[1pt]
\end{bmatrix}
=
\begin{bmatrix}
1& 0 \\[1pt]
-\kappa & 1  \\[1pt]
\end{bmatrix},
\end{equation}
where $\kappa$ is defined in \eqref{consta}. 
Then we can compute the material parameters \eqref{consta} in a similar fashion.

\end{document}